%% file: main.tex
\documentclass[acmsmall]{acmart}

\setcopyright{cc}
\setcctype{by}
\acmJournal{PACMMOD}
\acmYear{2026}
\acmVolume{4}
\acmNumber{6 (SIGMOD)}
\acmArticle{330}             
\acmMonth{12}
\acmDOI{10.1145/3856325}

\usepackage{algorithm}
\usepackage{algpseudocode}
\makeatletter
\def\fs@ruled{\def\@fs@cfont{\bfseries}\let\@fs@capt\floatc@ruled
  \def\@fs@pre{{\normalcolor\hrule height.8pt depth0pt}\kern2pt}%
  \def\@fs@post{\kern2pt{\normalcolor\hrule}\relax}%
  \def\@fs@mid{\kern2pt{\normalcolor\hrule}\kern2pt}%
  \let\@fs@iftopcapt\iftrue}
\makeatother
\usepackage{multirow}
\usepackage{subcaption}
\usepackage{booktabs}
\usepackage{enumitem}
\usepackage[framemethod=TikZ]{mdframed}
\usepackage{placeins}
\usepackage{wrapfig}

\newcommand{\bbb}{\noindent\textbf}

\newif\ifshowappendix
\showappendixtrue
\newcommand{\extversion}{the extended version of this paper~\cite{crane-extended}} 
\newcommand{\appref}[1]{\ifshowappendix Appendix~\ref{#1}\else\extversion\fi}
\newcommand{\algparen}{\ifshowappendix\ (Algorithm~\ref{alg:crane-ops})\fi}

\begin{document}

\title{Crane: An Accurate and Scalable Neural Sketch for Graph Stream Summarization}
\ifshowappendix\titlenote{Extended version, with the full appendix, of the article published in Proc.\ ACM Manag.\ Data 4, 6 (SIGMOD 2027), Article 330.}\fi


\author{Boyan Wang}
\authornote{Boyan Wang and Zhuochen Fan contributed equally to this work and share first authorship. This work was conducted by Boyan Wang during his internship at Pengcheng Laboratory under the supervision of Zhuochen Fan.}
\affiliation{%
  \institution{Pengcheng Laboratory}
  \country{China}}
\affiliation{%
  \institution{Hefei University of Technology}
  \country{China}}
\email{wangboyan@mail.hfut.edu.cn}

\author{Zhuochen Fan}
\affiliation{%
  \institution{Pengcheng Laboratory}
  \city{Shenzhen}
  \country{China}}
\email{fanzc@pku.edu.cn}
\ifshowappendix\authornotemark[2]\else\authornotemark[1]\fi
\authornote{Corresponding Author: Zhuochen Fan.}

\author{Dayu Wang}
\affiliation{%
  \institution{Nanyang Technological University}
  \city{Singapore}
  \country{Singapore}}
\email{dayu001@e.ntu.edu.sg}

\author{Fangcheng Fu}
\affiliation{%
  \institution{Shanghai Jiao Tong University}
  \department{School of Artificial Intelligence}
  \city{Shanghai}
  \country{China}}
\email{ccchengff@sjtu.edu.cn}

\author{Jiarui Guo}
\affiliation{%
  \institution{Peking University}
  \department{School of Computer Science}
  \city{Beijing}
  \country{China}}
\email{ntguojiarui@pku.edu.cn}

\author{Zeyu Luan}
\affiliation{%
  \institution{Pengcheng Laboratory}
  \city{Shenzhen}
  \country{China}}
\email{luanzy@pcl.ac.cn}

\author{Lei Zou}
\affiliation{%
  \institution{Peking University}
  \city{Beijing}
  \country{China}}
\email{zoulei@pku.edu.cn}

\author{Qing Li}
\affiliation{%
  \institution{Pengcheng Laboratory}
  \city{Shenzhen}
  \country{China}}
\email{liq@pcl.ac.cn}

\author{Tong Yang}
\affiliation{%
  \institution{Peking University}
  \department{School of Computer Science}
  \city{Beijing}
  \country{China}}
\email{yangtong@pku.edu.cn}

\begin{abstract}
Graph streams are rapidly evolving sequences of edges that convey continuously changing relationships among entities, playing a crucial role in domains such as networking, finance, and cybersecurity.
Their massive scale and high dynamism make obtaining accurate statistics challenging with limited memory constraints.
Traditional methods summarize graph streams through hand-crafted sketches, while recent studies have begun to replace these sketches with neural counterparts to improve adaptability and accuracy.
However, this shift faces a major challenge: under limited memory, dominant frequent items tend to overshadow rare ones, hindering the neural network's ability to recover accurate statistics.
To address this, we propose \textit{Crane}, a hierarchical neural sketch architecture for graph stream summarization.
Crane uses a hierarchical carry mechanism that automatically elevates frequent items to higher memory layers, reducing interference between frequent and infrequent items within the same layer.
To better accommodate real-world deployment, Crane further adopts an adaptive memory expansion strategy that dynamically adds new layers once the occupancy of the top layer exceeds a threshold, enabling scalability across diverse data magnitudes.
Experimental results show that Crane significantly reduces estimation error compared to state-of-the-art (SOTA) methods, while also offering competitive throughput.
\end{abstract}



\received{17 April 2026}
\received[revised]{19 August 2026}
\received[accepted]{12 September 2026}

\maketitle
\renewcommand{\shortauthors}{Boyan Wang, Zhuochen Fan, Dayu Wang, et al.}

\input{MainText/1-Intro}
\input{MainText/2-Related_Work}

\input{MainText/3-Definition}
\input{MainText/4-Design}

\input{MainText/5-Maths}
\input{MainText/6-Evaluation}

\input{MainText/7-Conclusion}

\begin{acks}
We sincerely thank the anonymous reviewers for their constructive comments.
This work was supported in part by the National Natural Science Foundation of China (NSFC) under Grant No. 62402012, in part by the Basic and Frontier Research Project of Pengcheng Laboratory (PCL) under Grant No. 2025QYB004, and in part by the Major Key Project of PCL under grant No. PCL2026A01.
\end{acks}

\bibliographystyle{ACM-Reference-Format}
\bibliography{sample}

\ifshowappendix
\clearpage
\appendix
\input{MainText/8-Appendix}
\FloatBarrier 
\fi

\end{document}

%% file: MainText/1-Intro.tex
\vspace{-0.1in}
\section{Introduction}
\label{sec:intro}

\subsection{Background and Motivation}

Graph streams, as time-ordered sequences of edge updates, encapsulate the continuously evolving relationships among entities in dynamic systems. 
Formally, a graph stream $\mathcal{S}$ is defined as a sequence $\{e_1, e_2, \dots, e_t, \dots\}$, where each item $e_t$ represents a directed edge from an origin node $o_t$ to a destination node $d_t$, augmented with a varying weight $w_t$ at timestamp $t$.
This model accommodates repeated edges with varying weights over time, enabling the accumulation of edge strengths and the construction of a dynamically extending graph $G(V, E)$ that reflects ongoing interactions among entities. 
Such representations are particularly suited for high-velocity data environments and have found widespread applications in big data domains \cite{wei2018topppr,ding2019efficient,liu2020truss,pacaci2020regular}, including network security \cite{eswaran2018spotlight,wu2019enhanced,bhatia2021mstream,bhatia2023sketch}, financial transaction analysis \cite{khodabandehlou2024fifraud,zhang2024survey,qi2025ldbc,han2025tempasd}, social network analysis \cite{zhou2014event,li2015real,kim2020densely,pandey2021terrace}, and beyond, where scalability and adaptability to streaming inputs are paramount.

Recently, there has been a notable shift in focus from traditional handcrafted data structures~\cite{zhao2011gsketch,tang2016tcm,gou2019fast,ko2020incremental,gou2022gss,chen2022horae,chen2022scube,jiang2023auxo,HourglassSketch} to learnable neural alternatives in the domain of graph stream summarization.
These neural structures seek to harness flexible, neural-driven compression mechanisms to attain superior compression ratios.
We summarize their main advantages as follows: 
first, the adaptive allocation of greater memory resources to salient information; 
and second, the effective utilization of global contextual insights to dynamically reconfigure storage layouts in a more data-conscious manner.

However, applying neural networks to graph stream summarization remains a challenging endeavor.
The fundamental challenge lies in effectively compressing the massive, unbounded information of high-velocity graph streams into strictly constrained memory spaces while preserving the fidelity of edge weights.

\begin{wrapfigure}{r}{0.50\linewidth}
    \vspace{-0.35cm}
    \centering
    \includegraphics[width=\linewidth]{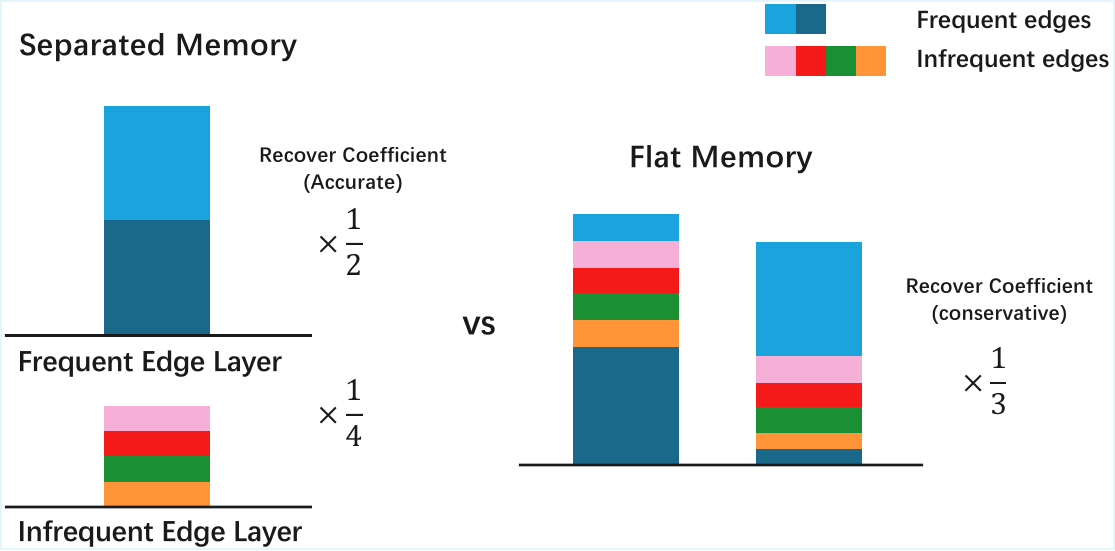}
    \caption{Overview of the motivating example.}
    \label{fig:intro-motivation}
\end{wrapfigure}

To date, the pioneering work, Mayfly~\cite{feng2023mayfly}, has demonstrated the potential of neural networks to surpass traditional handcrafted sketches. 
It introduces a memory-augmented network with bounded gradients, leveraging a pre-training followed by fine-tuning paradigm to achieve accuracy exceeding that of handcrafted sketches across datasets of varying scales. 
Mayfly further incorporates hash-based sharding to enable manual expansion of memory capacity at the onset of a task. 
Additionally, several designs for neural networks in general data streams offer valuable insights.
For instance, Learned Count-Min~\cite{hsu2019learning} trains a neural network as a heavy hitter oracle on a small subset of the stream to separate high-frequency items into dedicated buckets, reducing collisions in the remaining sketch, while LegoSketch~\cite{feng2025lego} introduces a modular architecture of memory bricks for flexible adaptation to dynamic budgets.

Nevertheless, these approaches encounter critical limitations.
\textbf{First and foremost}, Mayfly lacks mechanisms to differentiate and isolate edges based on frequency, a necessity highlighted in traditional sketch literature~\cite{yang2018elastic,zhou2018cold}.
Real-world graph streams typically exhibit highly skewed distributions, where a small number of frequent edges coexist with a vast number of infrequent edges.
When collisions occur among edges of similar magnitudes, neural networks can effectively recover the true frequencies (as shown in Figure~\ref{fig:intro-motivation}(left), where the separated coefficients allow the frequency to be easily recovered).
However, when frequent edges and infrequent edges are forced into the same memory cell, the network loses the critical context of this size disparity.
Unable to differentiate the distinct signals within the cell, the network defaults to a compromised, conservative estimate (as shown in Figure~\ref{fig:intro-motivation}(right)).
As a result, this inevitably leads to a severe overestimation of infrequent streams and a simultaneous underestimation of frequent ones. 
Furthermore, while previous methods like Learned Count-Min attempt to mitigate this by using a heavy-hitter oracle to separate streams, they rely on pre-training over a specific data subset. 
This approach essentially memorizes historical heavy hitters, making the model highly brittle and restricted in dynamic, real-world scenarios where stream distributions evolve rapidly or prior routing information is unavailable.
\textbf{Second}, the reliance on hash-based sharding for memory management necessitates manual configuration at task initiation.
This static allocation strategy is ill-suited for real-world scenarios where stream volumes escalate continuously, preventing the system from adapting dynamically to growing storage demands.

\vspace{-0.05in}
\subsection{Our Proposed Solution}

To address these challenges, we introduce \textbf{Crane}, a hierarchical neural architecture designed for robust graph stream summarization.
Unlike flat structures that force all edges into a single space, Crane organizes memory into a stack of layers and employs a \textbf{hierarchical carry mechanism} to physically segregate edges based on their magnitude.
This design effectively mitigates the interference problem: by automatically lifting frequent edges to higher layers, Crane ensures they do not monopolize the limited capacity reserved for infrequent edges in lower layers.

Functionally, Crane operates through a streamlined workflow consisting of \textbf{Store} and \textbf{Query} phases.
During the Store phase, a lightweight \textbf{encoding network} first maps each incoming edge to a two-dimensional basis matrix.
Instead of static accumulation, the insertion follows a bottom-up trajectory governed by our proposed \textbf{hierarchical carry mechanism}.
Specifically, the system performs an on-the-fly estimation at the current layer; if the estimated frequency exceeds a manually specified threshold, the excess mass is ``carried'' (promoted) to the subsequent layer.
This process dynamically disentangles the traffic: recurring patterns migrate to upper layers like arithmetic carries, while sporadic interactions remain at the bottom.
During the \textbf{Query} phase, Crane retrieves partial estimates from all allocated layers and aggregates them via a learnable linear decoder to produce the final high-fidelity frequency.

\begin{sloppypar}
Furthermore, this hierarchical workflow inherently enables \textbf{smooth automatic memory expansion}.
In contrast to hash-based sharding which requires pre-allocating fixed buckets, Crane monitors the saturation of the top-most layer during the Store phase.
New upper layers are dynamically allocated only when the current hierarchy is saturated.
This allows the memory footprint to grow logarithmically with stream volume, optimizing resource utilization in unbounded stream scenarios.
Our extensive experiments show that Crane reduces the Average Relative Error (ARE) by up to $33.8\times$ in edge weight estimation.
Also, our case studies validate that Crane reduces ARE by more than three orders of magnitude in graph stream analytics.
Finally, Crane is seamlessly integrated into Apache Flink \cite{carbone2015apache} with acceptable performance.
The code related to Crane is open-sourced on GitHub \cite{Crane-source}.
\end{sloppypar}

Our contributions can be summarized as follows:
\begin{itemize}[leftmargin=1em]\vspace{-0.05in}
\item We propose Crane, a hierarchical neural architecture for graph stream summarization that automatically segregates frequent and infrequent edges via a hierarchical carry mechanism.
\item We further introduce a smooth automatic memory expansion mechanism to adapt to real-world scenarios with continuously expanding network traffic, while reducing the memory space complexity from linear to logarithmic.
\item The results show that Crane achieves competitive throughput and significantly reduces the ARE in all our evaluations compared to SOTA methods under identical memory constraints.

\end{itemize}

%% file: MainText/2-Related_Work.tex
\section{Related Work}

Graph stream summarization has attracted considerable research attention in recent years. 
We categorize existing methods into handcrafted structures and neural structures.

\vspace{-0.1in}
\subsection{Handcrafted Structures}

Handcrafted graph stream summarization methods primarily focus on achieving accurate estimates under relatively abundant memory conditions, subsequently extending their accuracy to scenarios with limited memory. 
Notable approaches in this domain include TCM~\cite{tang2016tcm}, GSS~\cite{gou2019fast,gou2022gss}, Auxo~\cite{jiang2023auxo}, and HourglassSketch \cite{HourglassSketch}. 
TCM~\cite{tang2016tcm} extends the idea of the Count-Min Sketch (CMS) \cite{cmsketch}\footnote{CMS is a probabilistic data structure for estimating item frequencies in one-dimensional data streams. 
It uses $d$ arrays, each with $w$ counters and an associated hash function.
On an item's arrival, it increments the hashed counter in each array by 1, and a query returns the minimum of these $d$ counters to minimize over-estimation.
Its variants include the Count Sketch \cite{csketch}, the Conservative Update Sketch \cite{cusketch}, etc.} by leveraging a pre-defined $m \times m$ compressed matrix $M$ to store graph stream edges, where it employs a hash function $h(\cdot)$ with range $[0, m)$ to map an edge to a corresponding bucket in $M$.
For an incoming edge $o_i \to d_i$, TCM adds its associated weight $w_i$ to the value in bucket $M[h(o_i), h(d_i)]$, merging nodes with the same hash value.
This supports boolean queries (e.g., edge existence) and aggregation queries (e.g., summing weights in the $h(x)$-th column for outgoing edges from $x$). 
However, TCM suffers from poor query accuracy due to hashing collisions.
To improve accuracy, GSS~\cite{gou2019fast,gou2022gss} incorporates a $f$-bit fingerprint $\xi_v$ (a hash value) for each node $v$ alongside the weight in the compressed matrix.
GSS comprises an $m \times m$ matrix $M$ and an adjacency list as an extra buffer.
Upon inserting an edge $e_i = (\left \langle o_i, d_i \right \rangle; w_i; t_i)$, GSS computes fingerprints $\left \langle \xi_{o_i}, \xi_{d_i}\right \rangle$ and stores them with $w_i$ in $M[h(o_i), h(d_i)]$.
If edges with different fingerprint pairs collide, the newcomer is stored in the buffer.
To minimize buffer size, GSS uses square hashing to generate multiple candidate buckets $\{M[h_k(o_i), h_j(d_i)] \mid 1 \leq k \leq r, 1 \leq j \leq r\}$ for each edge, recording $\left \langle \xi_{o_i}, \xi_{d_i}, w_i, idx_{pair}\right \rangle$ in an empty candidate, where $idx_{pair}$ denotes the indices.
The leftover edges go to the buffer.
For large-scale streams, the growing buffer incurs high memory and insertion costs.
Auxo~\cite{jiang2023auxo} addresses scalability issues in prior designs by proposing a prefix-embedded tree (PET) that extends building blocks in a tree-style manner, achieving logarithmic insert/query times and reducing storage by a scale $\log |E|$, where $|E|$ is the edge set size. 
PET embeds binary prefixes into the tree, omitting $i$-bit prefixes on level $i$. 
To enhance memory utilization, Auxo introduces a proportional PET that expands levels incrementally. 
HourglassSketch~\cite{HourglassSketch} separates large-weight edges from light-weight edges using a two-stage structure.
Stage~1 uses a CocoSketch \cite{zhang2021cocosketch} to accurately record large-weight edges; Stage 2 employs a combination of TCM and TowerSketch \cite{yang2023sketchint} to handle the remaining small-weight edges.
Counters in Stage~2 are frozen upon successful Stage~1 insertion, preventing overestimation from accumulating over time.
However, these methods lack explicit accuracy correction mechanisms, limiting adaptability under severe memory constraints.

\subsection{Neural Structures}

Neural graph stream summarization methods focus on recovering precise estimates under severe memory constraints using neural networks. 
The pioneering method Mayfly~\cite{feng2023mayfly} introduces the first neural data structure specifically for graph stream summarization, surpassing handcrafted solutions like TCM and GSS by leveraging memory-augmented neural networks (MANNs) \cite{weston2014memory,graves2016hybrid} and one-shot meta-learning. 
Mayfly is trained in two offline phases: the larval phase, where it acquires fundamental summarization abilities from larval meta-tasks generated using synthetic data to promote generalization across distributions; and the metamorphosis phase, where it rapidly adapts to real graph streams via metamorphosis meta-tasks sampled from target data, enhancing specialization. 
Architecturally, Mayfly employs a novel joint storage paradigm that separately compresses node and edge weights within a shared memory module to preserve graph structural information, enabling efficient one-pass processing. 
It incorporates information pathways—configurable routes through the network that guide data flow during meta-learning—to support diverse graph queries without retraining, including edge queries (e.g., estimating edge weights), node queries (e.g., node degrees), connectivity queries (e.g., reachability), path queries (e.g., shortest path approximations), and subgraph queries (e.g., motif counting). 
This design mitigates issues like hash collisions in low-budget scenarios and achieves superior accuracy by exploiting deep implicit features.

Related neural data stream structures, while not graph-specific, offer insights for graph applications.
MetaSketch~\cite{cao2023meta} is a neural sketch for general item frequency estimation, using meta-learning on synthetic Zipf-distributed meta-tasks in a pre-training phase, with basic and advanced versions that adapt to skewed real streams via online sampling, outperforming traditional sketches through differentiable encoding/decoding.
LegoSketch~\cite{feng2025lego}, a scalable MANN for data stream sketching, coordinates multiple memory bricks for dynamic adaptation to varying space budgets and domains, incorporating normalized multi-hash embedding for domain-agnostic scalability, self-guided weighting loss for meta-task optimization, and memory scanning for feature reconstruction, achieving better space-accuracy trade-offs. 
However, they share two common issues: 1) lack of contextual optimization for memory management, limiting compression; 2) hash-based sharding for memory expansion struggles with continuously growing data volumes.

%% file: MainText/3-Definition.tex
\section{Problem Definition}
\label{sec:definition}

\textbf{Graph Stream.} 
A graph stream $\mathcal{S} = \{e_1, e_2, \dots, e_t, \dots\}$ is a continuous, unbounded sequence of edge updates. 
Each item $e_t$ in the stream represents a directed edge update at timestamp $t$, denoted as a tuple $e_t = (o_t, d_t, w_t)$, where $o_t \in V$ is the origin (source) node, $d_t \in V$ is the destination (target) node, and $w_t > 0$ is the weight of the update.
Updates thus only accumulate, placing the task in the cash-register model~\cite{muthukrishnan2005streams} that all sketches compared here target: an observed interaction cannot be withdrawn, so an edge's cumulative weight is non-decreasing as the stream proceeds.
Concretely, in the IP traces a node is an IPv4 address and each packet contributes one update, with $o_t$ and $d_t$ its source and destination address; in the user-interaction streams a node is a user and an update is one message or reply.
The underlying graph $G(V, E)$ dynamically evolves as new edges arrive.

\textbf{Edge Weight Estimation.}
The primary task in graph stream summarization is to estimate the cumulative weight of any given edge $e = (o, d)$.
Let $f(e)$ denote the true cumulative weight of edge $e$ over the stream, defined as:
\begin{equation}
    f(e) = \sum_{t \,:\, (o_t, d_t) = (o, d)} w_t.
\end{equation}
The sum ranges over the prefix of $\mathcal{S}$ observed when the query is issued---on the IP trace, the number of packets $o$ has sent to $d$ so far.

Beyond this primary task, a versatile graph stream summarization scheme is expected to support a broad range of analytical queries.
Following the query taxonomy adopted by recent graph stream sketches~\cite{tang2016tcm,feng2023mayfly}, we further consider the following three tasks, each of which can be answered by reading the underlying summary directly or by composing it with lightweight aggregations.
Errors in the estimated edge weights therefore propagate to the answers of these tasks, and \S\ref{subsec:casestudy} evaluates all three as downstream case studies.

\textbf{Node Flow Query.}
For a queried node $v \in V$, the node flow query reports the cumulative weight of all edge updates incident to $v$ along a specified direction.
We distinguish the out-flow and in-flow of $v$ as
\begin{equation}
    F^{\text{out}}(v) = \!\!\!\sum_{t \,:\, o_t = v}\!\!\! w_t,
    \qquad
    F^{\text{in}}(v) = \!\!\!\sum_{t \,:\, d_t = v}\!\!\! w_t.
\end{equation}
\textbf{Path Flow Query.}
Let $\psi = \langle e_1, e_2, \dots, e_k \rangle$ be a query path consisting of $k$ directed edges.
The path flow query returns the maximum amount of flow that $\psi$ can carry.
Adopting the maximum-flow formulation~\cite{feng2023mayfly}, this quantity is upper-bounded by the bottleneck (minimum) edge weight along the path, and is set to zero whenever any edge of $\psi$ is absent from the underlying graph:
\begin{equation}
    f_p(\psi) =
    \begin{cases}
        0, & \text{if } \exists\, e_i \in \psi : f(e_i) = 0, \\
        \min\limits_{1 \le i \le k} f(e_i), & \text{otherwise}.
    \end{cases}
\end{equation}
\textbf{Subgraph Flow Query.}
Let $G_Q(V_Q, E_Q)$ denote a query subgraph of $G(V, E)$ with $V_Q \subseteq V$ and $E_Q \subseteq E$.
The subgraph flow query returns the aggregated weight of all edges contained in $G_Q$:
\begin{equation}
    f_g(G_Q) = \sum_{e \in E_Q} f(e).
\end{equation}
On the IP trace, the three queries take concrete forms: $F^{\text{out}}(v)$ is the total traffic sent by host $v$, monitored to identify hosts whose sending volume departs from its usual level; $f_p(\psi)$ is the traffic a multi-hop route $\psi$ could sustain end to end; and $f_g(G_Q)$ is the total traffic exchanged inside a set of hosts of interest, such as a subnet.

Aggregations that select updates by a predicate on non-key attributes, such as the protocol or port of a packet trace, reduce to sums of standard edge queries by augmenting the aggregation key, $(o,d) \to (o,d,a)$; \appref{app:attributed} gives the construction, a worked example, and its scope.

\textbf{Objective.}
Given a strict memory budget $\mathcal{M}$, the goal is to maintain a compact data structure (sketch) that processes the stream $\mathcal{S}$ in a single pass. 
For any queried edge $e$, the sketch should return an estimated weight $\hat{y}(e)$ that closely approximates the true weight $f(e)$.
The objective is to minimize the estimation error across all queries, while ensuring the memory footprint does not exceed $\mathcal{M}$:
\begin{equation}
    \min \mathbb{E}_{e \in \mathcal{S}} \left[ \mathcal{L}(f(e), \hat{y}(e)) \right] \quad \text{s.t. Memory} \le \mathcal{M}
\end{equation}
where $\mathcal{L}$ is the error function (e.g., Mean Absolute Error).

%% file: MainText/4-Design.tex
\section{Crane Design}

\begin{figure*}[t]
    \centering
    \includegraphics[width=0.99\textwidth]{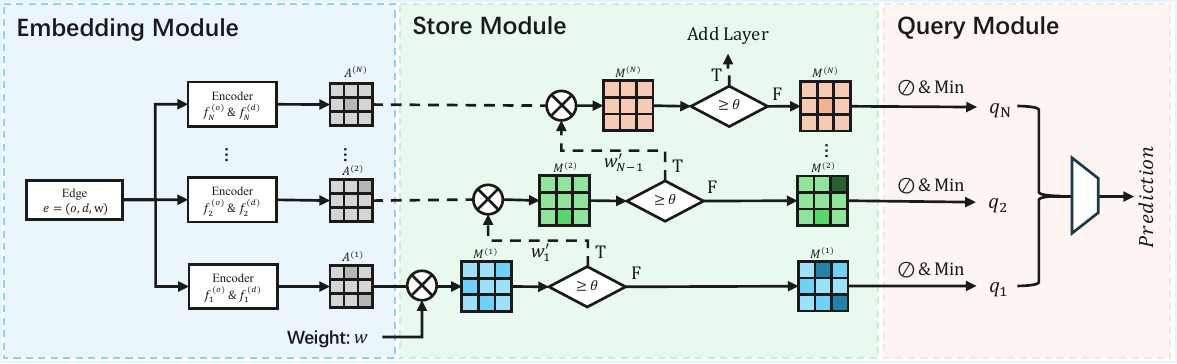}
    \vspace{-0.15in}
    \caption{The overall architecture of Crane.}
    \label{fig:design-main}
    \vspace{-0.1in}
\end{figure*}

\subsection{Overview}

The central design principle of Crane is to route edges of different weight magnitudes into corresponding memory layers, thereby alleviating the decoding difficulty that neural networks face when edges with vastly different weights are entangled in a flat structure.

As illustrated in Figure~\ref{fig:design-main}, Crane follows the standard Store-and-Query paradigm and comprises three core modules:
(1) an \textbf{Embedding} module, which encodes discrete edge identifiers into continuous embeddings;
(2) a \textbf{Store} module, which writes the edge embeddings along with their weights into a hierarchical memory;
and (3) a \textbf{Query} module, which recovers the cumulative weight of a queried edge from the hierarchical memory.

In the remainder of this section, we detail the implementation of these components and the training strategy.

\subsection{Embedding Module}
\label{sec:embedding-module}

Unlike prior neural data structures such as Mayfly, which use a single shared encoder to produce embeddings, Crane's hierarchical memory requires a distinct encoder pair for each layer.
Using independent encoders allows collisions that occur in a lower layer to be redistributed when edges are promoted to a higher layer, rather than being inherited rigidly.

As illustrated in the encoding branch of Figure~\ref{fig:design-main}, the Embedding module maps a discrete directed edge $e = (o, d)$ to a two-dimensional continuous representation called a \textit{basis matrix}.
Specifically, we first obtain fixed-length binary encodings $B_o$ and $B_d$ of the source and destination node identifiers.
At each layer $i \in \{1, \dots, N\}$, two independent MLP encoders transform these binary codes into positive continuous embeddings:
\begin{equation}
\mathbf{E}_o^{(i)} = f_i^{(o)}(B_o), \quad
\mathbf{E}_d^{(i)} = f_i^{(d)}(B_d),
\end{equation}
where $\mathbf{E}_o^{(i)} \in \mathbb{R}^{H}$ and $\mathbf{E}_d^{(i)} \in \mathbb{R}^{W}$ are element-wise strictly positive thanks to a Softplus activation $\sigma(x) = \log(1 + e^x)$, a smooth, strictly positive surrogate of ReLU, which guarantees numerical stability in the subsequent division operations.
We then construct the basis matrix $A^{(i)} \in \mathbb{R}^{H \times W}$ via the outer product:
\begin{equation}
A^{(i)} = \mathbf{E}_o^{(i)}\, (\mathbf{E}_d^{(i)})^\top.
\end{equation}
The outer-product construction treats each edge as a unique two-dimensional signal pattern on the memory grid, enabling the encoders to distribute the edge's information across the entire $H \times W$ surface rather than mapping it to a single bucket.
Because the two encoders are independent, the construction is also direction-sensitive: $A^{(i)}$ is the same for $(o,d)$ and $(d,o)$ only if the two nodes collide under both encoders at once, an event that \appref{app:direction} bounds like any other collision.
Algorithm~\ref{alg:embed} summarizes the per-layer embedding procedure.

\begin{algorithm}[t]
\small
\caption{Embedding Module}
\label{alg:embed}
\begin{algorithmic}[1]
\Require Edge $e=(o,d)$; layer index $i$;
         layer-specific encoders $f_i^{(o)},f_i^{(d)}$ (Softplus-activated).
\Ensure Basis matrix $A^{(i)}\in\mathbb{R}^{H\times W}$.
\vspace{0.2em}
\Function{Embed}{$e, i$}
  \State $B_o \leftarrow \textsc{BinaryEncode}(o)$; \quad
         $B_d \leftarrow \textsc{BinaryEncode}(d)$
  \State $\mathbf{E}_o^{(i)} \leftarrow f_i^{(o)}(B_o)$; \quad
         $\mathbf{E}_d^{(i)} \leftarrow f_i^{(d)}(B_d)$
  \State $A^{(i)} \leftarrow
         \mathbf{E}_o^{(i)}\,(\mathbf{E}_d^{(i)})^\top$
  \State \Return $A^{(i)}$
\EndFunction
\end{algorithmic}
\end{algorithm}

\subsection{Store Module}
\label{sec:store-module}

When edges of vastly different cumulative weights are stored in the same memory cells, the neural decoder must infer a large dynamic range from heavily entangled signals, which is intrinsically difficult.
A natural remedy is to separate edges of different magnitudes across layers.
However, realizing such a separation faces two key challenges.
First, mainstream sketching methods—both neural and hand-crafted—perform context-free updates, so they cannot directly access an estimate of the arriving edge's cumulative weight during insertion.
Second, modifying the encoding as a function of observed counts may break consistency between past and future encodings and complicate decoding.

Crane addresses both challenges with a hierarchical carry mechanism inspired by positional numeral systems.

As shown in Figure~\ref{fig:design-main}, the Store module operates analogously to carry propagation in a positional numeral system.
When a new edge $e = (o, d, w)$ arrives, Crane first generates basis matrices $\{A^{(i)}\}_{i=1}^{N}$ via the Embedding module, then writes the edge into the lowest memory layer:
\begin{equation}
M^{(1)} \leftarrow M^{(1)} + w \cdot A^{(1)}.
\end{equation}
After the write, Crane queries the local estimate at layer~1.
Since the memory is a superposition of non-negative basis matrices, the ratio between any memory cell and the corresponding basis entry gives an estimate of the edge's weight that the writes of colliding edges can only raise.
Taking the minimum over all positions yields the tightest such estimate, following the Count-Min principle.
We define this per-layer estimator as:
\begin{equation}
\label{eq:query-op}
Q(M, A) = \min_{p,q}\, \bigl(M \oslash A\bigr)_{pq},
\end{equation}
where $\oslash$ denotes element-wise division.
Writing $e$ with weight $w$ adds exactly $w \cdot A$ to the memory, so each ratio $(M \oslash A)_{pq}$ reports the multiple of $A$ present at that position, which is the weight itself: absent collisions every position returns $w$, and the writes of colliding edges can only raise it.
The local estimate at layer~1 is then $\hat{y}_1(e) = Q(M^{(1)}, A^{(1)})$.
If this estimate reaches a user-specified carry threshold $\theta > 1$, a carry is triggered: the edge's \emph{entire} accumulated contribution is removed from the current layer and a scaled portion is promoted to the next layer.
This process repeats at each subsequent layer until the estimate falls below~$\theta$.

Formally, at layer $i$, let $\hat{y}_i(e) = Q(M^{(i)}, A^{(i)})$ denote the current local estimate.
If $\hat{y}_i(e) \ge \theta$, Crane performs:
\begin{align}
M^{(i)}   &\leftarrow M^{(i)}   - \hat{y}_i(e) \cdot A^{(i)}, \label{eq:carry-remove} \\
M^{(i+1)} &\leftarrow M^{(i+1)} + \frac{\hat{y}_i(e)}{\theta} \cdot A^{(i+1)}, \label{eq:carry-promote}
\end{align}
and the procedure continues at layer $i+1$.
We write $w'_i := \hat{y}_i(e)/\theta$ for the promoted weight (cf.\ Figure~\ref{fig:carry-example}).
Unlike standard integer carry, which retains a residual in the current digit, Crane transfers the \emph{entire} accumulated value and divides it by $\theta$ before promoting.
This aggressive promotion is possible because the basis matrices operate in a continuous domain, and it maximizes interference isolation by leaving the lower layer completely clean for future insertions of tail edges.

A subtle consequence of this removal step is that $\hat{y}_i(e)$ is \emph{not} a strict upper bound on the residual weight of $e$ at layer~$i$.
The strictly positive embeddings ensure that hash collisions can only inflate the per-cell estimate, so in a purely additive sketch $\hat{y}_i(e)$ would indeed upper-bound the truth.
With carry, however, Eq.~\eqref{eq:carry-remove} subtracts a colliding edge's \emph{entire} accumulated contribution upon promotion, which also wipes out the collision-induced inflation that was previously shared with $e$'s cells and may depress some of them below $e$'s true residual.
In practice, the $\min$ over positions in $Q(\cdot,\cdot)$ largely mitigates this: a single carry rarely depresses every position of $A^{(i)}$ simultaneously, so at least one cell typically retains a reliable estimate.

\textbf{Automatic expansion.}
If the carry propagates to the topmost allocated layer $N$ and the estimate is still at least $\theta$, Crane triggers automatic expansion: provided $N < N_{\max}$, a new empty layer $M^{(N+1)} \leftarrow \mathbf{0}$ is allocated and the carry proceeds into it.
If instead $N = N_{\max}$, the hierarchy has reached its configured cap; the excess mass is left in the top layer $M^{(N)}$ and no further carry is issued.
This allows the memory hierarchy to grow organically with the stream volume, avoiding the need for pre-allocation.

\begin{figure}[t]
    \centering
    \includegraphics[width=0.98\linewidth]{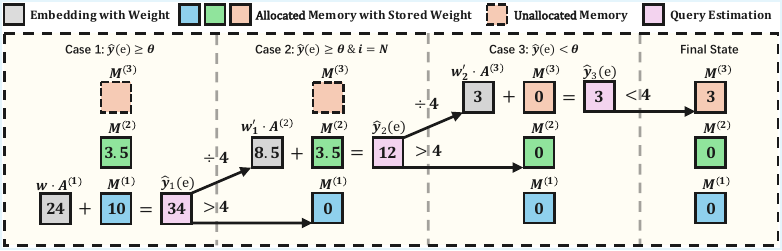}
    \vspace{-0.15in}
    \caption{Hierarchical Carry Mechanism and Automatic Memory Expansion Strategy.}
    \label{fig:carry-example}
    \vspace{-0.25in}
\end{figure}

\textbf{Example.} Figure~\ref{fig:carry-example} illustrates a concrete carry sequence with $\theta = 4$, starting from $N = 2$ allocated layers where the target edge~$e$ has an estimated weight of $10$ in layer~1 and $3.5$ in layer~2. 
At time $t$, a new update with weight $w = 24$ arrives.
After insertion into layer~1, the local estimate becomes $10 + 24 = 34$.
Since $34 > \theta$, a carry is triggered: the entire contribution ($34$) is removed from layer~1 via Eq.~\eqref{eq:carry-remove}, and $34 / \theta = 8.5$ is promoted to layer~2 via Eq.~\eqref{eq:carry-promote}.
Layer~2 now holds $3.5 + 8.5 = 12$.
Since $12 > \theta$, the carry continues: layer~2 is cleared and $12 / \theta = 3$ is promoted.
Because layer~2 is the current top and $N < N_{\max} = 4$, a new layer~3 is automatically allocated.
Layer~3 receives $3 < \theta$, so the carry terminates.
The final state is $[0,\; 0,\; 3]$ across the three layers, compactly encoding the total weight $3 \times \theta^{2} = 48$.

\textbf{Choosing $\theta$.}
The threshold is a fixed hyperparameter, not a learned one.
Its role is to trade per-layer occupancy against hierarchy depth.
For heavy hitters of cumulative weight $F_{\text{heavy}}$, keeping every layer's local estimate below $\theta$ within a layer budget $N_{\max}$ requires $\theta \gtrsim (F_{\text{heavy}}+1)^{1/N_{\max}}$ (Lemma~\ref{lemma:space}); below that the top layer saturates, while a much larger $\theta$ leaves edges concentrated in the lower layers and forfeits the separation the hierarchy provides.
Empirically the resulting sensitivity is U-shaped with a dataset-dependent optimum, and streams with larger $F_{\text{heavy}}$ prefer a larger $\theta$ (\S\ref{sec:ablation}); we use $\theta = 4$ throughout as a single default rather than tuning per dataset.

\textbf{Mini-batch conservative carry.}
The per-edge carry in Eqs.~\eqref{eq:carry-remove}--\eqref{eq:carry-promote} yields the most accurate promotion decisions but is strictly sequential: each edge must observe the memory state produced by all its predecessors, which prevents efficient GPU execution.
A naive parallelization that issues an independent carry for every batched edge whose local estimate is at least $\theta$ is unsafe---colliding edges would each subtract the full collision-inflated signal from the same cells, driving memory entries negative and breaking the non-negativity assumption that the $\min$-based estimator relies on.
To reconcile correctness with parallelism, Crane processes updates in mini-batches $B$ and performs a single \emph{conservative} carry per batch.
We first aggregate the per-edge basis matrices and write the batch into layer~$1$:
\begin{equation}
\setlength\abovedisplayskip{-0.01cm}
\setlength\belowdisplayskip{-0.01cm}
A^{(i)}_B = \sum_{e \in B} A^{(i)}_e, \qquad
M^{(1)} \leftarrow M^{(1)} + \sum_{e \in B} w_e\, A^{(1)}_e,
\end{equation}
and then evaluate a single batch-level estimate
\begin{equation}
\hat{y}^{(i)}_B = Q\!\left(M^{(i)}, A^{(i)}_B\right).
\end{equation}
Mirroring Eqs.~\eqref{eq:carry-remove}--\eqref{eq:carry-promote}, if $\hat{y}^{(i)}_B \ge \theta$ Crane performs a single shared carry for the entire batch:
\begin{align}
M^{(i)}   &\leftarrow M^{(i)}   - \hat{y}^{(i)}_B \cdot A^{(i)}_B, \\
M^{(i+1)} &\leftarrow M^{(i+1)} + \frac{\hat{y}^{(i)}_B}{\theta} \cdot A^{(i+1)}_B,
\end{align}
and the procedure continues at layer~$i+1$; otherwise the batch terminates at this layer. 
Because $A^{(i)}_B$ is the sum of all per-edge basis matrices in the batch, the denominator in $Q(M^{(i)}, A^{(i)}_B)$ is inflated by the collision signals that the batched edges contribute to one another's cells.
The resulting $\hat{y}^{(i)}_B$ is therefore systematically smaller than the per-edge $\hat{y}_i(e)$ that a strictly sequential schedule would observe, so the shared promotion under-estimates the total amount of mass that sequential processing would have carried out of this batch.
This is exactly why we call it a \emph{conservative} carry: collision-inflated signals deflate the batch-level estimate, guaranteeing that the subtraction $\hat{y}^{(i)}_B \cdot A^{(i)}_B$ stays within what every individual edge could have safely issued and thus preserves the non-negativity of $M^{(i)}$, at the cost of forfeiting some eager promotions.
In practice, this scheme recovers near-sequential accuracy while unlocking the GPU parallelism needed for high-throughput insertion (\S\ref{sec:ablation-batch} quantifies the accuracy--batch-size trade-off).
The complete pseudocode, unifying the sequential ($|B|{=}1$) and batched ($|B|{>}1$) carries with automatic memory expansion in a single \textsc{Store}/\textsc{Query} procedure, is provided in \appref{app:pseudocode}\algparen.

\subsection{Query Module}
\label{sec:query-module}

Given a query edge $e = (o, d)$, the Query module retrieves the cumulative weight by reversing the encoding process and aggregating estimates across all allocated layers.

\textbf{Per-layer estimation.}
At each layer $i$, Crane regenerates the basis matrix $A^{(i)}$ via the Embedding module and computes a local estimate using the $Q$ operator defined in Eq.~\eqref{eq:query-op}:
\begin{equation}
q_i(e) = Q\bigl(M^{(i)},\, A^{(i)}\bigr).
\end{equation}
Because the memory is a superposition of positive patterns, the writes of colliding edges can only raise each ratio $M^{(i)}_{pq} / A^{(i)}_{pq}$; the minimum filters out collision noise and yields the tightest estimate.

\textbf{Multi-layer decoding.}
The per-layer estimates are collected into a vector
\begin{equation}
\mathbf{q}(e) = [q_1(e),\, q_2(e),\, \dots,\, q_N(e)]^\top,
\end{equation}
and aggregated by a single learnable linear decoder:
\begin{equation}
\hat{y}(e) = \mathbf{w}^\top \mathbf{q}(e) + b,
\end{equation}
where $\mathbf{w} \in \mathbb{R}^N$ and $b \in \mathbb{R}$ are trainable parameters.
This decoder is intentionally lightweight: the modeling capacity is devoted to the per-layer encoders, while the decoder learns to weight each layer's estimate according to its noise level, naturally approximating inverse-variance weighting (see \appref{proof:variance}).
Deterministically reversing the carry would read the per-layer estimates as base-$\theta$ digits, $\hat{y}(e) = \sum_{i=1}^{N} \theta^{\,i-1} q_i(e)$.
Absent collisions this reversal recovers the cumulative weight exactly; with collisions the per-layer estimates are perturbed, and a pure reversal carries that error through unchanged.
Since the positional weights $\theta^{\,i-1}$ are constants, that reversal is exactly the decoder above under the fixed choice $\mathbf{w} = (1, \theta, \dots, \theta^{\,N-1})$ and $b = 0$.
Learning $\mathbf{w}$ therefore costs nothing extra, contains carry reversal as a special case, and additionally corrects the collision-induced error.

Unlike Mayfly, which feeds a rich set of statistics (sums, maxima, minima, node encodings, and even the total edge count) into a complex decode network, Crane uses only the minimum-ratio statistic $Q$ as input.
This deliberate restriction avoids two pitfalls: (1) \emph{shortcut learning}, where the decoder overfits to global stream magnitude rather than learning to resolve local collisions; and (2) low signal-to-noise ratio from sum-based pooling, which dilutes information density.
Shortcut learning is driven primarily by the total edge count: early in training, a decoder given this global statistic can already lower the loss with a coarse estimate derived from it rather than from the memory contents, and optimization settles on that shortcut; Crane's decoder receives no global statistic, so the shortcut is structurally unavailable.
The minimum ratio, in turn, is already an effective signal---the tightest per-layer estimate and the standard estimator of Count-Min-style sketches~\cite{cmsketch}---while augmenting it with less informative statistics, such as sums that accumulate interference, would weaken the training signal and give the decoder more freedom to overfit collision noise.

\vspace{-0.05in}
\subsection{Training}
\label{sec:training}

Crane is trained entirely offline on synthetic graph streams and is then directly evaluated on real workloads for the edge weight estimation task without any additional fine-tuning.
Unlike Mayfly, we do not introduce a separate adaptation phase on real graph streams, and all parameters of the modules are shared across tasks and datasets.

\vspace{-0.05in}
\subsubsection{Task formulation.}
Each training task $t$ simulates storing and querying a single graph stream instance.
Formally, task $t$ consists of a support set $S_t$ and a query set $Q_t$.
The support set
\[
\setlength\abovedisplayskip{-0.01cm}
\setlength\belowdisplayskip{-0.01cm}
S_t = \{x^{(t)}_1, x^{(t)}_2, \dots, x^{(t)}_{L_t}\}
\]
is a sequence of streaming edges
$x^{(t)}_\ell = (o_\ell, d_\ell, w_\ell)$,
which are fed to the \textsc{Store} operation\algparen{} in arrival order.
The query set
\[
\setlength\abovedisplayskip{-0.01cm}
\setlength\belowdisplayskip{-0.01cm}
Q_t = \{(e^{(t)}_j, y^{(t)}_j)\}_j
\]
contains a subset of edges $e^{(t)}_j = (o_j, d_j)$ together with their exact cumulative weights
\[
\setlength\abovedisplayskip{-0.01cm}
\setlength\belowdisplayskip{-0.01cm}
y^{(t)}_j \;=\; \sum_{\ell : (o_\ell, d_\ell) = (o_j, d_j)} w_\ell
\]
computed from the full support stream $S_t$.
During training, we reset the hierarchical memories $\{M^{(i)}\}$ at the beginning of each task, run \textsc{Store} on $S_t$ once, and then apply \textsc{Query} to all edges in $Q_t$ to obtain predictions $\hat{y}^{(t)}_j$.

\subsubsection{Synthetic graph-stream generation.}

Training directly on real-world graph streams risks introducing scenario-specific biases (e.g., particular network topologies or traffic patterns), which can impair generalization to unseen workloads.
To avoid this, we distill the common statistical regularities of real streams—heavy-tailed weight distributions, varying skewness, and diverse collision patterns—into a synthetic data generator, and train Crane exclusively on the resulting synthetic tasks so that it generalizes across different real-world scenarios.
We use a unified numerical ID space for nodes (e.g., integers mapped to fixed-length binary codes), so the encoders in \S\ref{sec:embedding-module} are trained in a domain-agnostic manner and can be reused on real streams without retraining.

For each task $t$, we construct a synthetic stream $S_t$ in three steps.
First, we sample a stream length $L_t$ from a range $[1, \Gamma]$ and generate endpoints $(o_\ell, d_\ell)$ by drawing node IDs independently from the global ID space.
This produces diverse graph structures and collision patterns within the sketch.
Second, we sample a Zipf distribution $\mathrm{Zipf}(\alpha_t)$ from a pool of exponents $\alpha_t \in [\alpha_{\min}, \alpha_{\max}]$, reflecting the skewed edge-weight distributions commonly observed in real graphs.
We then generate a normalized weight profile $\{\tilde{w}_\ell\}_{\ell=1}^{L_t}$ by i.i.d.\ sampling from $\mathrm{Zipf}(\alpha_t)$ and rescaling so that $\sum_\ell \tilde{w}_\ell = 1$.
Finally, we draw a total weight $W_t$ from a log-uniform range and set $w_\ell = W_t \tilde{w}_\ell$ for each edge.
The query set $Q_t$ consists of all distinct edges in $S_t$, together with their exact cumulative weights $\{y^{(t)}_j\}$ computed by aggregation.

\subsubsection{Loss function.}
Cumulative edge weights in graph streams typically follow a heavy-tailed distribution, spanning several orders of magnitude.
To ensure robust optimization across this wide dynamic range, we adopt the Mean Absolute Error (MAE) as our training objective.
Unlike squared error metrics that can be overly sensitive to outliers (i.e., heavy-hitter edges), the $L_1$ loss provides a stable gradient signal that encourages accurate estimation for both heavy hitters and tail edges.

Formally, given a task $t$ and its query set $Q_t$, the loss is defined as:
\begin{equation}
\label{eq:task-loss}
\mathcal{L}_t
\;=\;
\frac{1}{|Q_t|}
\sum_{(e^{(t)}_j, y^{(t)}_j) \in Q_t}
\left| y^{(t)}_j - \hat{y}^{(t)}_j \right|,
\end{equation}
where $\hat{y}^{(t)}_j$ denotes the cumulative weight predicted by the decoder, and $y^{(t)}_j$ is the ground-truth cumulative weight.
Minimizing this objective directly aligns the model's predictions with the true counts, ensuring that the learnable components—the per-layer encoders and the linear decoder—converge to a state that accurately preserves the stream statistics.

\vspace{-0.05in}
\subsubsection{Optimization.}
We train Crane by sampling batches $\mathcal{T}$ of tasks on the fly and minimizing the average task loss
\(
\mathcal{L} = \frac{1}{|\mathcal{T}|} \sum_{t \in \mathcal{T}} \mathcal{L}_t
\)
with AdamW~\cite{loshchilov2017decoupled}.
Within each task, edges are processed in batches of streaming updates so that the carry mechanism of \S\ref{sec:store-module}\algparen{} can be applied efficiently.
All learnable parameters, namely the encoder networks for each layer and the linear decoder, are optimized jointly.
Training uses only synthetic tasks; at evaluation time, we fix all parameters and run Crane in a purely streaming fashion on real graph streams, with no gradient updates and no dataset-specific tuning.

%% file: MainText/5-Maths.tex
\vspace{-0.05in}
\section{Mathematical Analysis}
\label{sec:theory}

In this section, we provide an idealized analysis of Crane.
Our goal is to explain why the hierarchical carry mechanism can improve scalability and reduce collision exposure.
We formalize logarithmic layer growth and an exponential decay bound on repeated cross-layer collisions, and we give heuristic interpretations for the decoder behavior and the comparison with flat structures.

\vspace{-0.05in}
\subsection{Preliminaries and Assumptions}
\label{subsec:prelim}
Let $\mathcal{S}$ be a graph stream of edge updates.
Let $f(e)$ denote the true cumulative weight of an edge $e$.
Crane maintains a hierarchy of $N$ memory layers $\{M^{(1)}, \dots, M^{(N)}\}$ with a promotion threshold $\theta > 1$.

\begin{itemize}[leftmargin=1em]
    \item \textbf{Assumption 1 (Layer-wise Bounded Collision).}
    For any two distinct edges $e$ and $e'$, let $C_i(e,e')$ denote the event that they collide at layer $i$.
    We assume that the probability of collision at each layer is bounded by a small constant $p \in (0,1)$, even conditioned on collisions in all preceding layers:
    \[
    \Pr\!\left(
        C_i(e,e')
        \mid
        C_1(e,e'),\ldots,C_{i-1}(e,e')
    \right)
    \leq p.
    \]
 
    \textbf{Constructive realization.}
    Assumption~1 is realizable by a hash-anchored neural encoder.
    At each layer $i$, let $h_o^{(i)}:V\rightarrow[H]$ and
    $h_d^{(i)}:V\rightarrow[W]$ be independently seeded pairwise-independent hash functions, and define
    \[
    \mathbf{E}_o^{(i)}(o)
    =
    s_o^{(i)}(o)
    \left(
    \epsilon \mathbf{1}_H
    +
    \mathbf{e}_{h_o^{(i)}(o)}
    \right),
    \qquad
    \mathbf{E}_d^{(i)}(d)
    =
    s_d^{(i)}(d)
    \left(
    \epsilon \mathbf{1}_W
    +
    \mathbf{e}_{h_d^{(i)}(d)}
    \right),
    \]
    where $\epsilon>0$, $\mathbf{e}_{h(\cdot)}$ is the corresponding one-hot vector, and
    $s_o^{(i)}(o),s_d^{(i)}(d)>0$ can be produced by lightweight neural networks with bounded Softplus outputs.
    The resulting basis matrix
    $A^{(i)}=\mathbf{E}_o^{(i)}(o)(\mathbf{E}_d^{(i)}(d))^\top$
    remains strictly positive, while its dominant cell is determined by the two hash functions.
    Hence, independently seeded layers provide a concrete neural construction satisfying the bounded-collision condition in Assumption~1.

    \textbf{Collision bound.}
    For two distinct edges $e=(o,d)$ and $e'=(o',d')$, the dominant-cell collision probability under the construction above is
    \[
    \Pr(C_i(e,e')) \le
    \begin{cases}
    \frac{1}{HW}, & o\neq o',\ d\neq d',\\[2pt]
    \frac{1}{W}, & o=o',\ d\neq d',\\[2pt]
    \frac{1}{H}, & o\neq o',\ d=d'.
    \end{cases}
    \]
    Therefore, for any two distinct edges, a conservative uniform bound is
    \[
    p=\max\left\{\frac{1}{H},\frac{1}{W}\right\}.
    \]

    The construction above establishes realizability of Assumption~1; our default unconstrained MLP encoders are more expressive, but are not claimed to provide a certified collision bound, so Theorem~\ref{thm:collision} should be interpreted conditionally on the stated bounded-collision assumption.
\end{itemize}

\textbf{Justification for Assumption 1:}
Assumption 1 does not require the learned neural encoders themselves to behave as pairwise-independent hash functions.
Instead, it only requires the probability of repeated collision to remain bounded across layers.
This bounded-collision condition serves as an idealized analytical abstraction for deriving collision-probability bounds.
In practice, the layer-specific encoders may further reduce the effective overlap between heavy hitters, while our theoretical analysis relies only on the stated collision bound $p$.

\vspace{-0.05in}

\subsection{Complexity Analysis}
\label{subsec:complexity}

We first address the computational efficiency of Crane at inference time, i.e., the space and time consumed by insertions and queries on a deployed sketch.
Training (\S\ref{sec:training}) is performed once offline and does not scale with the stream, so it is excluded from the bounds below; \S\ref{sec:throughput} reports its measured cost.
Unlike traditional sketches that often require linear memory expansion to handle heavy traffic, Crane utilizes a hierarchical carry-like mechanism to achieve logarithmic scaling in both space and time.

\begin{lemma}[Logarithmic Space Efficiency]
\label{lemma:space}
Given a promotion threshold $\theta$, the number of layers $N$ required to accurately represent an edge with maximum cumulative weight $F_{\text{max}}$ grows logarithmically, specifically $N \geq \log_\theta(F_{\text{max}} + 1)$.
Consequently, the total space complexity relative to the stream volume is logarithmic.
\end{lemma}

\begin{proof}
See \appref{proof:space}.
\end{proof}

\begin{theorem}[Logarithmic Amortized Time Complexity]
\label{thm:time}
Let $\mathcal{S}$ be a stream of $T$ edges with total accumulated weight $W_{\text{total}}$.
The total time complexity to process $\mathcal{S}$ is $\mathcal{O}(T \log_\theta W_{\text{total}})$, and the amortized time complexity per insertion is $\mathcal{O}(\log_\theta W_{\text{total}})$.
\end{theorem}

\begin{proof}
See \appref{proof:time}.
\end{proof}

\vspace{-0.05in}
\subsection{Error Bound Analysis}
\label{subsec:error-bound}

We now analyze the error bound, demonstrating how Crane's hierarchical structure and learnable components minimize collision interference.

\begin{theorem}[Exponential Decay of Pairwise Collision Probability]
\label{thm:collision}
Under Assumption~1, for any two distinct edges $e$ and $e'$, the probability that they collide consistently across $k$ layers decreases exponentially with $k$.
Specifically,
\[
\Pr\!\left(\bigcap_{i=1}^{k} C_i(e,e')\right) \leq p^k.
\]
\end{theorem}

\begin{proof}
See \appref{proof:collision}.
\end{proof}

\paragraph{Idealized Residual-Mass View.}
Let $\mathcal{E}$ be the set of all unique edges in the stream.
Let $W_{\text{total}} = \|f\|_1 = \sum_{e \in \mathcal{E}} f(e)$ denote the total weight of the stream.
In an idealized carry model, one may define a residual quantity
\[
\setlength\abovedisplayskip{-0.01cm}
\setlength\belowdisplayskip{-0.01cm}
W_{\text{res}} = \sum_{e \in \mathcal{E}} (f(e) \bmod \theta),
\]
which captures the mass that remains in the lowest layer after repeated promotion.
Under Assumption~1, this view suggests that the collision exposure of the bottom layer scales with a residual mass of order $W_{\text{res}}$ rather than the full mass $W_{\text{total}}$.
Accordingly, the bottom-layer contribution to the estimation error can be heuristically related to a term of order
\[
\mathcal{O}\left(pW_{\text{res}}\right).
\]

We emphasize that this is an idealized interpretation used to explain why hierarchical carry can reduce bottom-layer interference; it is not an exact error identity for the learned continuous-carry system.
See \appref{proof:residual} for further discussion.

\vspace{-0.05in}
\subsection{Optimization Properties}

\paragraph{Learnable Orthogonality.}
For the set of heavy hitters, training tends to reduce the overlap between basis matrices, as large inner products $\langle A_u, A_v \rangle$ contribute more strongly to reconstruction error.
This provides a heuristic explanation for why learned encoders may exhibit lower effective interference between heavy hitters in practice.
We do not, however, claim a formal collision-probability guarantee beyond the assumed bound in Assumption~1; see \appref{proof:learnable} for discussion.

\paragraph{Decoder Interpretation.}
The learnable linear decoder allows the model to assign different importance to different layers according to their empirical reliability.
Under an auxiliary idealized model in which layer-wise estimates are approximately unbiased and weakly correlated, the classical inverse-variance solution provides a useful interpretation of such weighting.
We therefore view the decoder as a practical reliability-aware aggregation mechanism rather than as a provably optimal estimator for the actual training objective.
See \appref{proof:variance} for discussion.

\vspace{-0.1in}
\subsection{Comparative Analysis: Hierarchical vs. Flat Structures}
\label{subsec:comparison}

Finally, we provide a qualitative comparison between Crane and Mayfly~\cite{feng2023mayfly}, which maps all updates into a single shared memory space $M_{\text{flat}}$.
Under the idealized residual-mass view above, Crane's hierarchical carry mechanism can reduce the collision exposure of tail edges relative to flat architectures.

\begin{sloppypar}
\paragraph{Interference Isolation.}
Let $\mathcal{S}$ be a graph stream obeying a heavy-tailed distribution, and let both structures operate under an identical per-layer memory size $\mathcal{W}$.
Under the idealized residual-mass interpretation, Crane can be viewed as reducing the bottom-layer collision exposure by a factor on the order of the ratio between the residual mass and the total mass:
\begin{equation}
\setlength\abovedisplayskip{-0.01cm}
\setlength\belowdisplayskip{-0.01cm}
    \frac{\mathbb{E}[\text{Error}_{\text{Crane}}]}{\mathbb{E}[\text{Error}_{\text{Flat}}]}
    \sim
    \mathcal{O}\left(\frac{W_{\text{res}}}{W_{\text{total}}}\right),
\end{equation}
where $W_{\text{total}} = \|f\|_1$ and $W_{\text{res}} = \sum_{e \in \mathcal{E}} \bigl(f(e) \bmod \theta\bigr)$.
This comparison is intended as an idealized qualitative explanation rather than a formal error guarantee for the learned continuous-carry system.
See \appref{proof:isolation} for discussion.
\end{sloppypar}

%% file: MainText/6-Evaluation.tex
\section{Evaluation}

\begin{figure*}[t]
	\vspace{-0.1cm}
	\centering
	\begin{minipage}{\linewidth}
		\centering
		\includegraphics[width=0.75\linewidth]{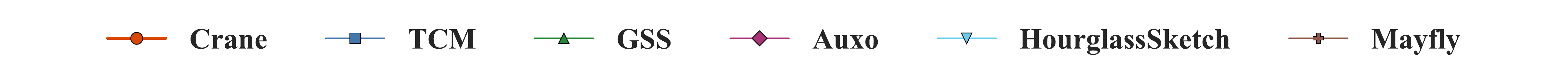}
		\vspace{-0.2cm}
	\end{minipage}
	\begin{minipage}{1 \linewidth}
        \captionsetup[subfigure]{skip=-0.02cm, labelfont+=footnotesize}
		\centering
		
		\subcaptionbox{\footnotesize Lkml (ARE)}[0.2\linewidth]{%
			\includegraphics[width=0.19\textwidth]{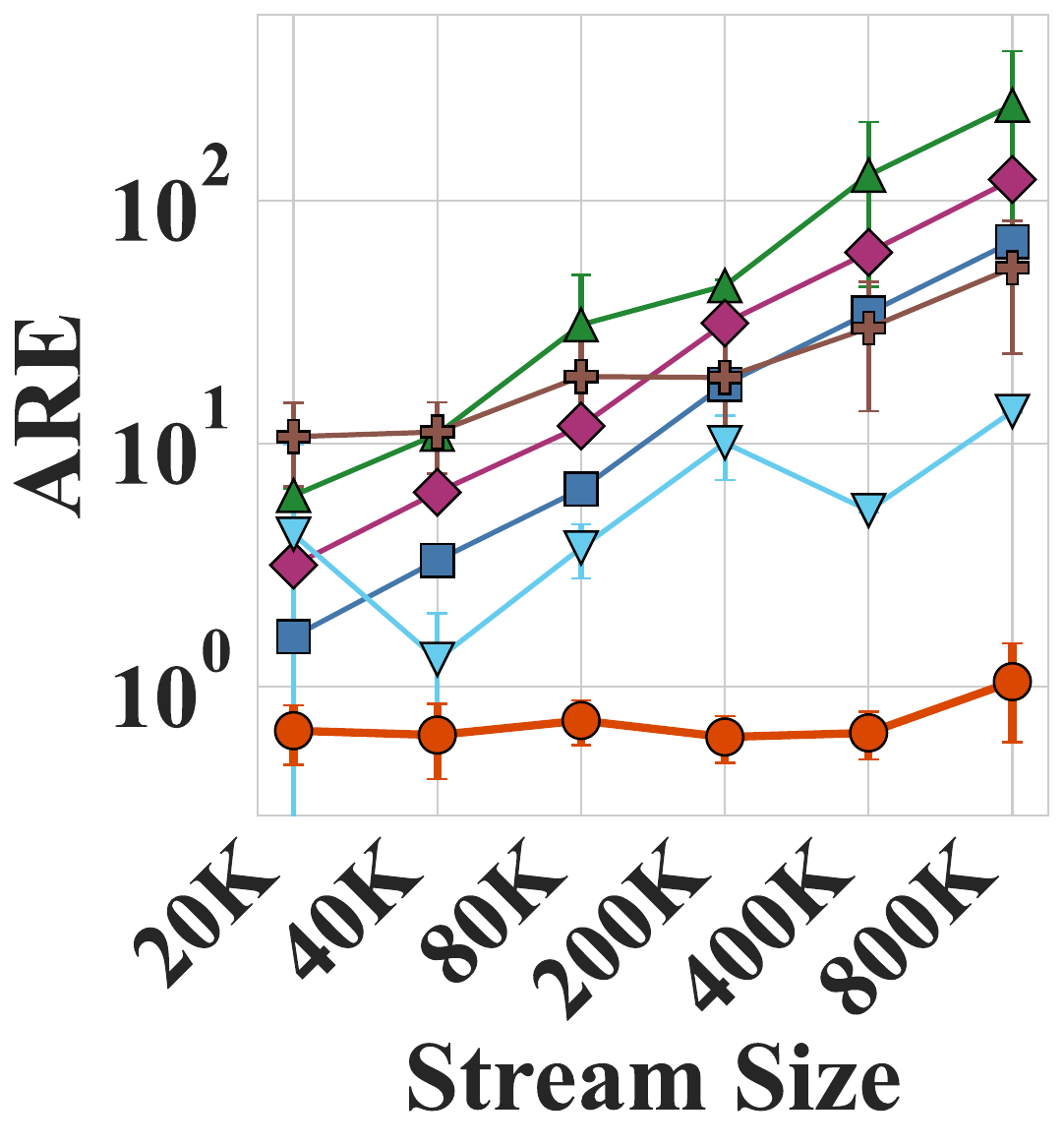}%
		}\hfill%
		\subcaptionbox{\footnotesize NotreDame (ARE)}[0.2\linewidth]{%
			\includegraphics[width=0.19\textwidth]{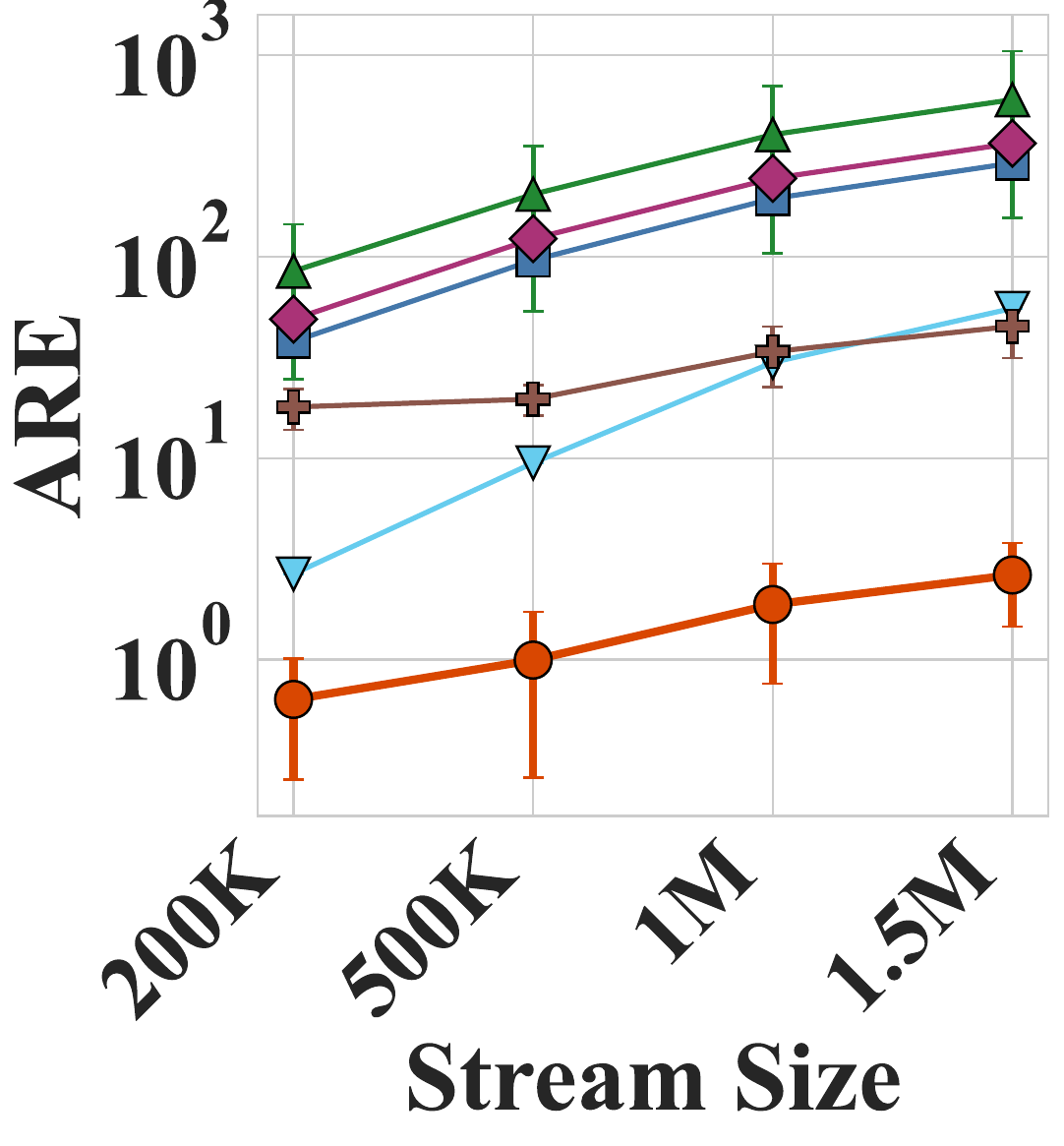}%
		}\hfill%
		\subcaptionbox{\footnotesize CAIDA (ARE)}[0.2\linewidth]{%
			\includegraphics[width=0.19\textwidth]{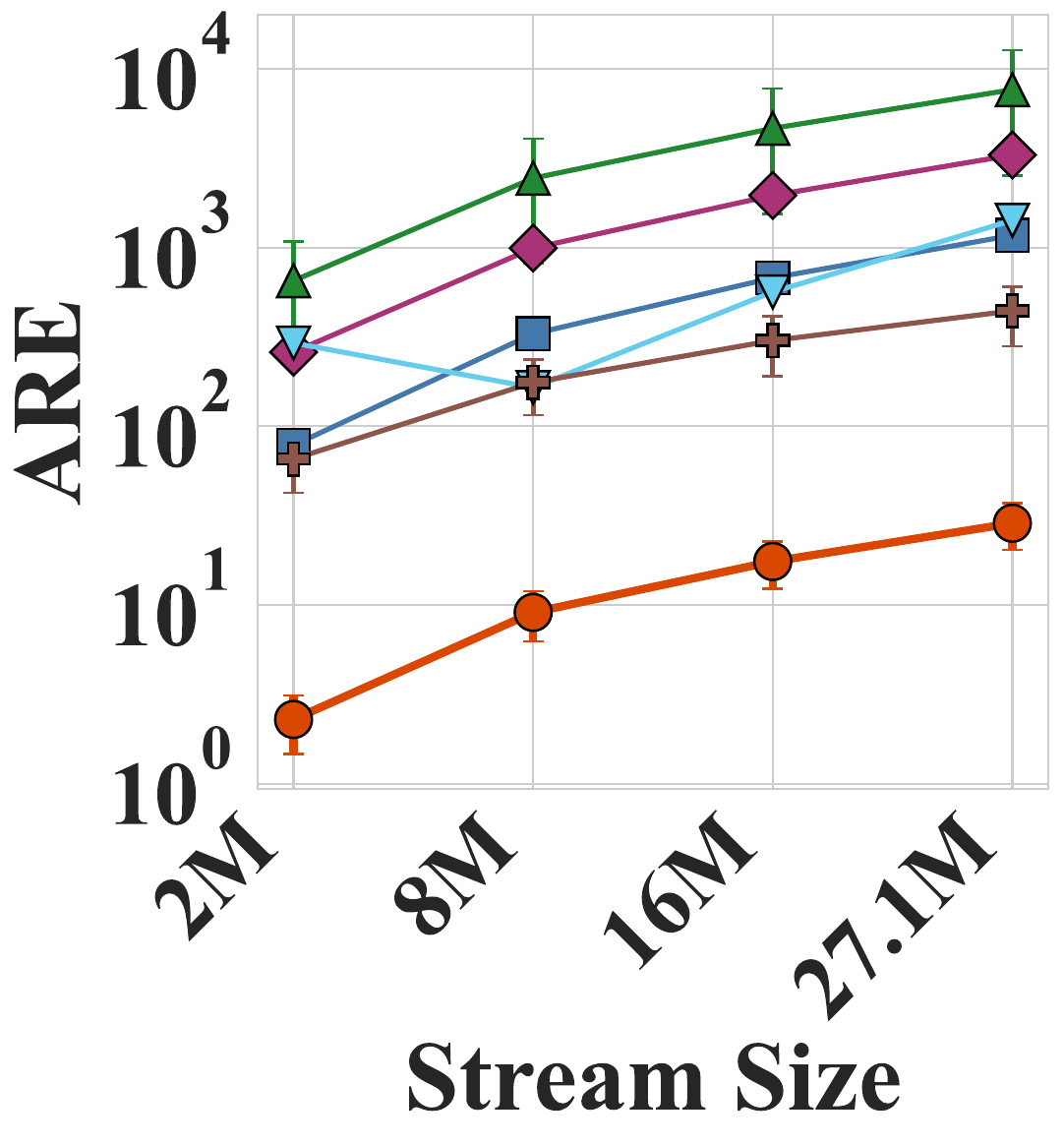}%
		}\hfill%
		\subcaptionbox{\footnotesize WikiTalk (ARE)}[0.2\linewidth]{%
			\includegraphics[width=0.19\textwidth]{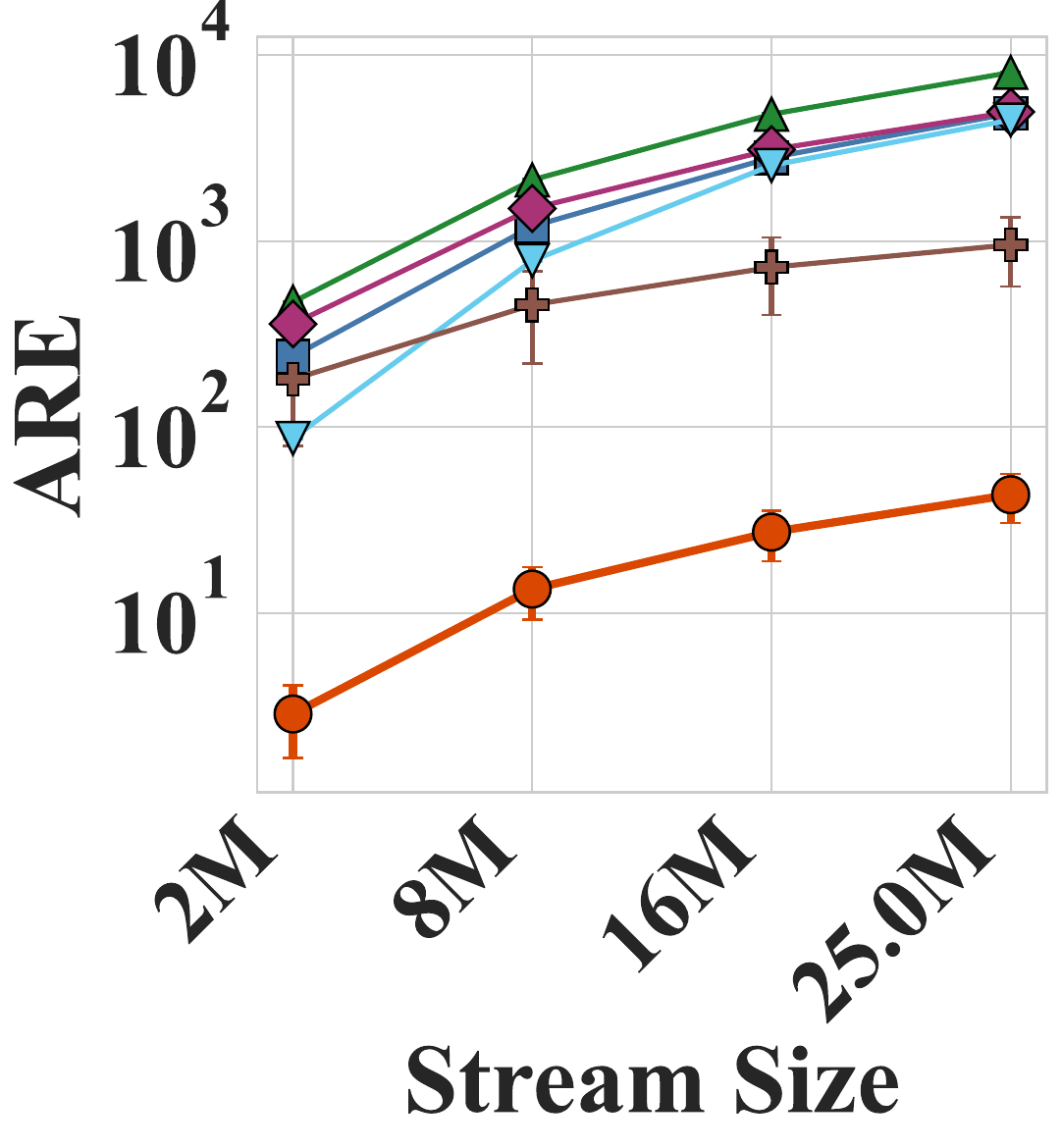}%
		}\hfill%
		\subcaptionbox{\footnotesize StackOverflow (ARE)}[0.2\linewidth]{%
			\includegraphics[width=0.19\textwidth]{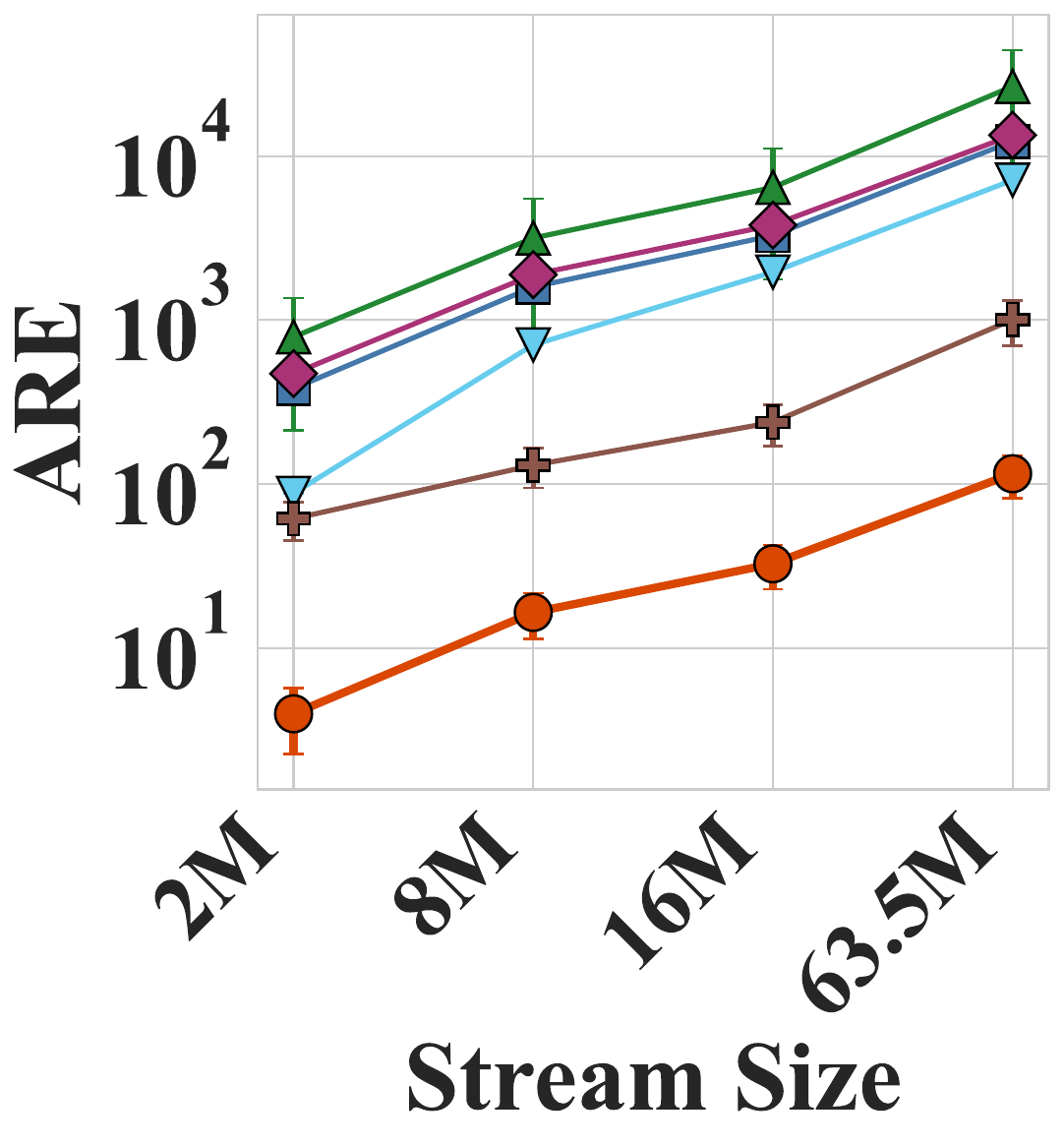}%
		}

        \vspace{0.3cm}
        
		\subcaptionbox{\footnotesize Lkml (AAE)}[0.2\linewidth]{%
			\includegraphics[width=0.19\textwidth]{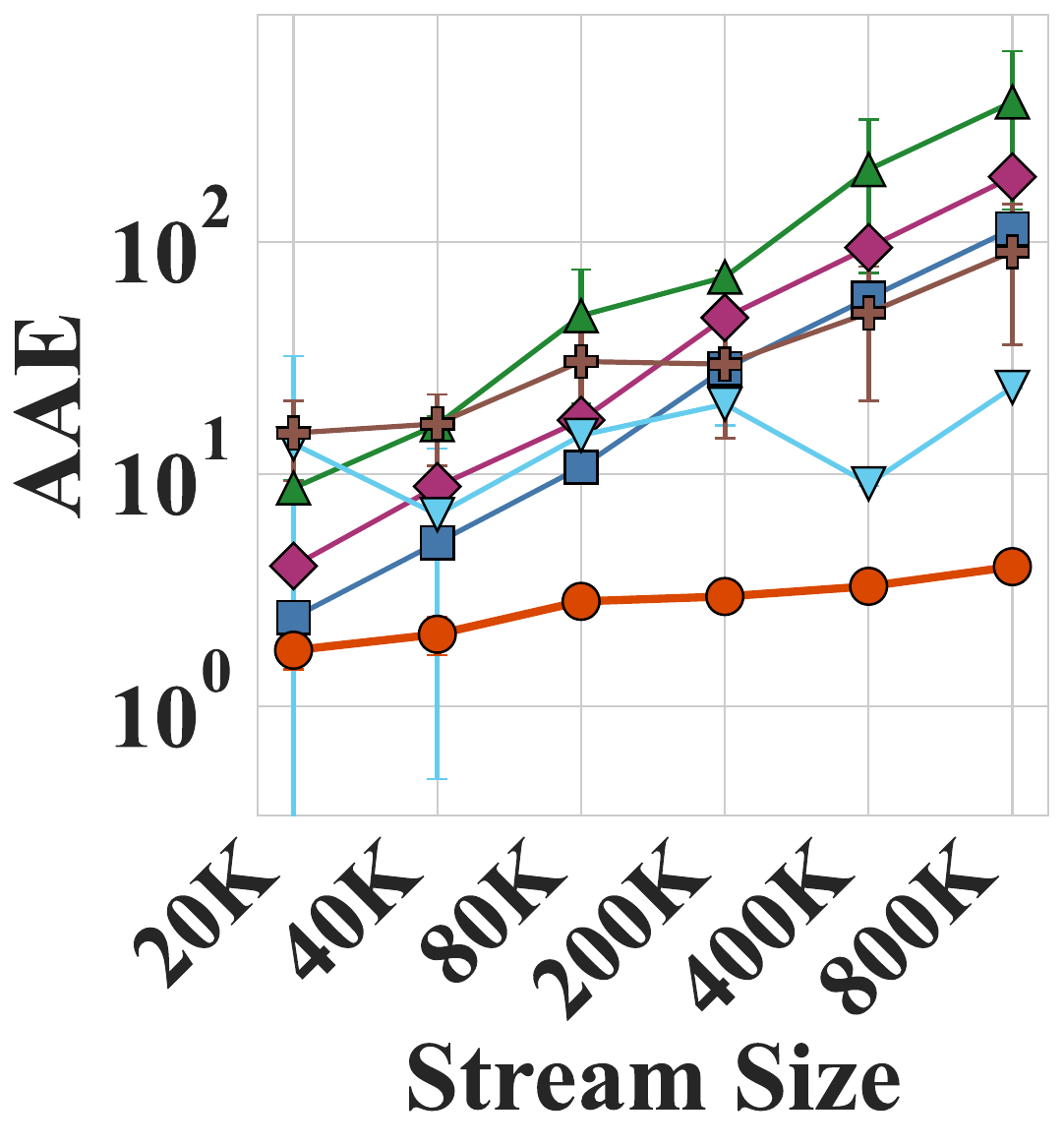}%
		}\hfill%
		\subcaptionbox{\footnotesize NotreDame (AAE)}[0.2\linewidth]{%
			\includegraphics[width=0.19\textwidth]{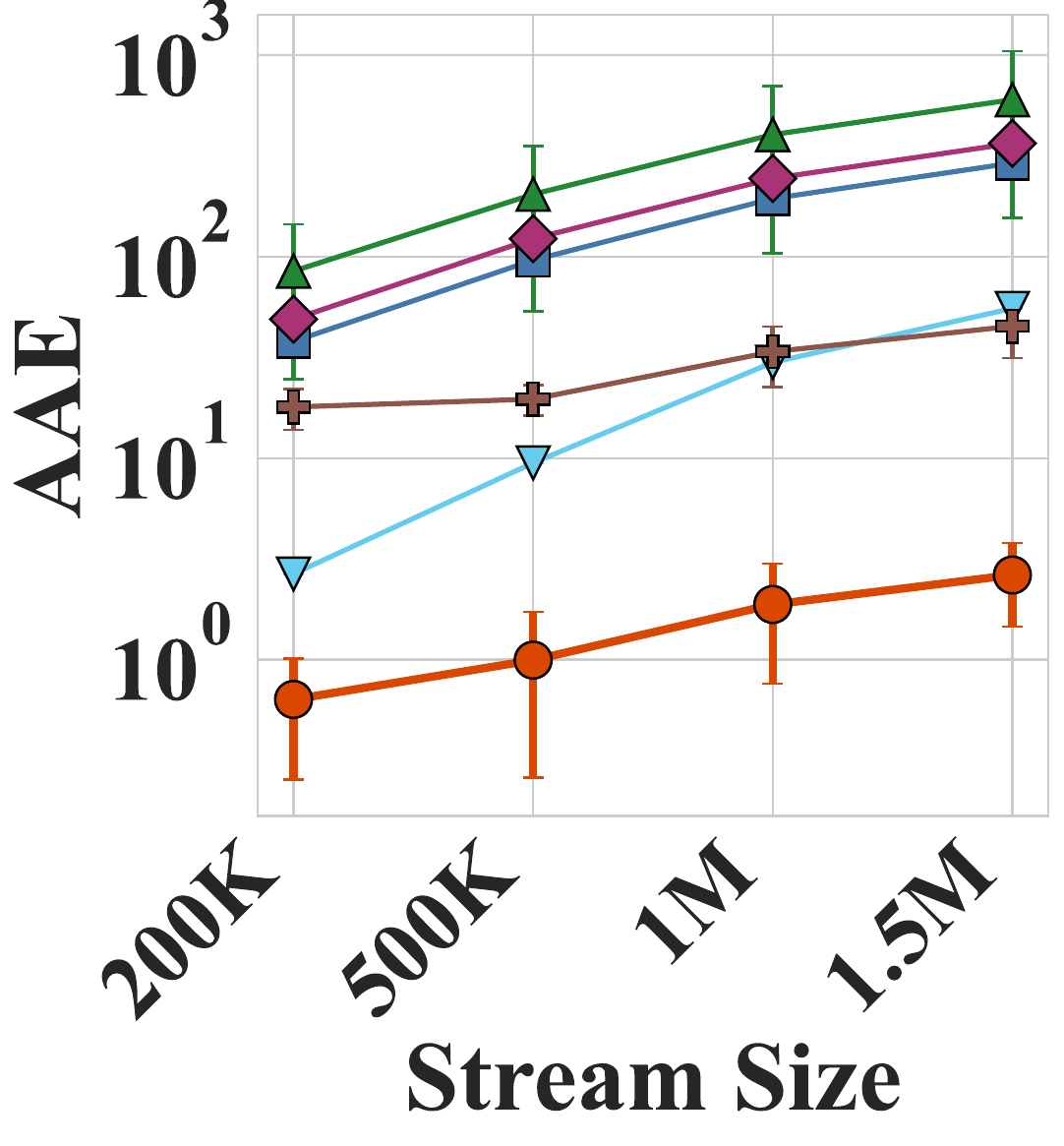}%
		}\hfill%
		\subcaptionbox{\footnotesize CAIDA (AAE)}[0.2\linewidth]{%
			\includegraphics[width=0.19\textwidth]{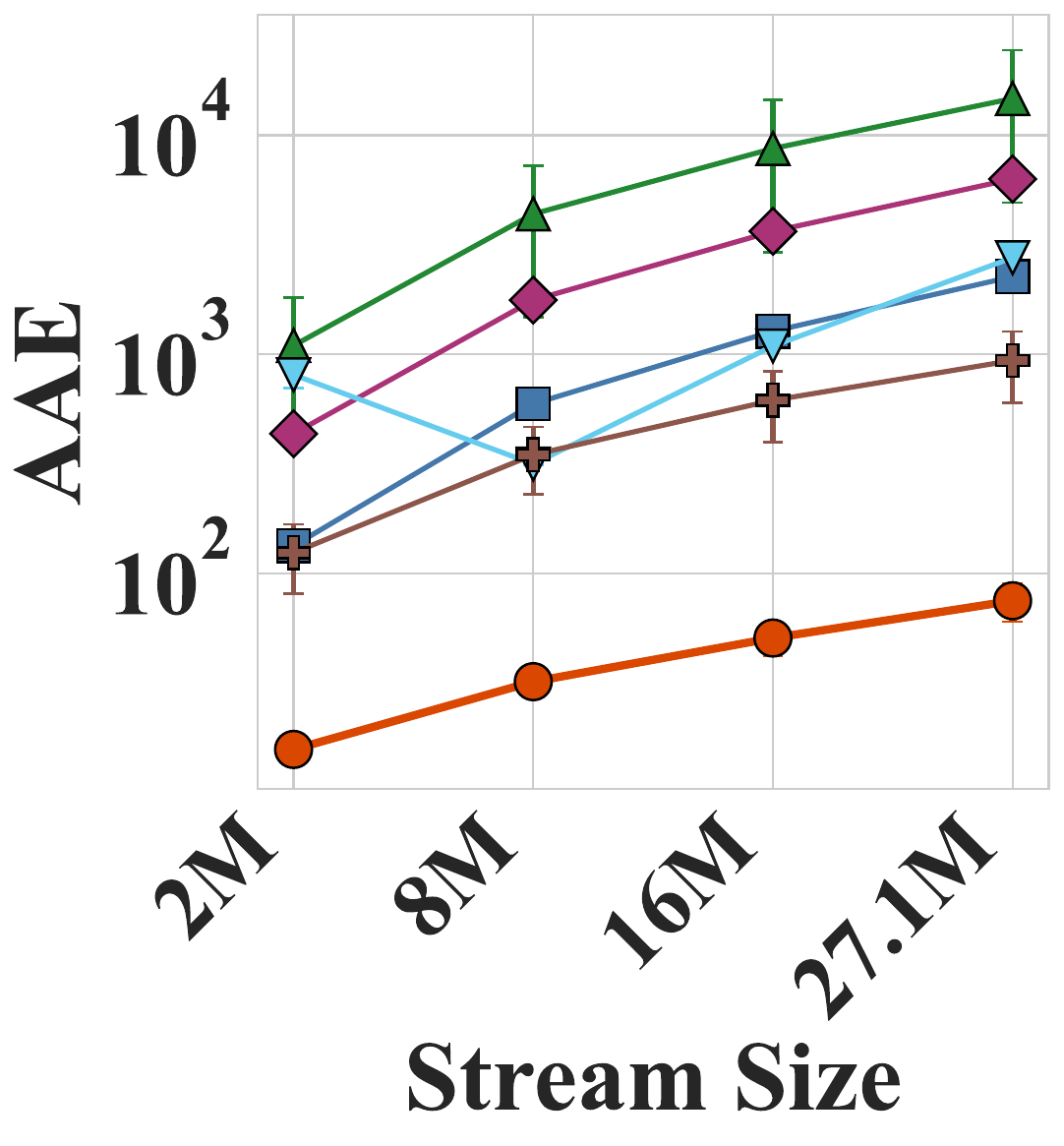}%
		}\hfill%
		\subcaptionbox{\footnotesize WikiTalk (AAE)}[0.2\linewidth]{%
			\includegraphics[width=0.19\textwidth]{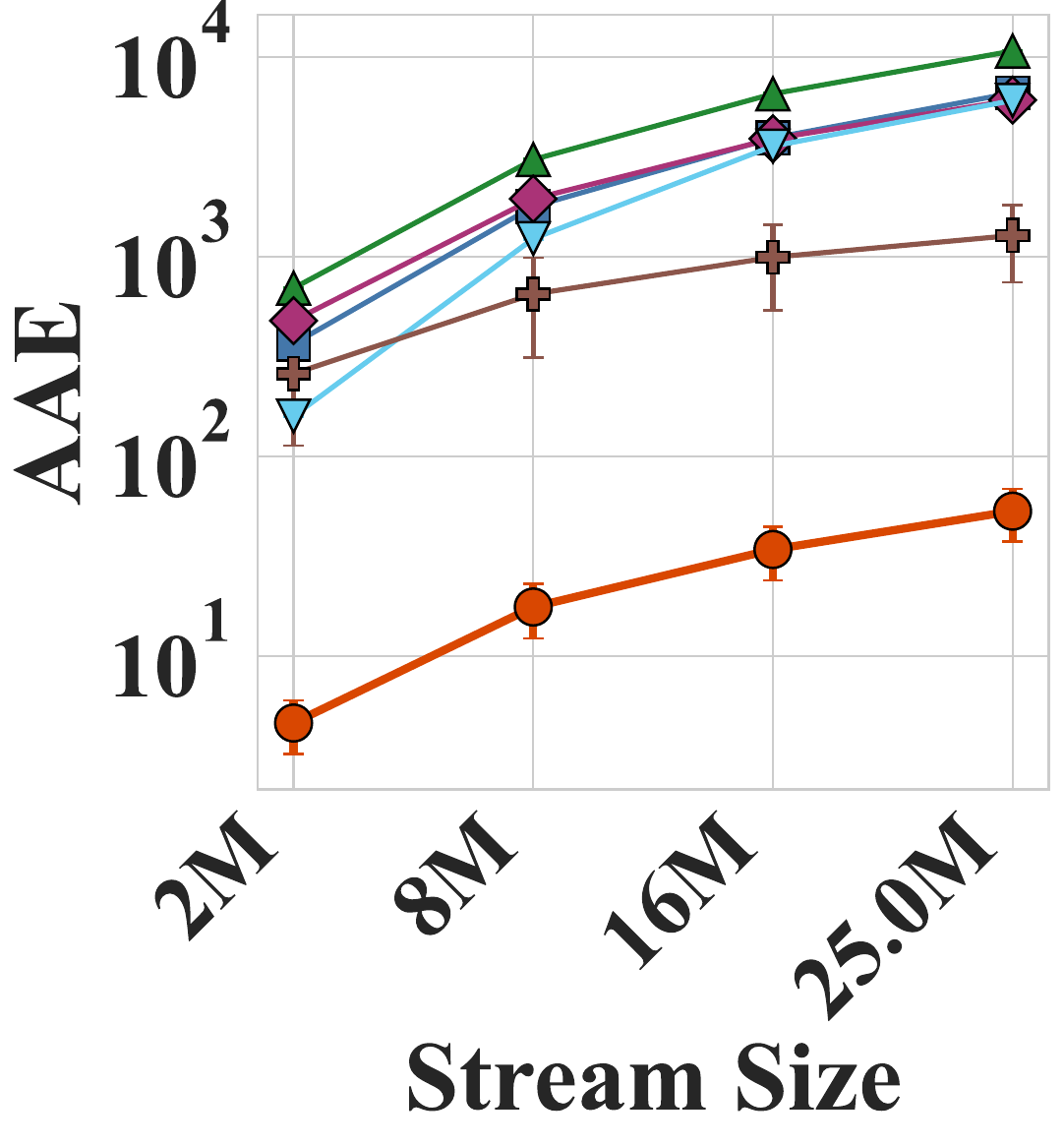}%
		}\hfill%
		\subcaptionbox{\footnotesize StackOverflow (AAE)}[0.2\linewidth]{%
			\includegraphics[width=0.19\textwidth]{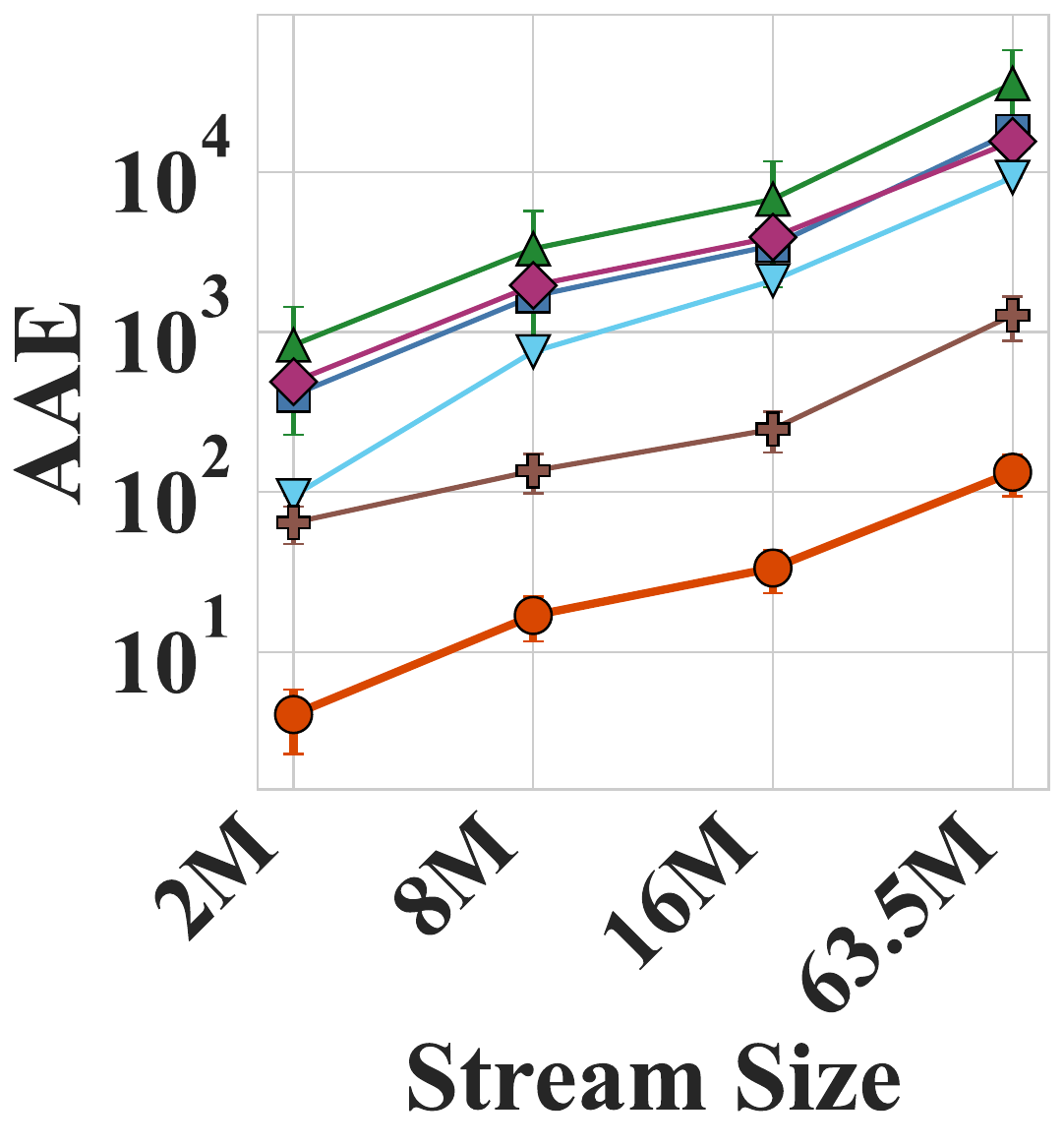}%
		}

		\vspace{-0.2cm}
		\caption{Edge weight estimation accuracy under a fixed 64KB memory budget (top: ARE, bottom: AAE).}
		\label{fig:accuracy}
	\end{minipage}
	\vspace{-0.5cm}
\end{figure*}

\begin{sloppypar}

\subsection{Experimental Setup}
\label{sec:setup}

\bbb{Implementation.} We implement \textit{Crane} and publicly release the source code for reproducibility.
Unless otherwise noted, experiments are conducted on a server equipped with 8 NVIDIA RTX 5090 GPUs (32\,GB memory each), a 128-core Intel(R) Xeon(R) Gold 6459C CPU, and 720\,GB of RAM.
The Flink case study (\S\ref{subsec:system}) is deployed on a separate workstation with a single NVIDIA RTX 5070\,Ti GPU (16\,GB), an AMD Ryzen 9 9900X CPU (12 cores / 24 threads), and 64\,GB of RAM.

\bbb{Baselines. } To ensure a comprehensive and fair comparison under strictly constrained memory budgets, we implement five baseline methods with specific adaptations:
\textit{Mayfly}~\cite{feng2023mayfly}, \textit{TCM}~\cite{tang2016tcm}, 
\textit{GSS}~\cite{gou2022gss}, 
\textit{Auxo}~\cite{jiang2023auxo},
and \textit{HourglassSketch}~\cite{HourglassSketch}.
To evaluate performance under strict resource constraints, we fix the memory budget to 64\,KB for all methods.
For \textit{Crane}, \textit{Auxo}, and \textit{GSS}, which inherently support automatic memory expansion, we explicitly enforce this 64\,KB cap as the hard upper limit during testing.
For the neural methods \textit{Mayfly} and \textit{Crane}, we do not include model parameters in the memory budget, since their parameter size is small, fixed after training, and independent of the stream size.

\bbb{Datasets.} We evaluate all methods on five widely used real-world datasets: Lkml \cite{nr}, CAIDA, WiKiTalk, StackOverflow \cite{paranjape2017motifs}, and NotreDame~\cite{jeong1999diameter}, 
described below.
These datasets are standard benchmarks in graph stream summarization, covering IP traces, online Q\&A forums, collaborative editing platforms, and web graphs.
\begin{itemize}[leftmargin=1em]

\item Lkml \cite{nr}\footnote{\url{http://konect.cc/networks/lkml-reply/}} is an email-reply network from the Linux Kernel Mailing List, where each edge represents one user replying to another. 
We follow the stream construction in Mayfly~\cite{feng2023mayfly} and treat time-ordered replies as a directed edge stream with unit weights.

\item CAIDA\footnote{\url{http://www.caida.org/data/passive/passive_dataset.xml}\label{fn:caida}} is an IP stream dataset derived from the CAIDA 2018 UCSD Anonymized Internet Traces.
We use one minute of traffic containing $27{,}121{,}713$ packets and $850{,}200$ distinct flows, where each flow is modeled as a directed edge between 32-bit encoded IPv4 addresses.

\item WiKiTalk\footnote{\url{https://networks.skewed.de/net/wiki_talk}} is a user–user communication network built from Wikipedia talk pages, where an edge from $u$ to $v$ means that user $u$ wrote a message on user $v$'s talk page.
We use the English subset from this collection, sort all interactions by their timestamps, and treat them as a directed edge stream.
The resulting stream contains $24{,}981{,}163$ edges and $9{,}379{,}561$ distinct user pairs.

\item StackOverflow~\cite{paranjape2017motifs}\footnote{\url{https://snap.stanford.edu/data/sx-stackoverflow.html}} is a user-user interaction graph from the SNAP collection.
Each temporal interaction (question, answer, or comment) is treated as a directed edge, and we process all edges in chronological order as a graph stream.
The resulting stream contains $63{,}497{,}050$ edges and $36{,}233{,}450$ distinct user pairs.

\item NotreDame~\cite{jeong1999diameter}\footnote{\url{https://snap.stanford.edu/data/web-NotreDame.html}} is a web graph from the \texttt{nd.edu} domain.
Nodes are web pages and edges are hyperlinks; we generate a graph stream by taking a random permutation of all directed edges.
In this dataset, each edge appears exactly once, so it does not exhibit the repeated-updates pattern of typical graph streams.
We mainly use it to evaluate how well different sketches handle hash collisions under a large number of distinct edges.
The stream comprises $1{,}497{,}134$ directed edges, all distinct.

\item MAWI\footnote{\url{https://mawi.wide.ad.jp/mawi/}} is a backbone IPv4 trace from the MAWI Working Group Traffic Archive ($1{,}205{,}795{,}081$ edge updates over $188{,}535{,}913$ distinct edges), modeled as for CAIDA and used only for the large-budget sweep of \S\ref{sec:accuracy}.

\end{itemize}

For each dataset, we further construct subsets of different stream lengths by randomly sampling contiguous segments, as summarized in Figure~\ref{fig:accuracy}.

\bbb{Metrics.}
To quantify estimation accuracy, we employ two standard metrics: Average Absolute Error (AAE) and Average Relative Error (ARE).
They are defined as:
\begin{equation}
\setlength\abovedisplayskip{-0.01cm}
\setlength\belowdisplayskip{-0.01cm}
    \text{AAE} = \frac{1}{|\mathcal{Q}|} \sum_{e \in \mathcal{Q}} \bigl|\hat{y}(e) - f(e)\bigr|,
\end{equation}
\begin{equation}
\setlength\abovedisplayskip{-0.01cm}
\setlength\belowdisplayskip{-0.01cm}
    \text{ARE} = \frac{1}{|\mathcal{Q}|} \sum_{e \in \mathcal{Q}} \frac{\bigl|\hat{y}(e) - f(e)\bigr|}{f(e)},
\end{equation}
where $\mathcal{Q}$ denotes the query set, and $f(e)$ and $\hat{y}(e)$ represent the true cumulative weight and the estimated value returned by the sketch, respectively.
In all experiments, we query every distinct edge that appears in the evaluated stream segment exactly once, $i.e.$, $\mathcal{Q}$ is the set of all observed edges in that segment.
This equal-weight protocol is the standard in graph-stream summarization~\cite{tang2016tcm,gou2022gss,jiang2023auxo,feng2023mayfly} and presumes no correlation between query popularity and edge weight; \appref{app:weighted} repeats the evaluation under weight-biased workloads in which heavy or tail edges are queried more often.

\bbb{Parameters.}
Across all Crane experiments, we set the synthetic stream length $\Gamma$ up to $600{,}000$, ensuring the model is exposed to sufficient accumulated collisions and long-term dependencies that mimic large-scale real-world benchmarks.
We build the distribution pool $\mathcal{P}$ using Zipf distributions whose skewness parameter $\alpha$ ranges from $0.3$ to $0.8$, and set the total weight sum for each task to range from 5 to 50 times the number of edges in the graph.
Training proceeds over $1{,}000$ tasks, each visited exactly once, optimized with AdamW~\cite{loshchilov2017decoupled} whose learning rate is annealed from $10^{-3}$ to $10^{-6}$ via a cosine schedule.
Regarding the model architecture, each layer-specific encoder $f_i^{(o)}, f_i^{(d)}$ in Crane is implemented as a Multi-Layer Perceptron (MLP) that maps the 32-dimensional binary representation space into a 64-dimensional embedding space:
(1) a linear layer ($32 \to 32$) followed by Batch Normalization (BN) and ReLU activation;
(2) an intermediate linear layer expanding the dimension ($32 \to 48$) followed by BN and ReLU;
and (3) a final linear projection ($48 \to 64$) activated by Softplus.
These widths follow a capacity trade-off governed by the narrowest layer, denoted $h_{\min}$ ($h_{\min} = 32$ above), which caps how well the encoder keeps the binary codes separated in embedding space.
In our experiments, $h_{\min} < 8$ degrades accuracy sharply (roughly $30\%$ higher ARE at $h_{\min} = 4$ than at $8$) as the bottleneck itself begins to induce collisions, while $h_{\min} > 32$ yields diminishing returns (roughly $10\%$ lower ARE at $128$), the residual error being dominated by collisions in the $H \times W$ memory.
Because the encoders run on the critical path of every insertion and query, we keep $h_{\min}$ within $8$--$32$; substantially larger node universes may warrant wider encoders to keep their node codes separated.
For the Mayfly baseline, we adopt the training parameters recommended by the original authors~\cite{feng2023mayfly}.

\vspace{-0.1in}
\subsection{Experiments on Accuracy}
\label{sec:accuracy}

We evaluate the accuracy of \textit{Crane} on the fundamental task of edge weight estimation against five baselines on real-world datasets, covering distinct traffic patterns ranging from heavy-tailed network flows to uniform web graphs.
Figure~\ref{fig:accuracy} reports the Average Relative Error (ARE) across varying stream lengths under a fixed 64KB memory budget.

\bbb{Heavy-Tailed Traffic.}
The most significant challenges in our evaluation are posed by CAIDA, WiKiTalk, and StackOverflow, which exhibit highly skewed distributions and bursty interaction patterns.
On CAIDA, flat sketches (TCM, GSS, Auxo) rapidly saturate, with errors exploding by orders of magnitude as the stream length grows from 2M to 27.1M.
This collapse confirms the infrequent flow interference problem: when heavy hitters share the same counters as rare edges, the resulting noise is catastrophic.
\textit{Mayfly} remains the strongest baseline at most scales (with \textit{HourglassSketch} briefly overtaking it at 8M), yet the best-performing prior sketch still yields ARE ranging from $65$ at 2M to $443$ at 27M.
In stark contrast, \textit{Crane} maintains an ARE of $2.29\sim28.70$ even as the stream volume grows by over $13\times$, reducing the strongest baseline's error by $15\times\sim29\times$.
Similarly, on the social interaction graphs WiKiTalk and StackOverflow, Crane outperforms the strongest baseline by a factor of roughly $7\times \sim 34\times$.
This robustness validates our core hypothesis: the hierarchical carry mechanism successfully isolates frequent flows in upper layers, preventing them from cannibalizing the capacity required for the ``long tail'' in lower layers.

\bbb{Distinct Edge Stress Test.}
The NotreDame dataset represents a fundamentally different regime: it consists entirely of distinct edges with unit weights, eliminating heavy hitters and isolating the sketch's ability to handle pure hash collisions.
In this setting, every collision contributes directly to estimation error.
Baselines struggle significantly here; hash-based methods such as TCM and Auxo saturate outright (ARE growing from tens to hundreds).
\textit{HourglassSketch}'s error-funnel mechanism keeps its error in check at smaller scales (ARE $2.66$ at 200K) but degrades rapidly as the number of distinct edges grows, reaching $55.4$ at 1.5M—at which point \textit{Mayfly} takes over as the strongest baseline (ARE $45.2$).
\textit{Crane}, by contrast, achieves an ARE below $3.0$ even at 1.5M distinct edges, improving over the strongest baseline at each scale by $4\times \sim 17\times$.
This indicates that the hierarchical architecture is not solely dependent on skewness for its advantage.
By distributing the load across multiple layers via learnable encoders, Crane effectively reduces the collision probability for any given edge, even in the absence of frequency disparity.

\bbb{Moderate Scale.}
On the Lkml email corpus, which represents a more moderate scale and skewness, the best baseline shifts from \textit{TCM} at the smallest scale (20K) to \textit{HourglassSketch} from 40K onward, where its two-stage design effectively captures the heavy hitters of this workload.
Even so, \textit{Crane} continues to dominate, suppressing error rates to sub-unit levels (ARE $<1.0$) at all but the largest stream length.
The fact that Crane maintains a clear margin (improving by $2\times \sim 16\times$ over the strongest baseline at each scale) even on this ``easier'' workload suggests that its modeling capacity does not overfit to extreme conditions.
Instead, it yields genuine gains in collision resolution and calibration that persist across varying scales.

\bbb{Summary.}
Across all datasets and stream lengths, \textit{Crane} reduces ARE by roughly an order of magnitude over the strongest prior sketch—whether neural (\textit{Mayfly}) or handcrafted (\textit{HourglassSketch})—and its error grows substantially slower with stream length, indicating that the hierarchical architecture and automatic expansion together yield a robust, dataset-agnostic summarization logic rather than memorization of workload patterns.

\begin{figure}[t]
	\centering
	\captionsetup{justification=raggedright}
	\begin{minipage}[t]{0.40\linewidth}
		\centering
		\includegraphics[width=\linewidth]{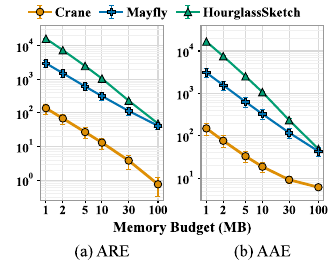}

		\vspace{-0.3cm}
		\caption{Large-budget sweep on the billion-edge MAWI trace; mean $\pm$ s.d.\ over random seeds.}
		\label{fig:mawi}
	\end{minipage}\hfill
	\begin{minipage}[t]{0.57\linewidth}
		\centering
		\includegraphics[width=\linewidth]{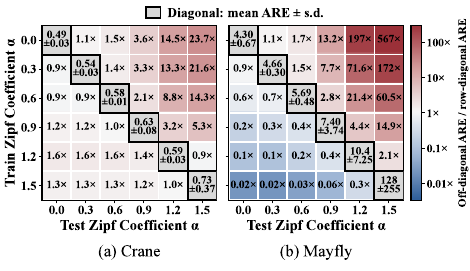}

		\vspace{-0.3cm}
		\caption{Cross-training across skewness levels: models trained at a single Zipf exponent $\alpha_{\text{train}}$ (rows), evaluated on all six synthesized Zipf datasets $\alpha_{\text{test}}$ (columns). Diagonal: matched-distribution ARE (mean $\pm$ s.d.); off-diagonal: ratio to its own row's diagonal.}
		\label{fig:crosstrain}
	\end{minipage}
	\vspace{-0.3cm}
\end{figure}

\bbb{Large Memory Budgets.}
The 64\,KB above is the \emph{per-unit} budget of the constrained regime that sketches target; larger aggregate memory is obtained by hash-partitioning the stream across independent units, as our Flink deployment does (\S\ref{subsec:system}).
To check that Crane's advantage survives at deployment-scale memory, we evaluate on MAWI, a backbone IPv4 trace of $1.21$ billion edge updates over $188.5$ million distinct edges, and sweep the total budget from 1\,MB to 100\,MB, spending it as $K = \text{budget}/64\,\text{KB}$ hash-routed instances ($K = 16$ to $1600$) in this configuration.
We compare against \textit{Mayfly} and \textit{HourglassSketch}, the strongest neural and handcrafted baselines of Figure~\ref{fig:accuracy}.
Figure~\ref{fig:mawi} shows the margin widening rather than narrowing as memory grows: Crane improves on the strongest baseline by $22\times$ at 1\,MB and by $55\times$ at 100\,MB, where it is the only method to reach sub-unit ARE.
The gain is concentrated in the tail.
Over the bottom $50\%$ of queries by cumulative weight Crane leads by up to $57\times$ across the sweep---consistent with the purpose of the hierarchy, which is to protect the long tail from the heavy hitters.

\begin{figure*}[t]
	\centering
	\begin{minipage}{\linewidth}
		\centering
		\includegraphics[width=0.75\linewidth]{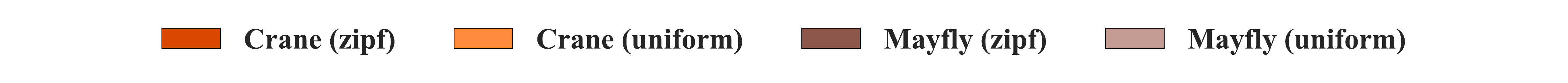}
		\vspace{-0.1cm}
	\end{minipage}
	\begin{minipage}{1 \linewidth}
        \captionsetup[subfigure]{skip=-0.02cm, labelfont+=footnotesize}
		\centering
		
		\subcaptionbox{\footnotesize Lkml (ARE)}[0.2\linewidth]{%
			\includegraphics[width=0.19\textwidth]{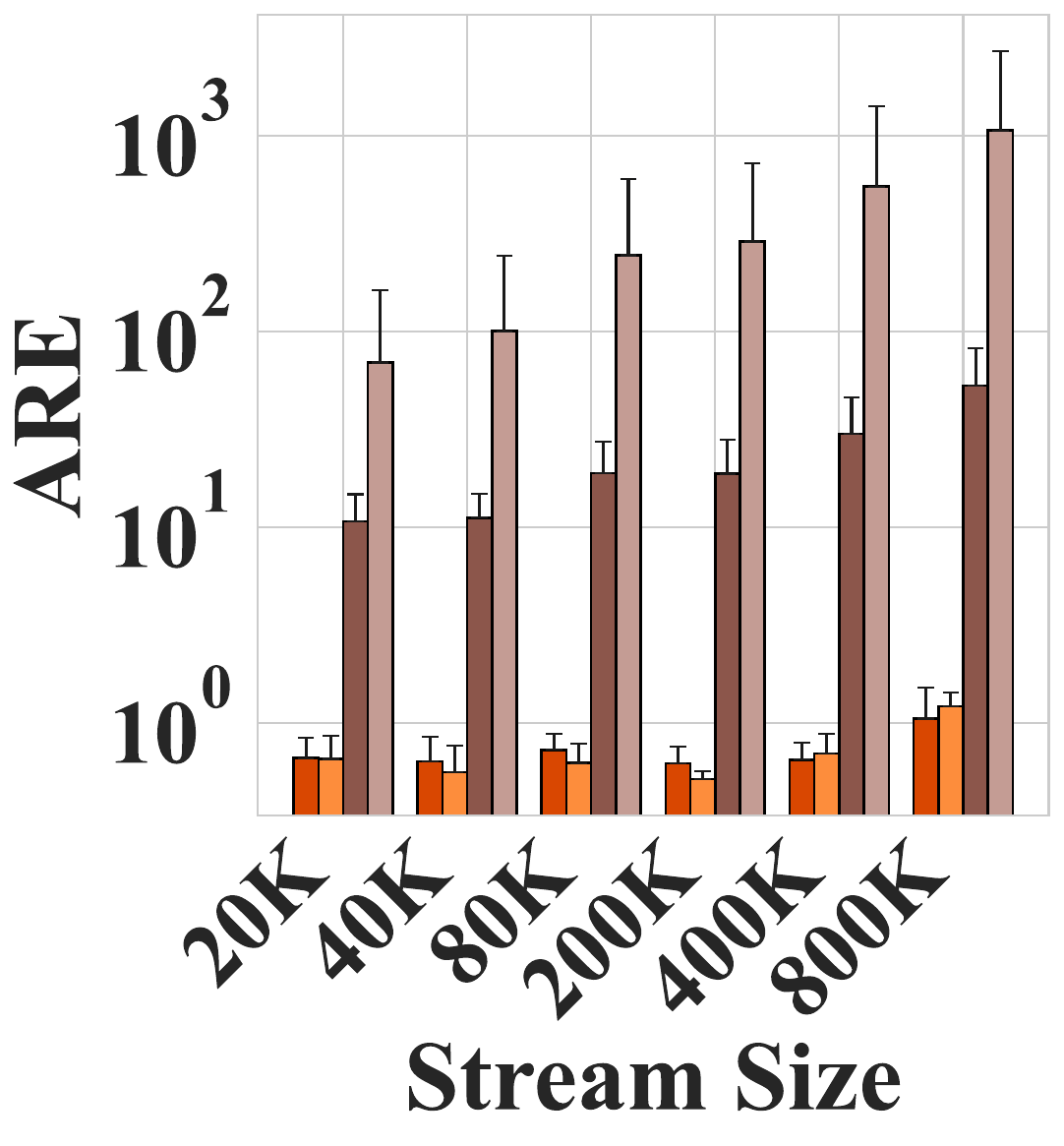}%
		}\hfill%
		\subcaptionbox{\footnotesize NotreDame (ARE)}[0.2\linewidth]{%
			\includegraphics[width=0.19\textwidth]{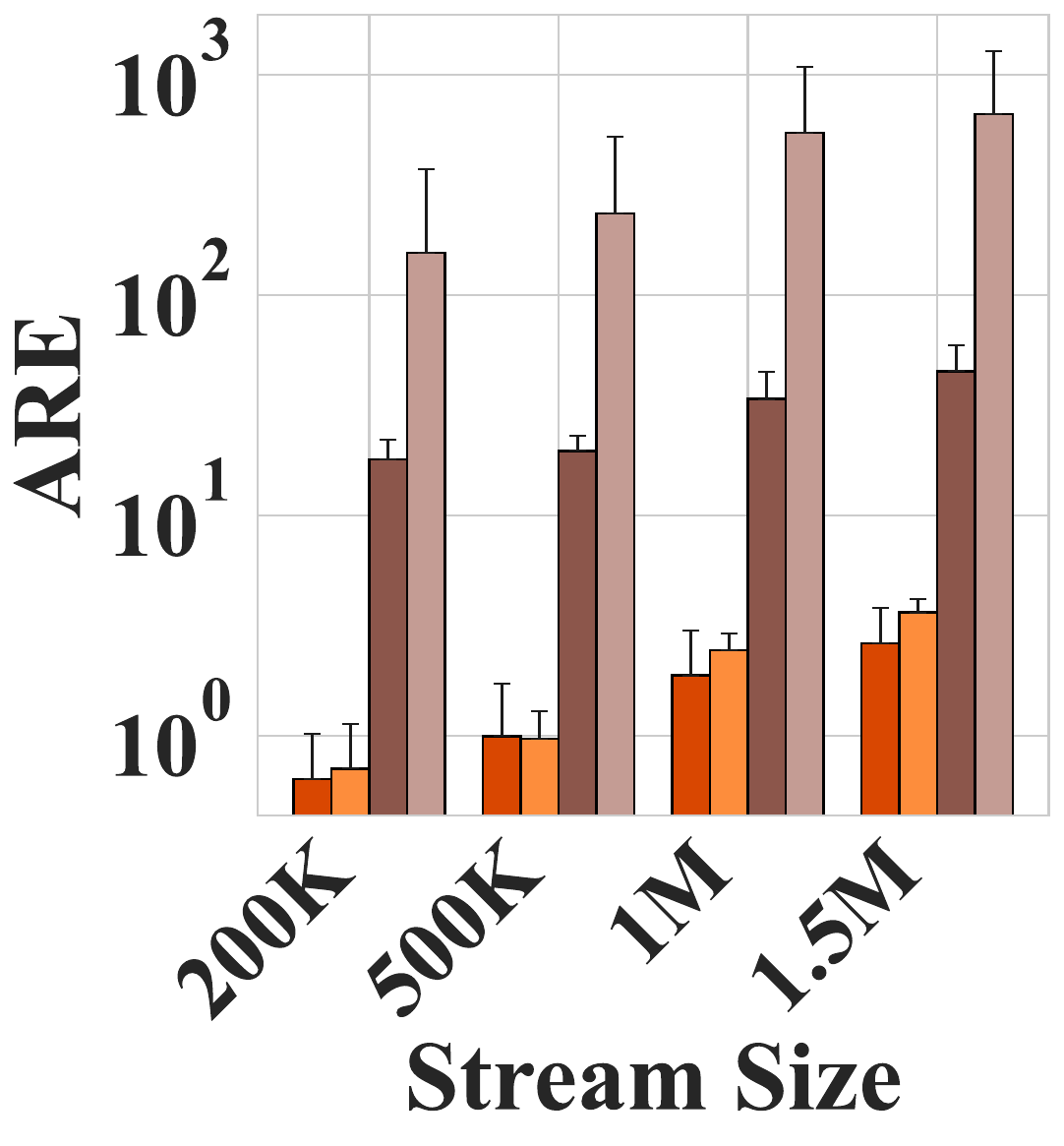}%
		}\hfill%
		\subcaptionbox{\footnotesize CAIDA (ARE)}[0.2\linewidth]{%
			\includegraphics[width=0.19\textwidth]{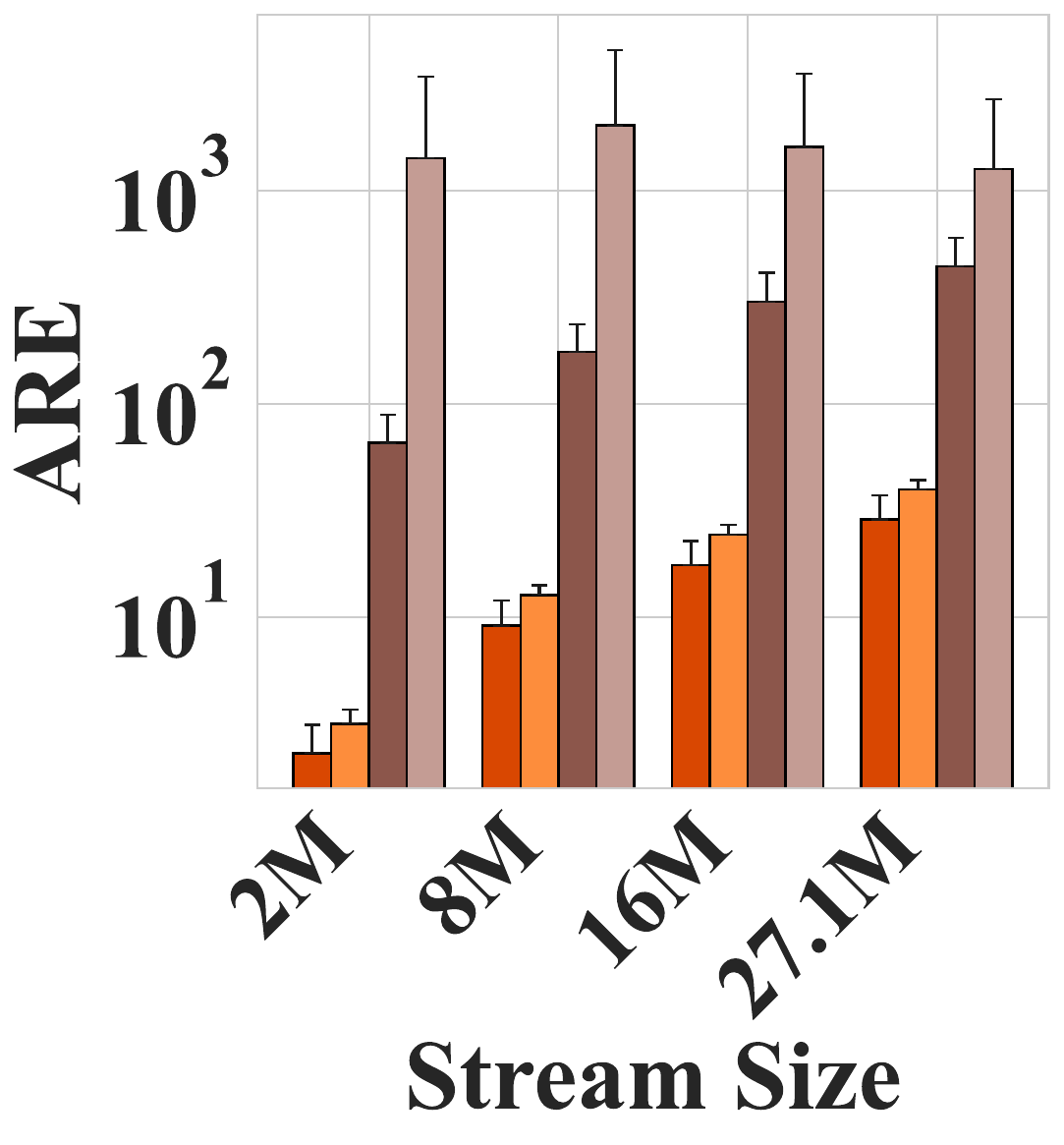}%
		}\hfill%
		\subcaptionbox{\footnotesize WikiTalk (ARE)}[0.2\linewidth]{%
			\includegraphics[width=0.19\textwidth]{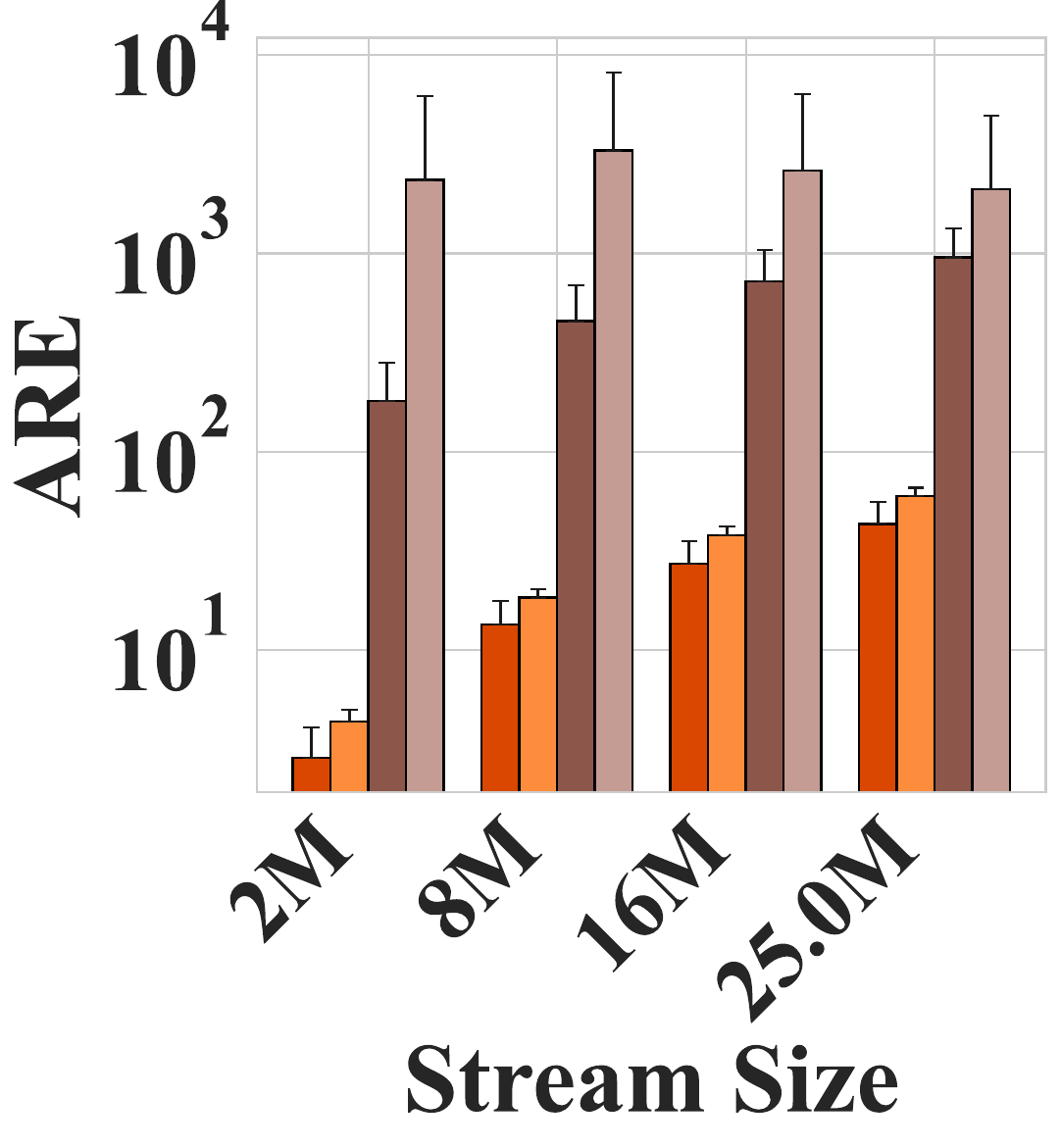}%
		}\hfill%
		\subcaptionbox{\footnotesize StackOverflow (ARE)}[0.2\linewidth]{%
			\includegraphics[width=0.19\textwidth]{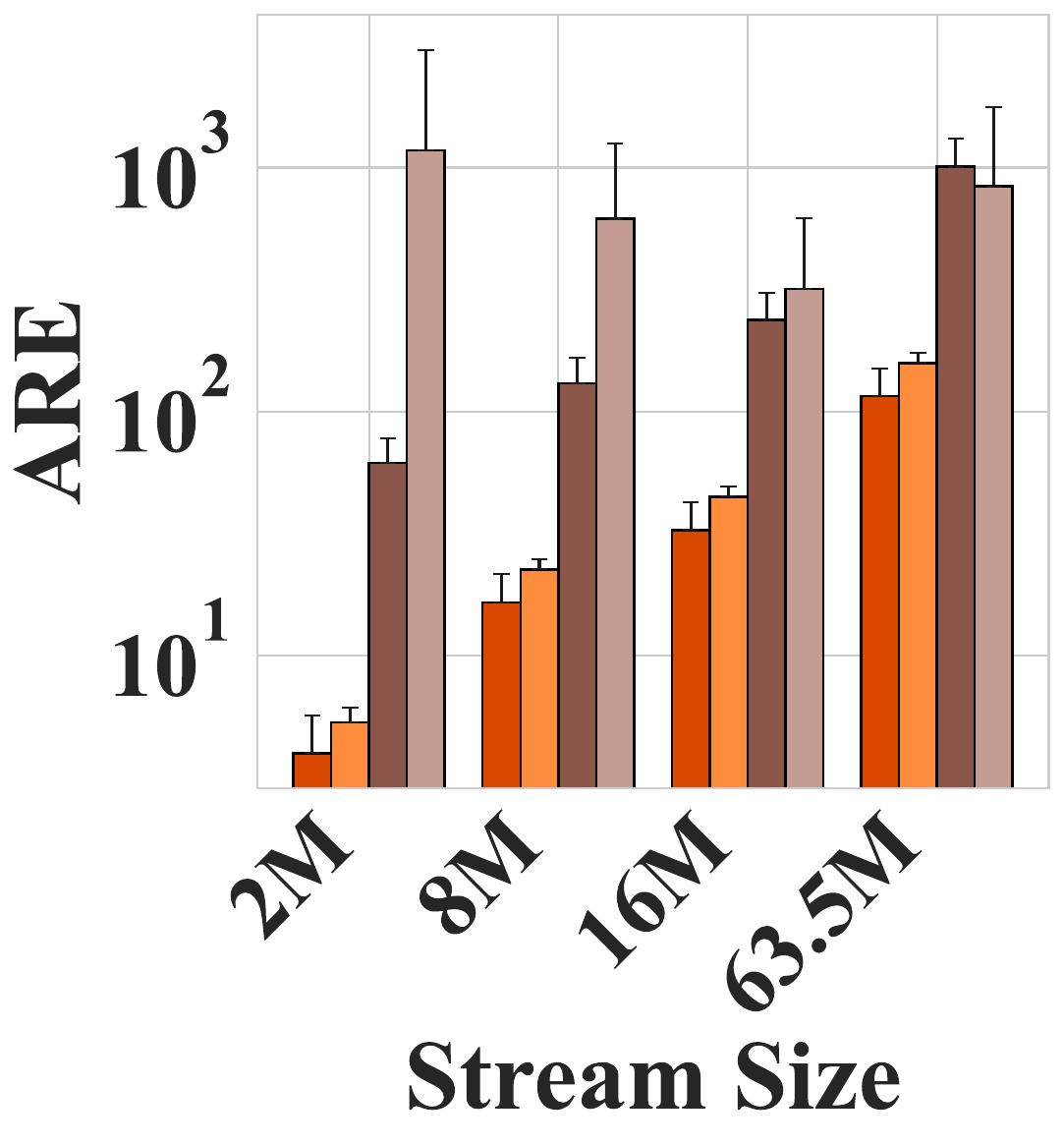}%
		}

		\vspace{-0.2cm}
		\caption{Robustness to training distribution: ARE of models trained on Uniform vs.\ Zipfian weight distributions.}
		\label{fig:robustness}
	\end{minipage}
	\vspace{-0.2cm}
\end{figure*}

\vspace{-0.1in}
\subsection{Experiments on Robustness}
\label{sec:robustness}

Crane is trained on synthetic streams, so a mismatch between the training distribution and a real workload can affect its accuracy; the following experiments quantify the cost of this synthetic-to-real gap.
We probe the gap from two directions: on the five real-world datasets, we evaluate models pre-trained on Uniform weight distributions (Zipf exponent $\alpha = 0$) against the default Zipf-trained models (Figure~\ref{fig:robustness}).
Measuring the accuracy lost under a shift of \emph{given} magnitude further requires control over both distributions, so the second experiment synthesizes the test streams as well: Crane and Mayfly models trained at each of six Zipf exponents $\alpha \in [0.0, 1.5]$ are evaluated on all six test distributions (Figure~\ref{fig:crosstrain}).
We vary the Zipf exponent rather than the distribution family: Zipf is the accepted model of stream skew, and the alternative families used in prior work lie close to it, so swapping families would change the test stream little and by an amount that is hard to quantify, whereas the exponent spans a wide range of skew with an exactly known train--test gap.

\bbb{Resilience to Distribution Mismatch.}
\textit{Mayfly} is sensitive to the prior: switching its training from Zipf to Uniform inflates its error by up to an order of magnitude on the heavy-tailed streams.
\textit{Crane}'s swing stays an order of magnitude smaller, and on the most skewed datasets (CAIDA, StackOverflow) the Uniform-trained \textit{Crane} still beats the Zipf-trained \textit{Mayfly} by a wide margin, suggesting that the hierarchy itself is robust to the choice of training prior.
On NotreDame, the Zipf-trained model wins although every edge is distinct with unit weight and there is no frequency disparity for the hierarchy to exploit.
This is the expected behavior once the prior is read as the \emph{weight range} the training streams cover rather than the shape of their distribution: cumulative weight $1$ lies inside the range the Zipfian streams reach routinely and outside the Uniform ones (minimum above $6$, mean above $26$), and retraining the Uniform model with unit weights reverses the ordering (ARE $2.63 \to 1.77$).
For the same reason, Mayfly's Uniform variant degrades more sharply still.
On the datasets evaluated here, the gap thus lies in skew and weight range, and on NotreDame it is coverage of the weight range that governs transfer.

\bbb{Cross-Training across Skewness Levels.}
Crane attains the lower ARE in all 36 cells of the grid, and a mismatched prior costs it at most $23.7\times$ its matched value, against $567\times$ for Mayfly.
The penalty is asymmetric in the same direction as on NotreDame: training on a prior heavier-tailed than the test stream costs far less accuracy than the reverse---so the training pool should include distributions heavier-tailed than the expected deployment workload.
Read as an out-of-distribution check, the off-diagonal cells bound the accuracy lost when a deployed stream drifts away from the training distribution; should the drift be severe enough to require retraining, the \emph{Training Cost} paragraph of \S\ref{sec:throughput} reports its measured cost.

\vspace{-0.07in}
\subsection{Ablation and Parameter Sensitivity}
\label{sec:ablation}
\vspace{-0.03in}

To rigorously quantify the contribution of each design component, we conduct ablation studies focusing on the depth of the hierarchical memory and the dynamics of the carry mechanism.
Table~\ref{tab:ablation_combined} summarizes the Average Relative Error (ARE) with the number of memory layers varied, while Figure~\ref{fig:param_sensitivity} details the sensitivity to the carry threshold and mini-batch size.

\begin{table}[h]
\centering
\vspace{-0.1cm}
\caption{Memory layer ablation study (ARE, mean $\pm$ std over random seeds): hierarchical ($\text{ML}=4$) vs.\ flat ($\text{ML}=1$).}
\vspace{-0.3cm}
\label{tab:ablation_combined}
\begin{tabular}{lccc}
\toprule
Dataset & ML=1 & ML=4 & Reduction \\
\midrule
Lkml & 3.57$\pm$3.42 & 0.72$\pm$0.02 & 5.0$\times$ \\
NotreDame & 12.50$\pm$13.34 & 1.53$\pm$0.72 & 8.2$\times$ \\
CAIDA2018 & 111.83$\pm$111.47 & 14.40$\pm$4.29 & 7.8$\times$ \\
WiKiTalk & 165.50$\pm$165.57 & 21.69$\pm$6.63 & 7.6$\times$ \\
StackOverflow & 293.65$\pm$288.58 & 42.20$\pm$12.66 & 7.0$\times$ \\
\bottomrule
\end{tabular}
\vspace{-0.2cm}
\end{table}

\bbb{Effectiveness of Hierarchical Memory.}
The core premise of Crane is that separating edges vertically across layers reduces interference between heavy hitters and the long tail.
To validate this, we compare the standard Crane configuration ($N=4$ layers) against a collapsed ``flat'' variant ($N=1$), where all capacity is forced into a single layer with equivalent total parameters.

Table~\ref{tab:ablation_combined} shows that the hierarchical structure provides decisive benefits in both accuracy and stability.
Collapsing the hierarchy inflates ARE by roughly $7\times$ on average, with the degradation most pronounced on the heavily skewed datasets.
Equally striking is the collapse in estimator stability: with a single layer the run-to-run standard deviation is essentially equal to the mean, meaning a flat sketch produces near-random estimates depending on the random seed, whereas four layers yield a tight variance around the mean.
This confirms the limitation of flat neural sketches (\S\ref{sec:intro}): without vertical separation, the learnable encoders cannot disentangle the overwhelming gradients of frequent edges from the subtle signals of rare ones, and training converges to widely different local optima across seeds.
The training traces of these runs locate the instability in the optimization itself: with everything but the seed held fixed, the flat variant's final validation ARE spans $3.1$--$39.9$---a $13\times$ spread---while the four-layer model stays within $0.64$--$1.29$.
The flat runs do not merely differ in accuracy; each converges to an output regime set by its seed, ranging from systematically negative estimates to a uniform inflation above $20\times$ the true mean, and in $48\%$ of the flat runs' epochs the mean prediction is negative on the training or the validation pass, against $3\%$ for the hierarchical runs.
These anomalies confirm the source of the near-random behavior that the standard deviations of Table~\ref{tab:ablation_combined} record: the flat sketch provides too little usable training signal, and the optimization collapses.

A loss-swap ablation in \appref{app:loss} rules out the training objective as the source of this advantage: exchanging Crane's MAE loss for Mayfly's balanced MSE keeps ARE within the range the default spans across training seeds.

\begin{figure}[t]
	\centering
	\captionsetup{justification=raggedright}
	\captionsetup[subfigure]{skip=-0.02cm}
	\begin{minipage}[b]{0.49\linewidth}
		\centering
		\includegraphics[width=0.95\linewidth]{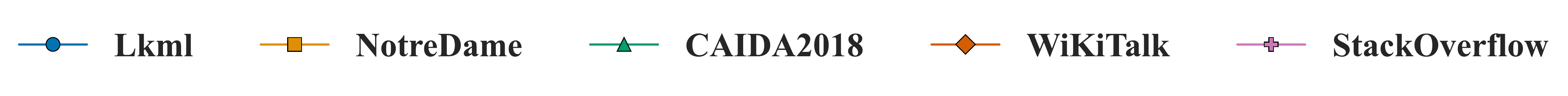}

		\subcaptionbox{\small Carry Threshold\label{fig:param-threshold}}{%
			\includegraphics[width=0.49\linewidth]{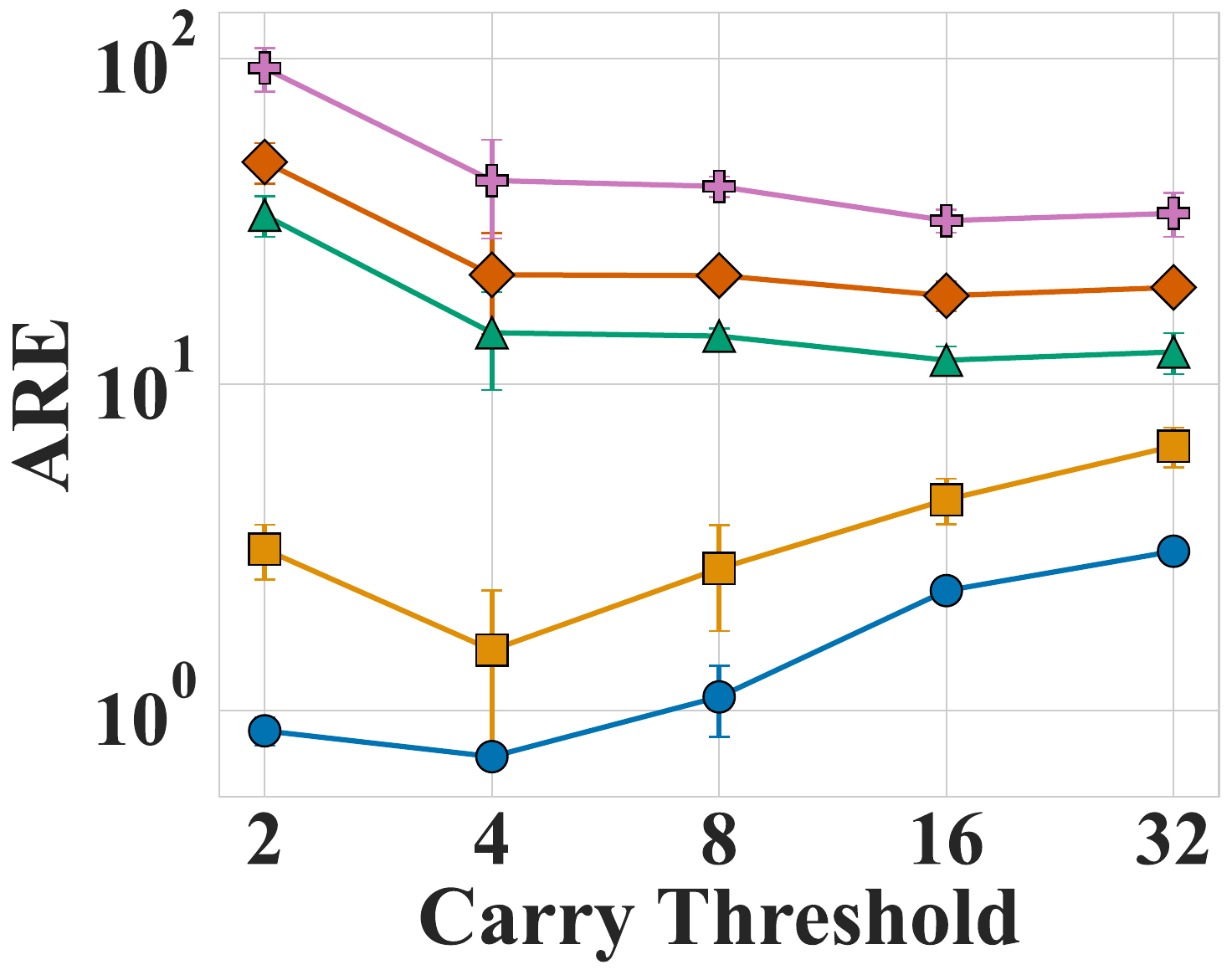}%
		}\hfill%
		\subcaptionbox{\small Mini-Batch Size\label{fig:param-batch}}{%
			\includegraphics[width=0.49\linewidth]{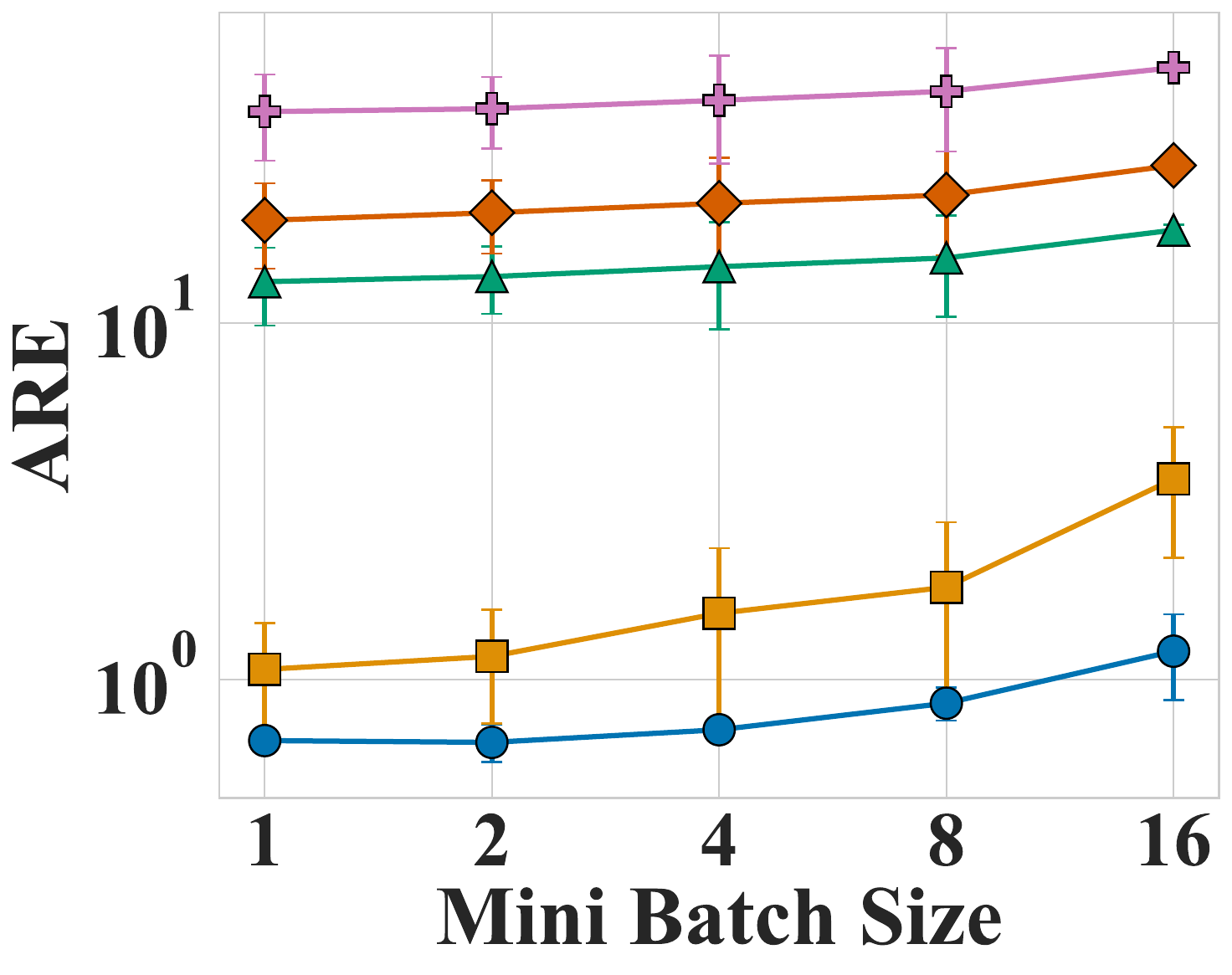}%
		}

		\vspace{-0.3cm}
		\caption{Accuracy (ARE) of \textit{Crane} under varying carry threshold $\theta$ and mini-batch size $b$.}
		\label{fig:param_sensitivity}
	\end{minipage}\hfill
	\begin{minipage}[b]{0.49\linewidth}
		\centering
		\includegraphics[width=0.95\linewidth]{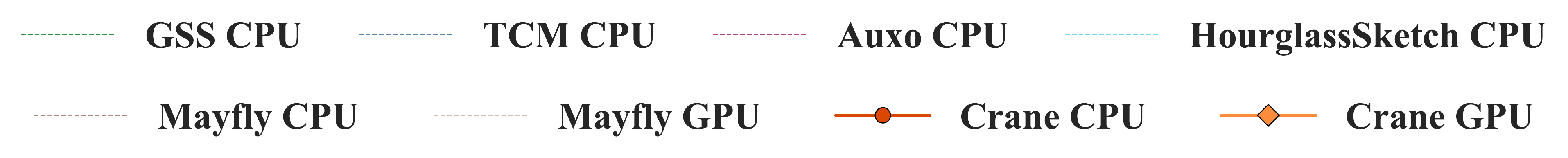}

		\subcaptionbox{\small Store Throughput\label{fig:throughput-store}}{%
			\includegraphics[width=0.49\linewidth]{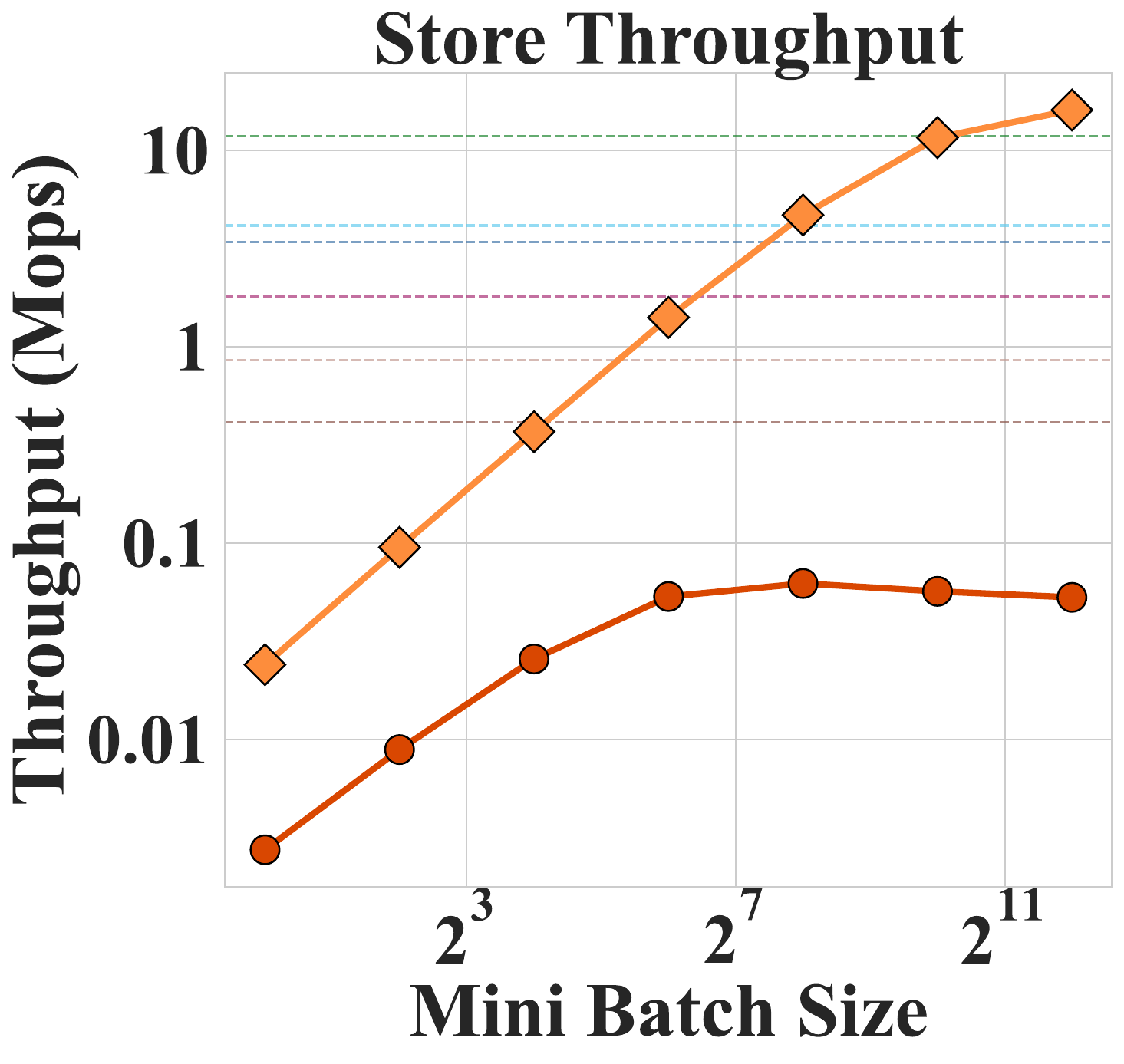}%
		}\hfill%
		\subcaptionbox{\small Query Throughput\label{fig:throughput-query}}{%
			\includegraphics[width=0.49\linewidth]{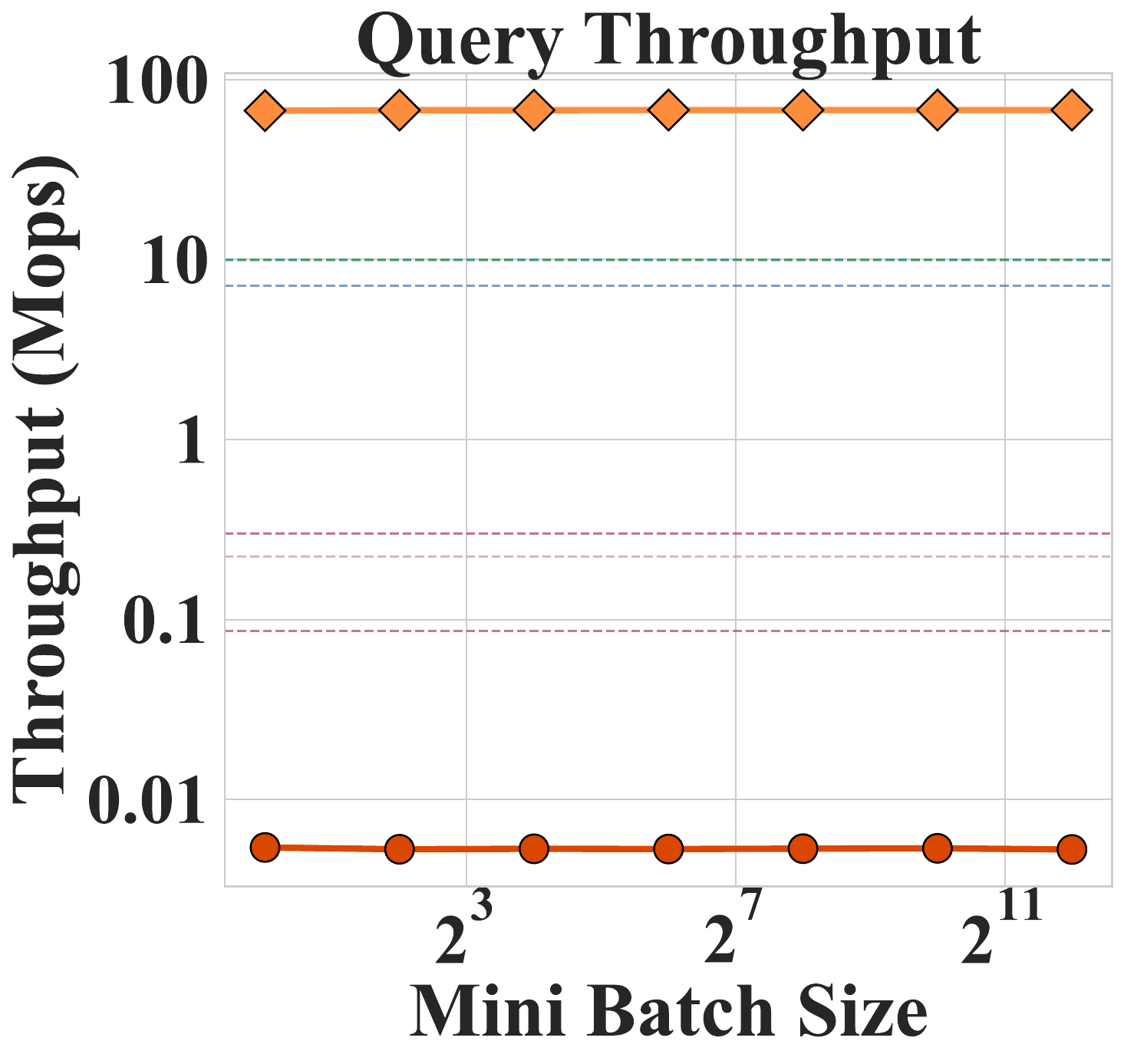}%
		}

		\vspace{-0.3cm}
		\caption{Store and query throughput on the Lkml dataset.}
		\label{fig:throughput}
	\end{minipage}
	\vspace{-0.3cm}
\end{figure}

\bbb{Sensitivity to Carry Threshold.}
The carry threshold $\theta$ governs the ``viscosity'' of the vertical flow; it determines how aggressive the model is in promoting edges to higher layers.
Figure~\ref{fig:param-threshold} plots the ARE as $\theta$ varies from $2$ to $32$.

The sensitivity curve is \emph{U-shaped}, with a dataset-dependent optimum: moderate promotion ($\theta=4$) is best on lighter workloads, while heavy-tailed traces prefer a larger $\theta$ that lets each layer amortize bursty traffic before committing to a promotion.
Across both regimes, $\theta=2$ is consistently the worst choice---overly eager promotion injects unstable carry noise into upper layers before the per-layer estimate stabilizes---which validates the design motivation of the hierarchical carry mechanism (\S\ref{sec:store-module}).
We therefore adopt $\theta=4$ as a single near-optimal compromise across workloads, avoiding per-dataset tuning.

\bbb{Impact of Mini-Batch Strategy.}
\label{sec:ablation-batch}
Figure~\ref{fig:param-batch} shows that ARE grows monotonically with the mini-batch size $b$ across all datasets, reflecting the residual information lost when aggregating updates within a batch.
The gap between $b=1$ and $b=4$ is small, while scaling beyond $b=4$ incurs a sharply increasing cost.
We therefore adopt $b=4$ as the default: a near-optimal accuracy operating point that retains the throughput advantages of batched GPU execution.

\subsection{Experiments on Throughput}
\label{sec:throughput}

Figure~\ref{fig:throughput} reports store and query throughput on the Lkml dataset, comparing Crane against the baselines on both CPU and GPU.

\bbb{Store Throughput and Mini-Batch Scaling.}
Figure~\ref{fig:throughput-store} reports store throughput as a function of the mini-batch size $b$.
On GPU, the conservative batch carry\algparen{} allows Crane to amortize neural-encoding and memory-update costs across batched edges: throughput grows from $0.024$ Mops at $b=1$ to $16.0$ Mops at $b=4096$---roughly a $670\times$ speedup---catching up with the fastest hand-crafted sketch \textit{GSS} ($11.8$ Mops) at $b\!\approx\!1024$ and overtaking it at $b\!=\!4096$.
On CPU, batching also helps but saturates much earlier, around $0.06$ Mops at $b\!\approx\!256$, because CPUs cannot hide kernel launch and matrix-multiply latencies the way GPUs can.
Mayfly, the only other neural baseline, caps at $0.85$ Mops even on GPU since its sum-pooled decoder forbids batched updates at the carry level; Crane's conservative batch carry is therefore the main reason it reaches hash-sketch-level store speed.

\begin{sloppypar}
\bbb{Query Throughput.}
Figure~\ref{fig:throughput-query} shows the opposite picture on the query side: Crane's GPU query throughput sits at a flat $\approx\!67.7$ Mops regardless of $b$, because every query reduces to a forward pass through the encoders followed by $N$ element-wise divisions and a single linear decoder---all operations that map naturally to dense GPU kernels.
This is $6.8\times$ faster than the two fastest hand-crafted sketches, \textit{GSS} ($10.0$ Mops) and \textit{HourglassSketch} ($9.9$ Mops), $9.5\times$ faster than \textit{TCM} ($7.1$ Mops), and two orders of magnitude faster than any other neural baseline.
On CPU, query throughput is a modest $0.005$ Mops---a reminder that Crane is designed as a GPU-first sketch, where the batched decoder is the workload's natural home.
\end{sloppypar}

\bbb{Throughput--Accuracy Trade-off.}
At $b\!=\!4096$, Crane on GPU delivers $16.0$ Mops on store and $67.7$ Mops on query---matching or exceeding every hand-crafted baseline on both operations---while simultaneously reducing ARE by roughly an order of magnitude (Figure~\ref{fig:accuracy}).
The only regime in which Crane is the slower choice is the CPU-only store ($0.06$ Mops vs.\ $1.8\sim11.8$ Mops for hash sketches), which stems from two inherent costs: (1)~neural encoders are heavier than static hash functions, and (2)~the carry mechanism serializes updates across layers, whereas hash sketches enjoy fully independent per-edge updates.
For deployments that can afford a GPU---or even a shared one, as our Flink case study in \S\ref{subsec:system} demonstrates---Crane offers a strict Pareto improvement: higher accuracy at equal or better throughput.

\bbb{Training Cost.}
Training is a one-off offline cost: one standard run under the setup of \S\ref{sec:setup} spends $6.1$\,h in optimization for Crane and $6.7$\,h for Mayfly on a single RTX 5090.
These times scale with the GPU and CPU used and with the size of the training streams; where a downstream query task calls for fine-tuning, as in \S\ref{subsec:casestudy}, adapting an existing checkpoint is typically cheaper than a full run.

\vspace{-0.11in}
\subsection{Case Studies}
\label{subsec:casestudy}
\vspace{-0.04in}

To demonstrate Crane's capability beyond edge weight estimation, we evaluate the three downstream graph-stream queries defined in \S\ref{sec:definition}---node flow, path flow, and subgraph flow---on five real workloads (D1--D5).

\bbb{Setup.}
The five workloads instantiate the three queries of \S\ref{sec:definition} in order, and every query instance is a deterministic export of a real application feed or a standardized benchmark workload.
\begin{sloppypar}
\begin{itemize}[leftmargin=1em]
\item \textbf{D1 (FinBench transfers).}\footnote{\url{https://github.com/ldbc/ldbc_finbench_transaction_impls}} LDBC FinBench is the industry-derived financial graph benchmark, and its simple read TSR2, the transfer-in and transfer-out sums of an account, is exactly the node flow query of \S\ref{sec:definition}.
The stream replays the benchmark's official account-transfer history in creation-time order with the transfer amount as edge weight ($1.45$M transfers over $368$K accounts), and the query accounts are the benchmark driver's own curated parameter files ($18{,}360$ queries).
\item \textbf{D2 (CAIDA host accounting).}\textsuperscript{\ref{fn:caida}} A prefix-preserving anonymized CAIDA backbone trace ($2$M flow records, $110$K hosts).
The workload is an operator's routine accounting sweep, the directional cumulative traffic of every active host: the usage monitoring and provisioning reads of daily operation.
\item \textbf{D3 (RouteViews AS paths).}\footnote{\url{http://archive.routeviews.org/}\label{fn:routeviews}} A gapless $24$-hour window of real BGP announcements from the RouteViews route-views2 collector, expanded into one update per adjacent AS pair ($16$M updates over $49.7$K ASes).
Queries are the $1{,}000$ AS paths most announced in the held-out next hour, and the ground truth is the bottleneck mass of the path, the minimum accumulated link weight.
\item \textbf{D4 (CAIDA prefix matrix).}\textsuperscript{\ref{fn:caida}} An anonymized CAIDA backbone trace ($2$M flow records, $110$K hosts), queried at subgraph granularity: $1{,}000$ activity-ranked directed /16-to-/16 traffic-matrix cells ($3$--$170$ member edges each), the aggregate operators compute for capacity planning and peering analysis; the prefix-preserving anonymization keeps /16 groups structurally real subnets.
\item \textbf{D5 (RouteViews AS footprints).}\textsuperscript{\ref{fn:routeviews}} A gapless $24$-hour window of RouteViews BGP announcements, expanded into adjacent AS pairs ($16$M updates over $49.7$K ASes) and queried at subgraph granularity: the ego subgraphs, all adjacent AS links in both directions, of $1{,}000$ activity-ranked ASes ($3$--$296$ links each), the aggregate an operator audits for a network's connectivity footprint.
\end{itemize}
\end{sloppypar}

All methods keep the 64\,KB budget (node flow is evaluated per direction, each answered within its own 64\,KB sketch, and reported as the mean over the two directions; \ifshowappendix Table~\ref{tab:nodeflow_directions} in \fi\appref{app:workloads} reports each direction separately) and follow the same temporal protocol: the learned methods fine-tune on supervision drawn only from the first 25\% of each stream in arrival order to adapt to the new query tasks (both update only the decoder with the pretrained encoders frozen), and queries are issued at the end of the stream; the baselines are \textit{Mayfly} and \textit{HourglassSketch}, the strongest of Figure~\ref{fig:accuracy}.
Beyond ARE we report the accuracy of what a user would actually receive: nDCG@100 and Recall@100 against the ground-truth top-100.

\begin{figure}[t]
	\centering
	\captionsetup{justification=raggedright}
	\begin{minipage}[b]{0.54\linewidth}
		\centering
		\includegraphics[width=0.8\linewidth]{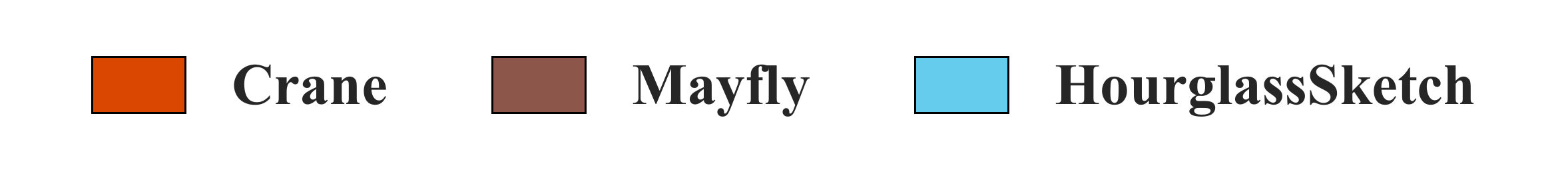}

		\includegraphics[width=0.32\linewidth]{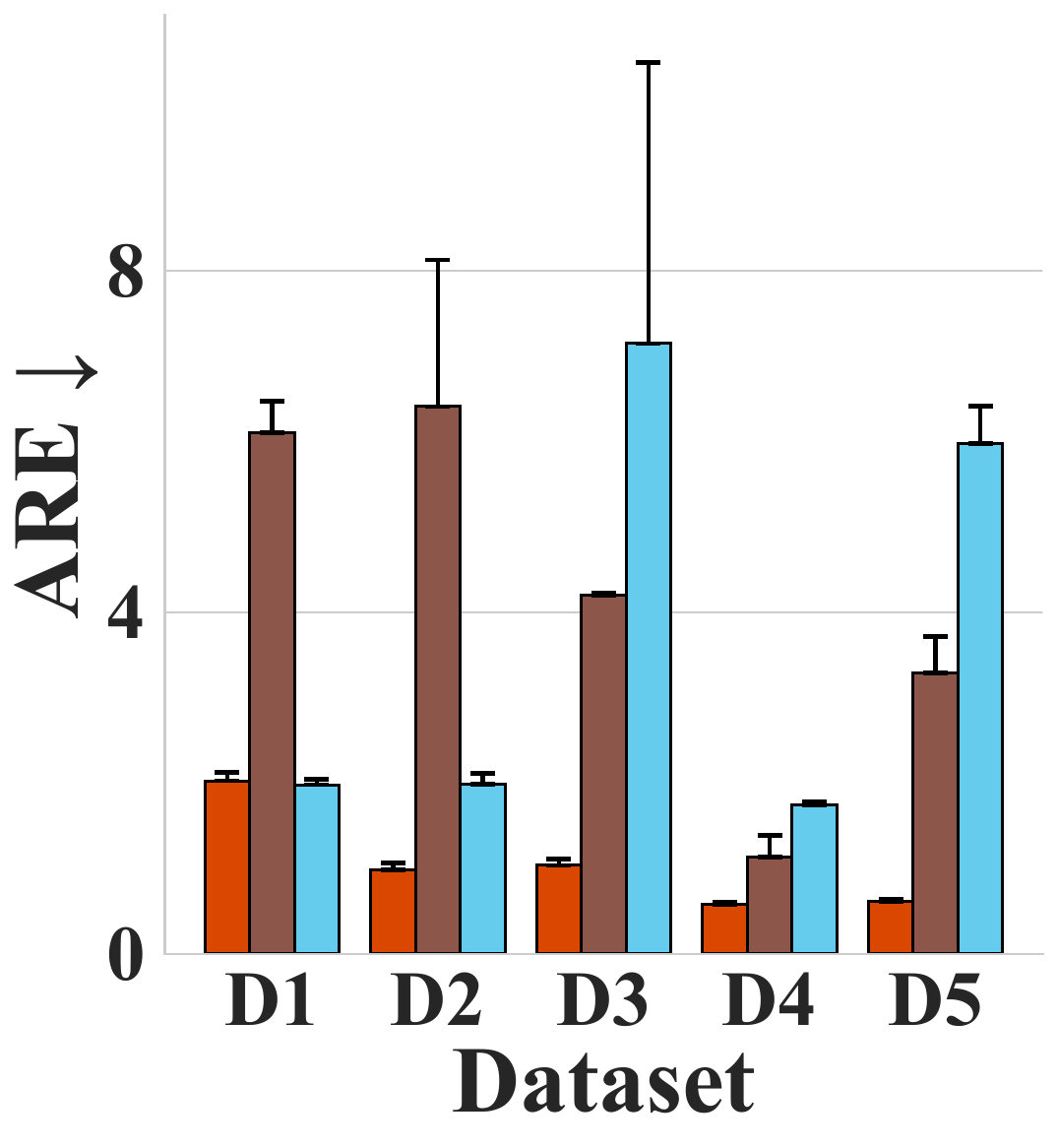}\hfill
		\includegraphics[width=0.32\linewidth]{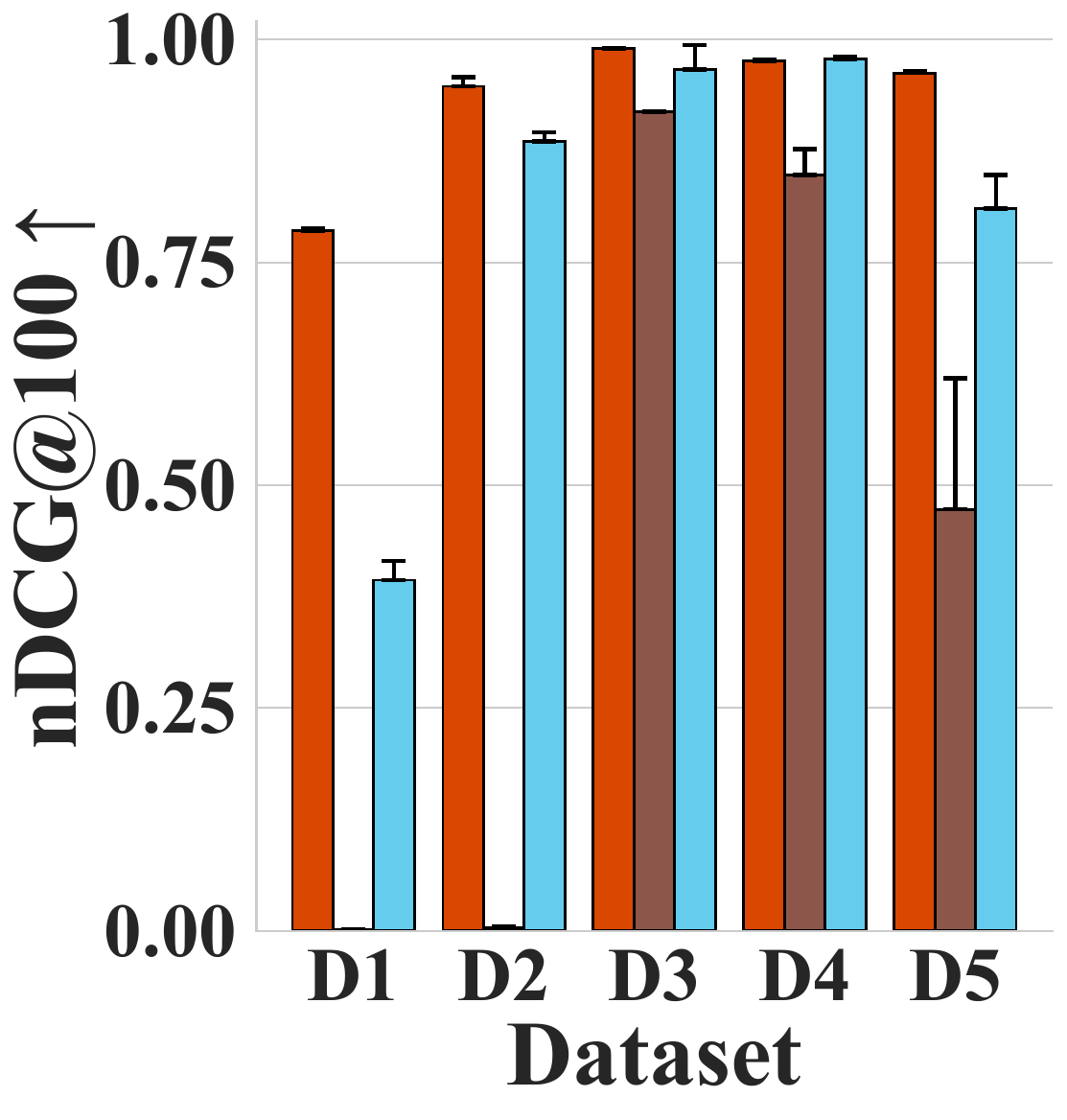}\hfill
		\includegraphics[width=0.32\linewidth]{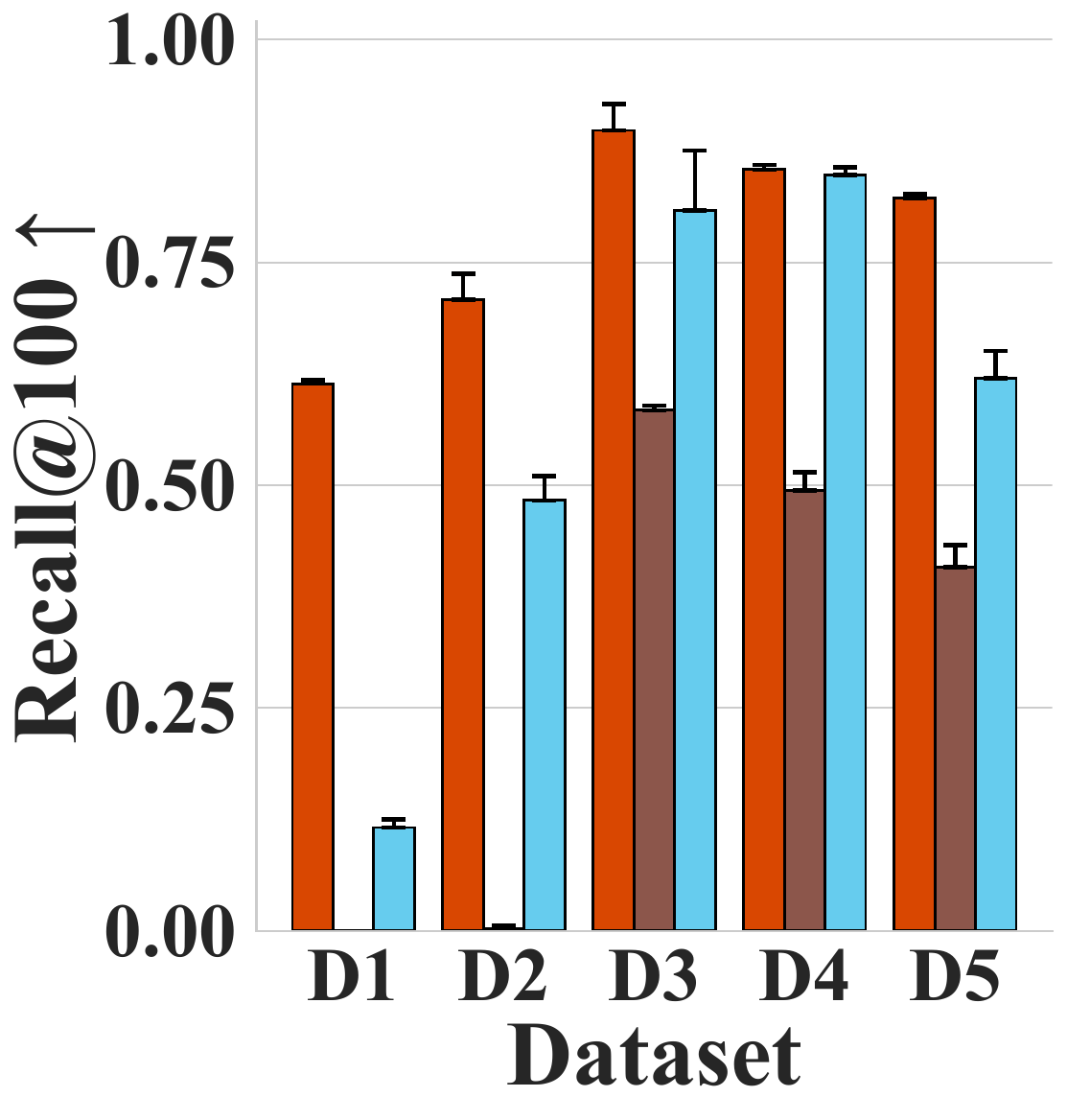}

		\vspace{-0.25cm}
		\caption{Case study on five real workloads; mean $\pm$ s.d.\ over random seeds. }
		\label{fig:casestudy}
	\end{minipage}\hfill
	\begin{minipage}[b]{0.43\linewidth}
		\centering
		\includegraphics[width=\linewidth]{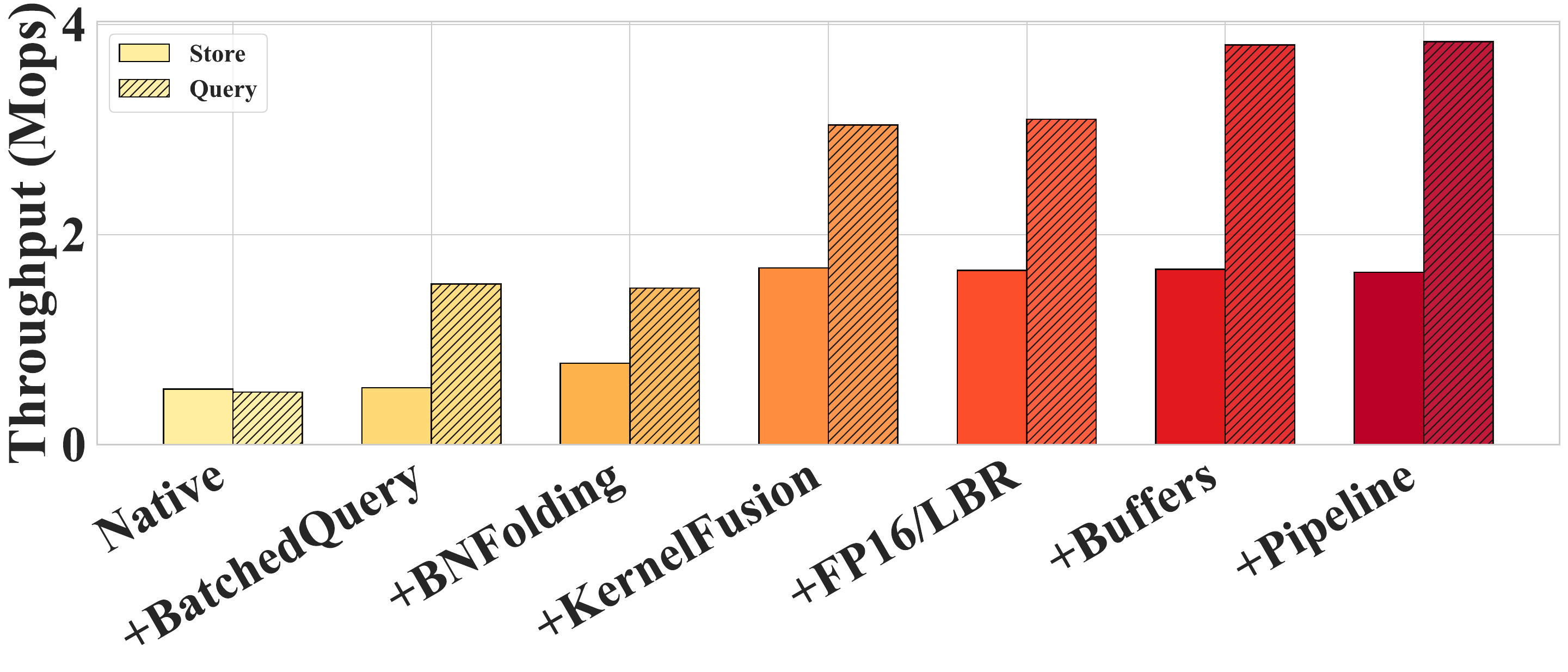}

		\vspace{-0.25cm}
		\caption{Throughput optimization ablation from native PyTorch.}
		\label{fig:flink-ablation}
	\end{minipage}
	\vspace{-0.1cm}
\end{figure}

\bbb{Results.}
Figure~\ref{fig:casestudy} reports all three metrics. Crane achieves the best overall performance, consistently outperforming or matching the strongest baseline across all workloads and metrics.
On ARE, Crane is the most accurate on four of the five workloads, improving on the strongest baseline by up to $5.4\times$ (D5), and ties HourglassSketch within one seed s.d.\ on D1.
The ranking metrics separate the methods more sharply. On the node flow workloads (D1, D2), the ARE ordering conceals what is actually returned: Mayfly's rankings collapse outright, and where ARE cannot separate Crane from HourglassSketch (D1), the top-$100$ accounts HourglassSketch returns reach only half of Crane's nDCG@100 and a fifth of its Recall@100.
On the path and subgraph flow workloads (D3--D5), HourglassSketch ranks competitively and matches Crane on D4, while Mayfly's rankings weaken again on the subgraph workloads; across the five workloads only Crane sustains high-quality rankings, with nDCG@100 between $0.79$ and $0.99$.
Crane's advantage on edge weight estimation thus carries over to the downstream queries composed from it.

\vspace{-0.08in}
\subsection{System Implementation}
\label{subsec:system}
\vspace{-0.02in}

We implement \textit{Crane} as a stateful operator on top of Apache Flink \cite{carbone2015apache} to demonstrate that it is practical for modern graph-stream analytics.
The full system stack consists of (i)~a PyTorch implementation of \textit{Crane} with custom Triton~\cite{tillet2019triton} kernels for the hot paths, (ii)~a Flink \textit{Python DataStream} job that binds the operator to two input streams, and (iii)~a dockerised cluster with an HDFS-backed namenode and up to four TaskManagers that share the single GPU on the workstation.

\bbb{Flink Dataflow.}
Edges are ingested from HDFS as two batched text streams---one carrying \textsc{Store} updates and one carrying \textsc{Edge-Query} requests, both partitioned into buckets keyed by a hash of the edge identifier.
The two streams are \texttt{connect}ed and dispatched to a \texttt{KeyedCoProcessFunction} that holds a per-bucket \textit{Crane} instance.
Each TaskManager slot runs an independent operator that processes its assigned buckets; per-key managed state buffers in-flight query batches until an event-time timer fires, at which point all queued queries for the bucket are resolved in a single fused call.
This design exposes bucket-level data parallelism to Flink, so end-to-end throughput is bounded only by the number of slots and GPUs available in the cluster.

\bbb{Heterogeneous Optimization Pipeline.}
A literal PyTorch port of \textit{Crane}'s \textsc{Store}/\textsc{Query} operators leaves considerable performance on the table: small per-edge tensors, BatchNorm modules, layout mismatches, and per-batch Python overhead each introduce milliseconds of latency.
We close this gap by stacking six classes of inference-time optimizations drawn from the machine-learning compiler and high-performance computing literatures:
\begin{enumerate}[leftmargin=*,itemsep=0pt,topsep=2pt]
    \item \textbf{Launch amortization.}
    All 512-edge query mini-batches queued for a bucket are concatenated and dispatched in a single GPU call, removing $\mathcal{O}(10^{2})$ per-batch kernel-launch overheads per bucket.
    \item \textbf{Algebraic simplification.}
    BatchNorm layers are algebraically folded into the preceding linear layers, and the $2L$ per-layer embedding networks are reorganized into a single concatenated Layer-1 GEMM plus two batched matrix multiplies for Layers~2--3~\cite{vanholder2016efficient}.
    \item \textbf{Operator fusion.}
    The outer-product / ratio-minimum kernel of \textsc{Query} and the outer-product / carry kernel of \textsc{Store} are replaced by custom Triton kernels that avoid materializing the $[B,N,H,W]$ intermediate tensor and keep the hierarchical carry logic on the device~\cite{tillet2019triton}.
    \item \textbf{Mixed precision and data layout.}
    The embedding path is executed in FP16 and emitted in $[N,B,\cdot]$ layout so that the fused query kernel can ingest it directly, eliminating a per-invocation \texttt{permute+contiguous+float} copy on the embedding tensor~\cite{micikevicius2017mixed}.
    \item \textbf{Memory management.}
    Pinned host buffers, GPU scratch buffers, and the query output tensor are pre-allocated in \texttt{open()} and reused across buckets, removing repeated \texttt{cudaHostAlloc} calls.
    \item \textbf{Asynchronous pipelining.}
    A dedicated CUDA stream paired with a ring of 32 pinned buffers allows the CPU to decode the next 512-edge batch while the GPU processes the current one, fully overlapping decoding, host-to-device DMA, and compute on the \textsc{Store} pipeline.
\end{enumerate}

\begin{figure}[t]
	\vspace{-0.1cm}
    \centering
	\begin{minipage}{\linewidth}
		\centering
		\includegraphics[width=0.6\linewidth]{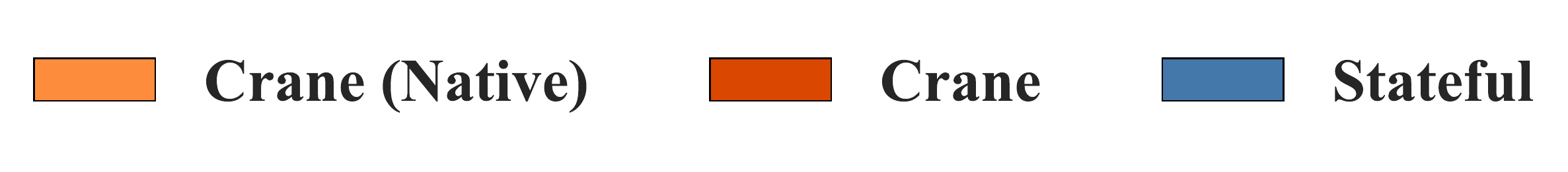}
		\vspace{-0.1cm}
	\end{minipage}
	\begin{minipage}{1 \linewidth}
		\captionsetup[subfigure]{skip=-0.08cm}
		\centering

		\subcaptionbox{\small Local (Store)\label{fig:flink-local-store}}{%
			\includegraphics[width=0.235\linewidth]{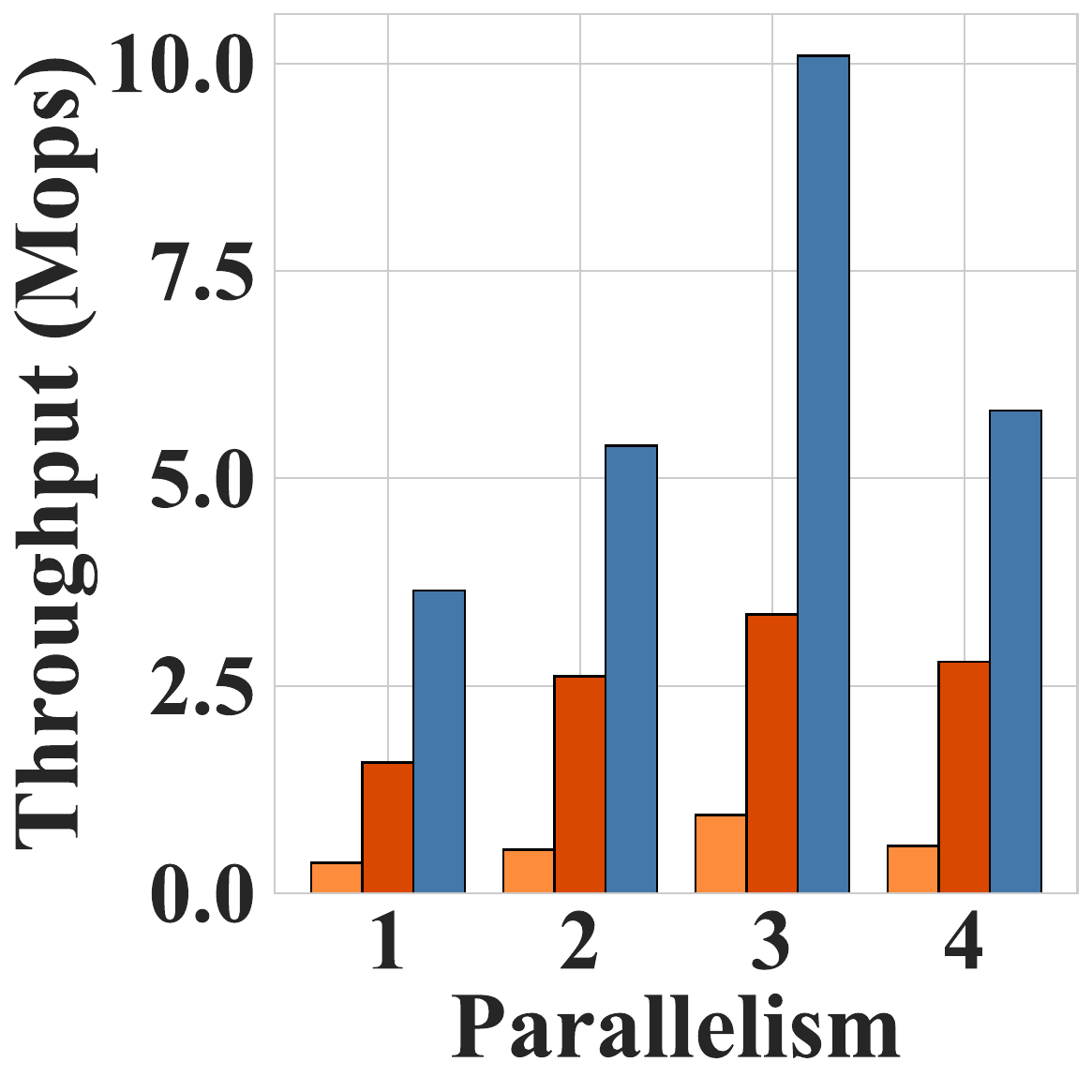}%
		}\hfill%
		\subcaptionbox{\small Local (Query)\label{fig:flink-local-query}}{%
			\includegraphics[width=0.235\linewidth]{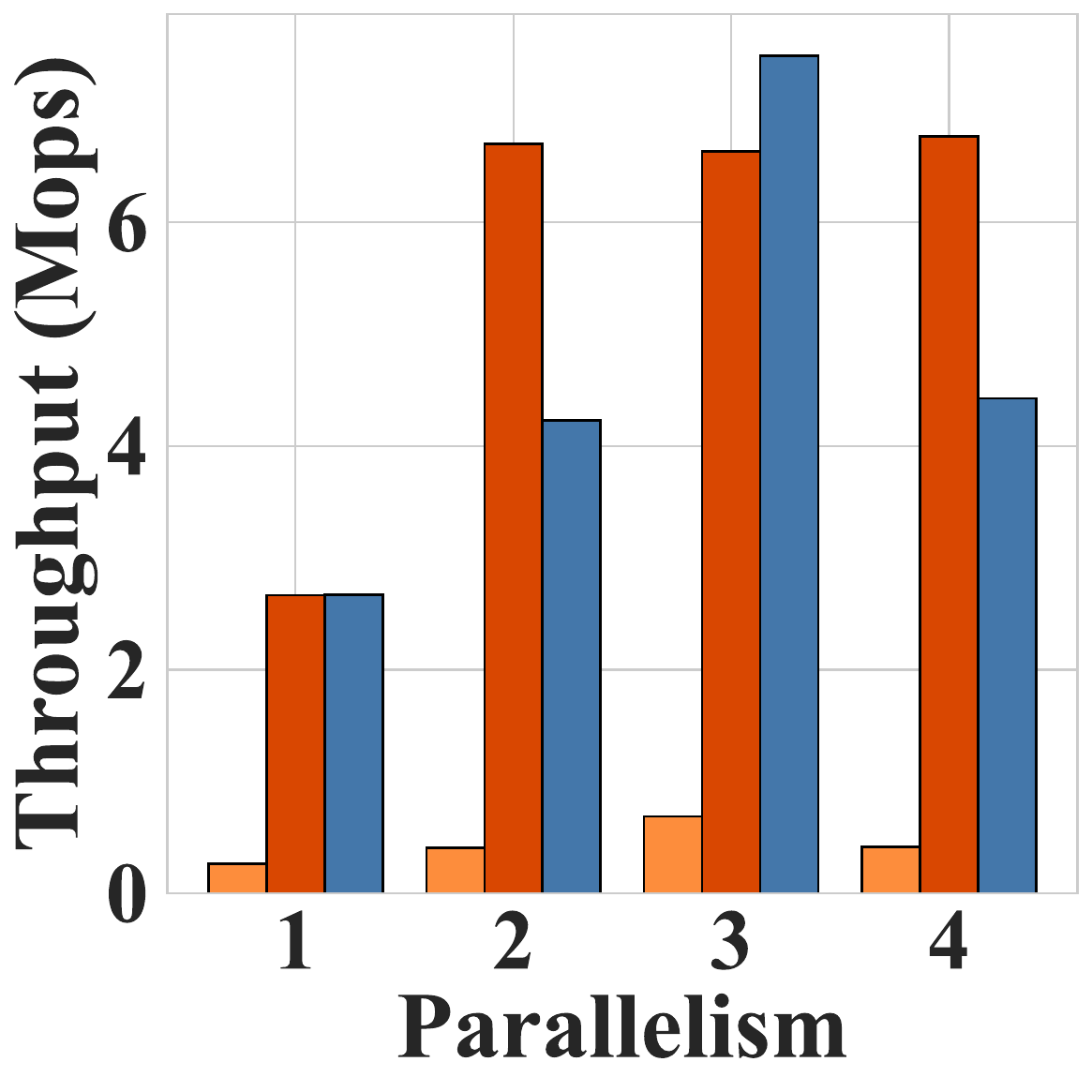}%
		}\hfill%
		\subcaptionbox{\small Cluster (Store)\label{fig:flink-cluster-store}}{%
			\includegraphics[width=0.235\linewidth]{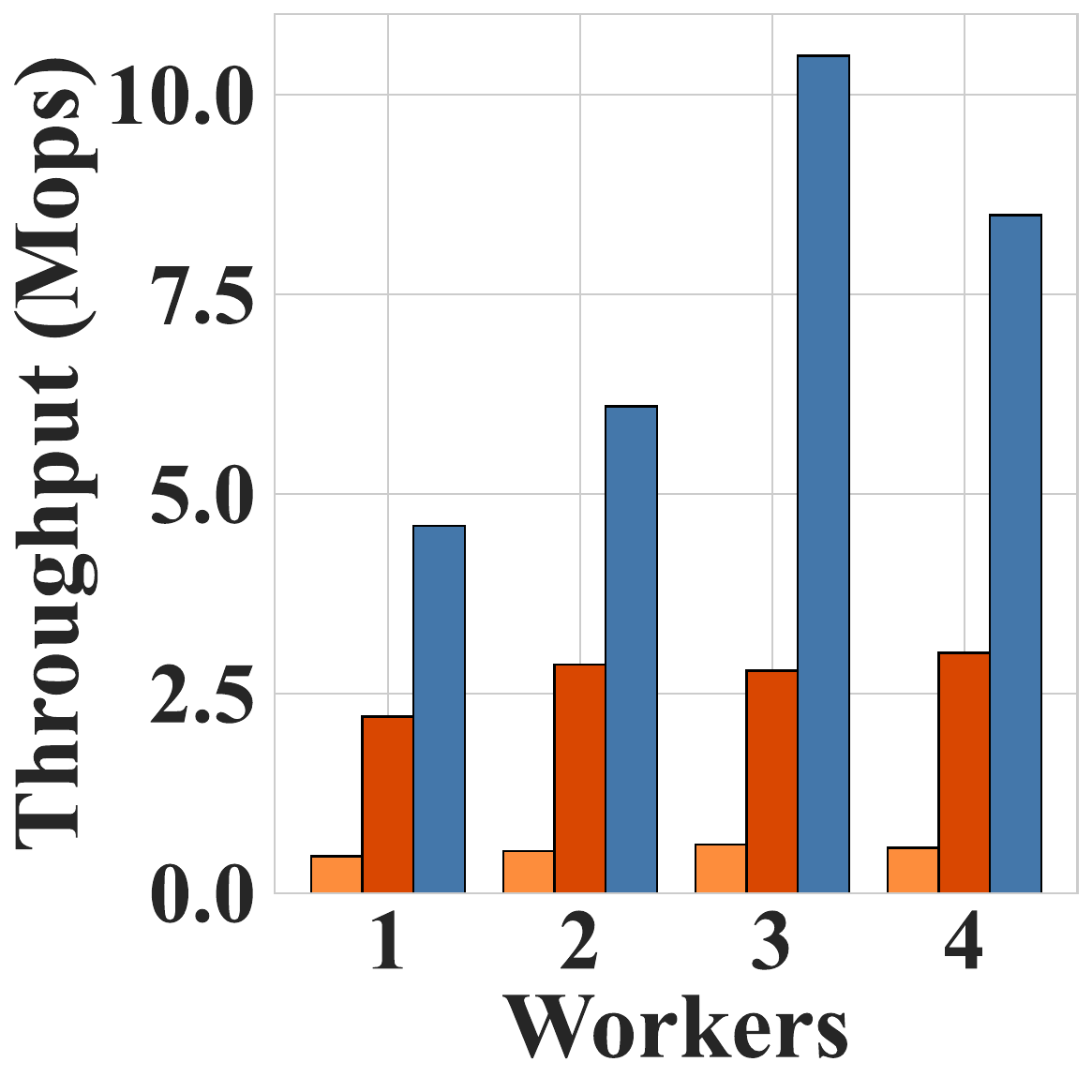}%
		}\hfill%
		\subcaptionbox{\small Cluster (Query)\label{fig:flink-cluster-query}}{%
			\includegraphics[width=0.235\linewidth]{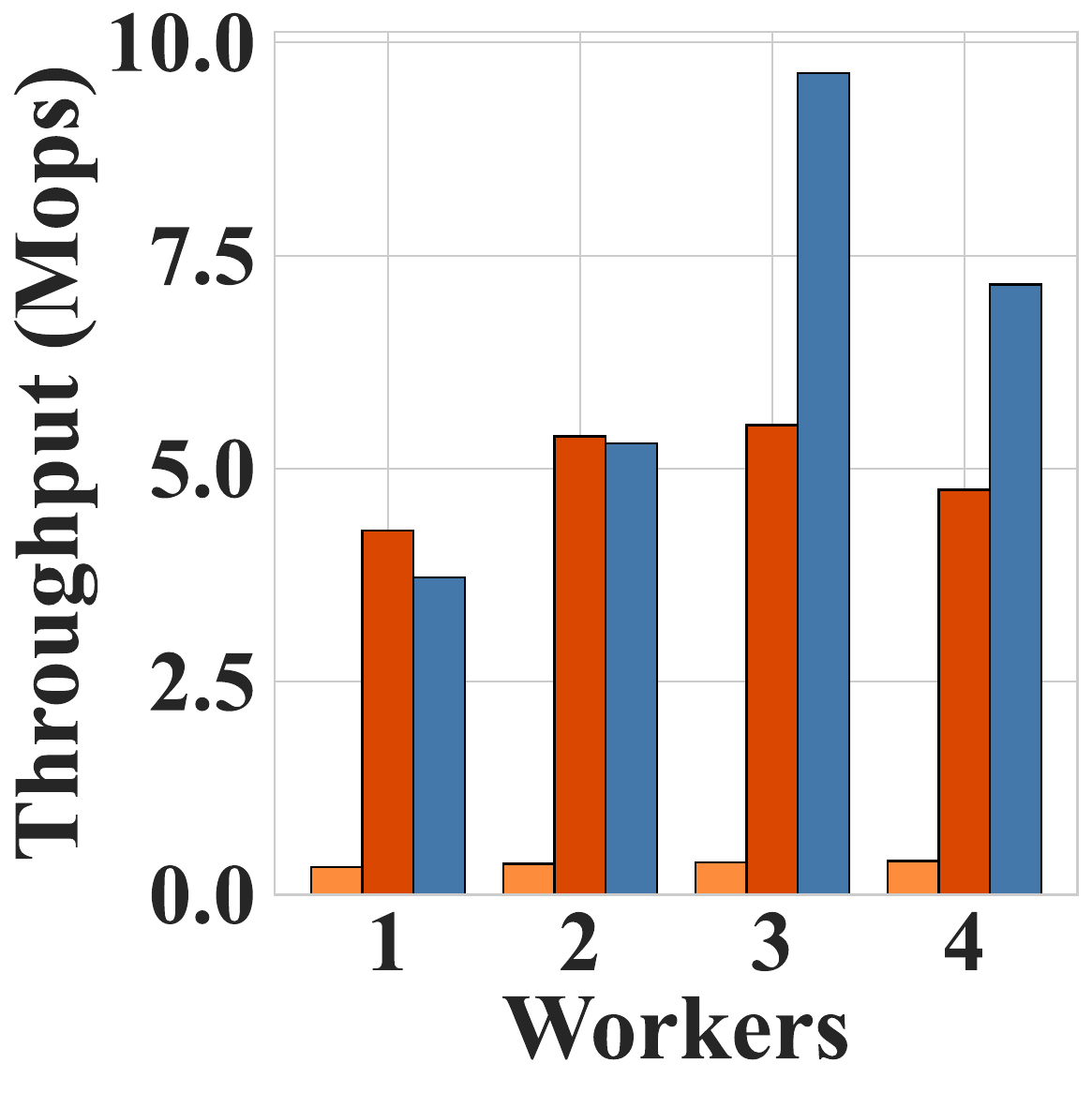}%
		}

		\vspace{-0.4cm}
		\caption{Throughput under Flink deployment: \textsc{Store} and \textsc{Query} under local and cluster configurations.}
		\label{fig:flink-throughput}
	\end{minipage}
	\vspace{-0.5cm}
\end{figure}

\bbb{Ablation Study.}
We quantify the contribution of each optimization class via a progressive ablation on CAIDA (single GPU, parallelism $=1$).
Starting from a vanilla PyTorch deployment, we enable the six optimizations one at a time.
As shown in Figure~\ref{fig:flink-ablation}, the cumulative effect yields a $7.7\times$ speedup on query throughput and $3.1\times$ on store throughput.
None of these techniques alters \textit{Crane}'s output; small residual differences in ARE are attributable to FP16 rounding and are numerically equivalent up to rounding.

\bbb{Local and cluster scalability.}
Figure~\ref{fig:flink-throughput} reports end-to-end \textsc{Store} and \textsc{Query} throughput along two scaling axes: \emph{local parallelism}, where we vary the number of operator slots within a single TaskManager (each slot pinned to one CPU core), and \emph{cluster workers}, where we vary the number of TaskManagers (each worker configured with parallelism $=2$).
We compare three variants: the unoptimized PyTorch port \textit{Crane (Native)} (the A0 baseline of Figure~\ref{fig:flink-ablation}), our fully optimized operator \textit{Crane}, and a \textit{Stateful} variant that additionally coalesces in-flight queries through Flink's managed state and event-time timers.
Two consistent trends emerge across both axes.
First, optimizations alone lift \textit{Crane} by an order of magnitude over \textit{Crane (Native)}, which never exceeds $\approx\!0.6$\,Mops on either path---confirming that the speedups reported earlier, rather than parallelism, are what close the gap with handcrafted baselines on commodity hardware.
Second, throughput scales sub-linearly and saturates around three workers / slots: with the GPU shared, additional workers chiefly help to overlap HDFS decode and CPU-side ingestion with GPU compute, after which contention on the single GPU outweighs the marginal benefit.
The absolute numbers also fall short of the direct-GPU throughput reported in Figure~\ref{fig:throughput}, mainly because the Flink runtime imposes per-record serialization, JVM--Python boundary crossings, and text-to-tensor format conversion that the standalone GPU benchmark does not pay.

\end{sloppypar}

%% file: MainText/7-Conclusion.tex
\section{Conclusion}

In this paper, we propose Crane, a hierarchical neural sketch for graph stream summarization. 
Unlike existing methods, Crane utilizes a carry-like promotion mechanism to effectively separate frequent and infrequent edges across layers, solving the issue of interference in limited memory environments.
Further, Crane supports smooth automatic expansion, allowing it to handle streams of unknown or infinite length efficiently. 
Theoretical proofs confirm the space efficiency and error bounds of our design. 
Extensive evaluations show that Crane consistently outperforms SOTA methods, reducing ARE by one to three orders of magnitude while maintaining competitive throughput.
Crane also delivers acceptable performance when integrated into Apache Flink.

%% file: MainText/8-Appendix.tex
\section{Mathematical Analysis}
\subsection{Proof of Logarithmic Space Efficiency (Lemma~\ref{lemma:space})}
\label{proof:space}
In Crane's hierarchical mechanism, the weight stored at layer $i$, denoted as $v_i$, can be viewed as analogous to the magnitude of a digit in an idealized base-$\theta$ system.
The total estimated weight is approximately $\sum_{i=1}^N v_i \cdot \theta^{i-1}$.
Considering the worst-case scenario where each layer is filled to its capacity (just below the threshold $\theta$) before a carry occurs, the maximum capacity $C(N)$ of a Crane sketch with $N$ layers is given by the geometric series:
\begin{equation}
    C(N) \approx \sum_{i=1}^{N} (\theta - 1) \cdot \theta^{i-1} = (\theta - 1) \frac{\theta^N - 1}{\theta - 1} = \theta^N - 1.
\end{equation}
To accommodate a stream where the maximum edge weight is $F_{\text{max}}$, we require $C(N) \geq F_{\text{max}}$.
Solving for $N$:
\begin{equation}
    \theta^N \geq F_{\text{max}} + 1 \implies N \geq \log_\theta(F_{\text{max}} + 1).
\end{equation}

Thus, the depth of the hierarchy $N$ scales as $\mathcal{O}(\log_\theta F_{\text{max}})$.
For a fixed per-layer memory size $\mathcal{W}$, the total memory $N \times \mathcal{W}$ therefore grows logarithmically with the represented weight magnitude.

\subsection{Proof of Logarithmic Amortized Complexity (Theorem~\ref{thm:time})}
\label{proof:time}
Let $N_t$ denote the number of allocated memory layers when the $t$-th edge $e_t$ is inserted.
According to the \textsc{Store} operation, an insertion may touch at most $N_t$ active layers.
Therefore, the insertion cost satisfies
\begin{equation}
    \tau(e_t) = \mathcal{O}(N_t).
\end{equation}
Due to the carry-like promotion with threshold $\theta$, accommodating cumulative weight $W_t$ requires
\begin{equation}
    N_t = \mathcal{O}(\log_\theta W_t).
\end{equation}
Since $W_t \le W_{\text{total}}$ for all $t$, the per-insertion cost is upper-bounded by $\mathcal{O}(\log_\theta W_{\text{total}})$.
Summing over the entire stream of $T$ edges gives
\begin{equation}
    \tau_{\text{total}} = \sum_{t=1}^{T} \tau(e_t)
    \le
    \sum_{t=1}^{T} c \cdot \log_\theta W_{\text{total}}
    =
    \mathcal{O}(T \log_\theta W_{\text{total}}),
\end{equation}
where $c > 0$ is a constant hidden in the $\mathcal{O}(\cdot)$ bound.
Dividing by $T$, the amortized complexity is $\mathcal{O}(\log_\theta W_{\text{total}})$.

\subsection{Proof of Collision Probability Decay (Theorem~\ref{thm:collision})} 
\label{proof:collision} 
Let $e$ be a heavy hitter promoted up to layer $k$.
The estimation accuracy depends on recovering its carried values.
Let $e'$ be a fixed distinct edge, and let $C_i=C_i(e,e')$
denote the event that $e$ and $e'$ collide in layer $i$.
Under Assumption 1, $P(C_1) \leq p$ and
$P(C_i \mid C_1,\ldots,C_{i-1}) \leq p$ for $i\geq 2$.
For an error to persist up to layer $k$, collisions must occur simultaneously at every step of the promotion path.
The joint probability is:
\begin{equation}
\begin{aligned}
P(\text{Collision across } k \text{ layers})
&=
P\Bigl(\bigcap_{i=1}^{k} C_i\Bigr) \\
&=
P(C_1)\prod_{i=2}^{k}
P(C_i\mid C_1,\ldots,C_{i-1}) \\
&\leq p^k.
\end{aligned}
\end{equation}
Since $p < 1$, the probability $p^k$ decays exponentially as $k$ increases.
This implies that the significant ``quotients'' stored in higher layers are protected by an effective collision rate of $p^k$.

\subsection{Discussion of the Idealized Residual-Mass View}
\label{proof:residual}
The discussion below is intended as an idealized interpretation rather than a formal proof for the learned continuous-carry system.

\begin{itemize}[leftmargin=1em]
    \item \textbf{Bottom-Layer View:} In an idealized decomposition, one may associate each edge $e'$ with a residual quantity $r_{e'} = f(e') \pmod \theta$.
    Let $I_{e, e'}^{(1)} \in \{0, 1\}$ be the indicator random variable that equals $1$ if $e'$ collides with $e$ in layer $1$.
    Under Assumption~1, $\mathbb{E}[I_{e, e'}^{(1)}] \le p$.
    In this idealized view, the expected bottom-layer collision exposure satisfies
    \begin{equation}
    \begin{aligned}
        \mathbb{E}\left[\text{Noise}^{(1)}\right]
        &= \mathbb{E}\left[ \sum_{e' \neq e} I_{e, e'}^{(1)} \cdot r_{e'} \right] \\
        &= \sum_{e' \neq e} \mathbb{E}\left[I_{e, e'}^{(1)}\right] \cdot r_{e'} \\
        &\le p \sum_{e' \neq e} r_{e'}
        \le p W_{\text{res}}.
    \end{aligned}
    \end{equation}

    \item \textbf{Higher-Layer View:} Based on Theorem~\ref{thm:collision}, repeated collision events across multiple layers become increasingly unlikely.
    If one further adopts an idealized promotion model, the higher-layer contribution can be heuristically related to terms of the form
    \begin{equation}
        \mathbb{E}[\text{Noise}^{(i)}] \lesssim p^i \cdot \frac{W_{\text{total}}}{\theta^{\,i-1}}.
    \end{equation}
\end{itemize}

Taken together, this idealized picture motivates relating the bottom-layer collision exposure to a term of order
\[
\mathcal{O}\left(p W_{\text{res}}\right),
\]
while the higher-layer contribution becomes increasingly unlikely as the number of layers grows.
We emphasize again that this is an interpretive argument rather than an exact error identity for the learned continuous-carry system.

\subsection{Heuristic Analysis of Encoder Orthogonality Optimization}
\label{proof:learnable}
While Assumption~1 imposes a bounded-collision condition for analysis, the actual neural encoders $f_i(\cdot)$ dynamically adapt to the data distribution.
Let $A_u, A_v$ be basis matrices for edges $u, v$.
The collision interference is proportional to their inner product $\langle A_u, A_v \rangle$.
During training, the gradient descent update for encoder parameters $\phi$ is given by:
\begin{equation}
    \phi \leftarrow \phi - \eta \frac{\partial \mathcal{L}}{\partial \langle A_u, A_v \rangle} \frac{\partial \langle A_u, A_v \rangle}{\partial \phi}
\end{equation}
For heavy hitters that significantly contribute to the reconstruction loss $\mathcal{L}$, this provides a continuous optimization pressure that actively penalizes large inner products.
Consequently, the gradient dynamics may reduce $\langle A_u, A_v \rangle$ for heavy hitters, providing a heuristic explanation for why the effective interference between heavy hitters can be lower in practice than the conservative collision bound $p$ in Assumption~1.

\subsection{Interpretation of the Linear Decoder under an Idealized MSE Model}
\label{proof:variance}
Let the weight estimate from layer $i$ be a random variable $Q_i$ with variance $\sigma_i^2$, and let $w_i$ denote the $i$-th entry of the decoder weight vector $\mathbf{w}$ defined in \S\ref{sec:query-module}.
Lower layers experience higher collision rates (larger $\sigma_i^2$), while higher layers are sparser.
If one ignores cross-covariance terms and assumes the layer-wise estimates are approximately unbiased, the variance of the linearly decoded estimate $\hat{y} = \sum_{i=1}^N w_i Q_i$ can be written as
\[
\text{Var}(\hat{y}) \approx \sum_{i=1}^N w_i^2 \sigma_i^2.
\]

If one further considers the auxiliary optimization problem of minimizing a Mean Squared Error (MSE) surrogate subject to the unbiasedness constraint $\sum_{i=1}^N w_i = 1$,
\begin{equation}
    \min_{\mathbf{w}} \sum_{i=1}^N w_i^2 \sigma_i^2
    \quad
    \text{s.t.}
    \quad
    \sum_{i=1}^N w_i = 1,
\end{equation}
then the classical Lagrangian solution is
\[
w_i = \frac{\sigma_i^{-2}}{\sum_{j=1}^N \sigma_j^{-2}},
\]
which corresponds to inverse-variance weighting.

This classical result provides a useful interpretation for why a learnable linear decoder may assign smaller weights to noisier lower layers.
In Crane, however, we use this argument only as an idealized analogy for the learned decoder, rather than as a proof that the trained decoder is provably optimal for the actual training objective.

\subsection{Discussion of Qualitative Interference Isolation}
\label{proof:isolation}

\paragraph{Problem Setup.}
Consider a target edge $e_{\text{target}}$ with low cumulative weight (a tail edge) stored in a sketch with per-layer memory size $\mathcal{W}$.
For comparison, we use the same idealized per-layer collision bound $p$ for both structures.
We compare the collision exposure in two architectures:
\begin{itemize}[leftmargin=1em]
    \item \textbf{Flat Structure (e.g., Mayfly):} Accumulates the full weight $f(e')$ of interfering edges into a single shared memory space.
    \item \textbf{Hierarchical Structure (Crane):} Under the idealized residual-mass view, a quantity analogous to $r_{e'} = f(e') \pmod \theta$ is associated with the bottom layer, while larger magnitudes are progressively moved upward.
\end{itemize}

\paragraph{Noise in Flat Structure.}
In a standard Count-Min-style flat structure, the expected collision exposure is bounded by the full weights of all colliding edges:
\begin{equation}
    \mathbb{E}[\text{Error}_{\text{Flat}}]
    =
    \mathbb{E}\left[ \sum_{e' \neq e_{\text{target}}} I_{e', e_{\text{target}}} \cdot f(e') \right]
    \le
    p \sum_{e' \neq e_{\text{target}}} f(e')
    \le
    p W_{\text{total}},
\end{equation}
where $W_{\text{total}} = \|f\|_1$.

\paragraph{Noise in Crane.}
Under the same idealized residual-mass view, the bottom-layer collision exposure can be related to residual quantities $r_{e'}$:
\begin{equation}
    \mathbb{E}[\text{Error}_{\text{Crane}}^{(1)}]
    =
    \mathbb{E}\left[ \sum_{e' \neq e_{\text{target}}} I_{e', e_{\text{target}}}^{(1)} \cdot r_{e'} \right]
    \le
    p \sum_{e' \neq e_{\text{target}}} r_{e'}
    \le
    p W_{\text{res}}.
\end{equation}

\paragraph{Comparison.}
Comparing these two idealized expressions suggests that the relative bottom-layer collision exposure scales on the order of
\begin{equation}
    \mathcal{R}
    \sim
    \mathcal{O}\left(\frac{W_{\text{res}}}{W_{\text{total}}}\right).
\end{equation}
This comparison is qualitative and is intended only to explain why hierarchical carry can reduce collision exposure for tail edges on heavy-tailed streams.

\section{Expressiveness of the Edge Encoding}
\label{app:expressive}

\subsection{Declarative Aggregation over Attributed Streams}
\label{app:attributed}

Graph streams often carry attributes beyond the endpoint pair---a packet trace records a protocol and a port, a transaction stream a currency and a channel---and declarative aggregation over such streams selects the updates to aggregate with a predicate on those attributes.
We describe here how Crane answers this class of queries without modifying the \textsc{Store} or \textsc{Query} operations of Algorithm~\ref{alg:crane-ops}.

\bbb{Key augmentation.}
Let every update carry a discrete attribute value $a$ in addition to $(o,d,w)$.
We take the aggregation key to be the triple $(o,d,a)$ rather than the pair $(o,d)$, and accumulate
\[
f(o,d,a) \;=\; \sum_{t \,:\, (o_t,d_t,a_t) = (o,d,a)} w_t .
\]
No architectural change is needed: the Embedding module consumes fixed-length binary codes (\S\ref{sec:embedding-module}), so it suffices to encode the pair $\langle d, a\rangle$ in place of $d$.
Reserving $b$ bits for the attribute field shrinks the addressable node space by a factor $2^{b}$, which is the only structural cost of the construction.

Two properties follow directly, and they are the ones that make the key augmentation the right reduction rather than a bookkeeping trick.
First, two updates that share an attribute value but differ in their endpoints, say $(u,v,\texttt{red})$ and $(x,y,\texttt{red})$, are distinct keys and therefore receive distinct basis matrices---as they must, since they belong to distinct aggregation groups.
Second, an edge whose attribute changes over time contributes to several keys rather than overwriting one: updates on $(u,v)$ carrying \texttt{red} accumulate in $f(u,v,\texttt{red})$ and those carrying \texttt{blue} in $f(u,v,\texttt{blue})$, so neither group is corrupted by the other.

\bbb{Predicate queries.}
A predicate that selects an attribute set $\mathcal{A}$ is answered by summing standard edge queries,
\begin{equation}
\label{eq:pred-agg}
f_{\mathcal{A}}(o,d) \;=\; \sum_{a \in \mathcal{A}} f(o,d,a),
\end{equation}
each term being one invocation of \textsc{Query}.
The same composition carries over to the node flow and subgraph flow queries of \S\ref{sec:definition}: restricting either to $\mathcal{A}$ replaces every edge term $f(e)$ by $f_{\mathcal{A}}(e)$.

\begin{sloppypar}
\bbb{Worked example.}
Consider the attributed stream
\[
(u,v,\texttt{red},3),\quad
(u,v,\texttt{blue},5),\quad
(u,v,\texttt{red},2),\quad
(x,y,\texttt{red},4),
\]
which Crane stores under three keys, with $f(u,v,\texttt{red}) = 5$, $f(u,v,\texttt{blue}) = 5$, and $f(x,y,\texttt{red}) = 4$.
The aggregation ``total weight of the \texttt{red} or \texttt{blue} edges from $u$ to $v$'' is $f_{\{\texttt{red},\texttt{blue}\}}(u,v) = 5 + 5 = 10$ by Eq.~\eqref{eq:pred-agg}, obtained from two \textsc{Query} calls; restricting the predicate to \texttt{red} returns $5$.
The update $(x,y,\texttt{red},4)$ is excluded by the endpoint restriction rather than by the predicate.
\end{sloppypar}

\bbb{Scope.}
The reduction assumes a discrete attribute domain of known width, and its cost is the enlarged key space: up to $2^{b}$ times as many distinct keys compete for the same budget, so accuracy under an attributed workload tracks Crane's behavior on a proportionally larger edge set (cf.\ the distinct-edge stress test of \S\ref{sec:accuracy}).
Predicates over continuous attributes, and selections ranging over an unbounded set of endpoint pairs, fall outside this construction and would require a different key design.

\subsection{Directionality of the Basis Matrix}
\label{app:direction}

Crane encodes an edge as the outer product $A_{(o,d)} = \mathbf{E}_o (\mathbf{E}_d)^\top$ of embeddings produced by two \emph{independent} encoders, one distinct pair per layer (\S\ref{sec:embedding-module}); we fix one layer and omit the layer index $(i)$ throughout.
We record here exactly when this encoding can fail to separate $(o,d)$ from $(d,o)$.

Write $\mathbf{u}_x = f^{(o)}(B_x) \in \mathbb{R}^{H}$ and $\mathbf{v}_x = f^{(d)}(B_x) \in \mathbb{R}^{W}$ for the source- and destination-role embeddings of a node $x$, so that $A_{(o,d)} = \mathbf{u}_o \mathbf{v}_d^\top$ and $A_{(d,o)} = \mathbf{u}_d \mathbf{v}_o^\top$.
Both are rank-one with strictly positive entries, and two such matrices coincide precisely when their factors are proportional:
\begin{equation}
\label{eq:dir-cond}
\mathbf{u}_o \mathbf{v}_d^\top = \mathbf{u}_d \mathbf{v}_o^\top
\iff
\exists\, \alpha > 0 :\;
\mathbf{u}_d = \alpha\, \mathbf{u}_o
\;\text{ and }\;
\mathbf{v}_o = \alpha^{-1} \mathbf{v}_d .
\end{equation}
Direction is therefore lost for the ordered pair $(o,d)$ only if $o$ and $d$ collide in the source encoder \emph{and} collide in the destination encoder, with reciprocal scale factors.

Two consequences are worth stating.
First, the degeneracy in which the two embeddings of a \emph{single} edge become proportional, $\mathbf{E}_o = c\,\mathbf{E}_d$, presupposes $H = W$, since otherwise the vectors lie in spaces of different dimension; this is not vacuous for Crane, which uses $H = W = 64$ (\S\ref{sec:setup}).
That degeneracy makes $A_{(o,d)}$ symmetric, but symmetry alone does not conflate the two directions: confusing $(o,d)$ with $(d,o)$ requires condition~\eqref{eq:dir-cond}, which constrains the embeddings of \emph{both} nodes under \emph{both} encoders.
Under the bounded-collision abstraction of Assumption~1, this edge-level direction-confusion event is conservatively bounded by $p$ per layer and by $p^{k}$ across $k$ layers (Theorem~\ref{thm:collision})---the same order as any pairwise edge collision, which the hierarchy resolves by the same mechanism.

Second, the use of separate source and destination encoders keeps the two constraints distinct.
Had Crane used a single shared encoder, $\mathbf{u}_x = \mathbf{v}_x$ for every $x$, the two conditions in Eq.~\eqref{eq:dir-cond} would collapse into the single condition $\mathbf{u}_d \propto \mathbf{u}_o$, and direction would be lost whenever $o$ and $d$ collided in that one encoder.
The design of \S\ref{sec:embedding-module} therefore makes direction loss strictly rarer, and no separate direction-preserving mechanism is required.

\section{Store and Query Pseudocode}
\label{app:pseudocode}

Algorithm~\ref{alg:crane-ops} lists the \textsc{Store} and \textsc{Query} operations of \S\ref{sec:store-module} and \S\ref{sec:query-module} in full, unifying the sequential carry ($|B|{=}1$), the conservative batch carry ($|B|{>}1$), and automatic memory expansion in a single procedure.
In practice the early exit at $\hat{y}^{(i)}_B < \theta$ branches on a device-resident scalar and forces a CPU--GPU synchronization at every layer, so a GPU implementation may drop the exit and sweep every allocated layer: a below-threshold layer issues a zero carry and the loop proceeds to the next layer, while every carry actually issued still satisfies the same threshold rule; we have verified that this relaxation does not measurably affect accuracy.

\begin{algorithm}[t]
\small
\caption{Store and Query Operations ($|B|{=}1$: sequential carry; $|B|{>}1$: conservative batch carry)}
\label{alg:crane-ops}
\begin{algorithmic}[1]
\Require Number of allocated layers $N$;
         memories $\{M^{(i)}\!\in\!\mathbb{R}^{H\times W}\}_{i=1}^N$;
         carry threshold $\theta>1$; max layers $N_{\max}$;
         decoder parameters $\mathbf{w}\!\in\!\mathbb{R}^N,\,b\!\in\!\mathbb{R}$.
\vspace{0.2em}
\Statex
\Function{Store}{$B = \{e_k = (o_k,d_k,w_k)\}_{k=1}^{|B|}$}
  \State $A^{(i)}_e \leftarrow \textsc{Embed}(e, i)$
         \textbf{for all} $e \in B,\ 1 \le i \le N$
  \State $A^{(i)}_B \leftarrow \sum_{e \in B} A^{(i)}_e$;\quad
         $M^{(1)} \leftarrow M^{(1)} + \sum_{e \in B} w_e\, A^{(1)}_e$
  \State $i \leftarrow 1$
  \While{$i \le N$}
    \State $\hat{y}^{(i)}_B \leftarrow
           \min_{p,q}\bigl(M^{(i)} \oslash A^{(i)}_B\bigr)_{pq}$
    \If{$\hat{y}^{(i)}_B < \theta$}
      \textbf{break}
    \EndIf
    \If{$i = N$ \textbf{and} $N = N_{\max}$}
      \textbf{break}
    \EndIf
    \If{$i = N$ \textbf{and} $N < N_{\max}$}
      \State $N \leftarrow N + 1$
             \Comment{Auto-expand}
      \State $M^{(N)} \leftarrow \mathbf{0}^{H\times W}$
      \State $A^{(N)}_B \leftarrow
             \sum_{e \in B} \textsc{Embed}(e, N)$
    \EndIf
    \State $M^{(i)} \leftarrow M^{(i)} - \hat{y}^{(i)}_B \cdot A^{(i)}_B$
    \State $M^{(i+1)} \leftarrow M^{(i+1)}
           + \frac{\hat{y}^{(i)}_B}{\theta} \cdot A^{(i+1)}_B$
    \State $i \leftarrow i + 1$
  \EndWhile
\EndFunction
\vspace{0.2em}
\Statex
\Function{Query}{$e=(o,d)$}
  \State $\{A^{(i)}\}_{i=1}^N \leftarrow
         \{\textsc{Embed}(e, i)\}_{i=1}^N$
  \For{$i \leftarrow 1$ \textbf{to} $N$}
    \State $q_i \leftarrow \min_{p,q}\bigl(M^{(i)} \oslash A^{(i)}\bigr)_{pq}$
  \EndFor
  \State $\mathbf{q} \leftarrow [q_1, \dots, q_N]^\top$
  \State \Return $\hat{y}(e) \leftarrow \mathbf{w}^\top \mathbf{q} + b$
\EndFunction
\end{algorithmic}
\end{algorithm}

\section{Evaluation}

\begin{figure*}[t]
	\centering
	\begin{minipage}{\linewidth}
		\centering
		\includegraphics[width=0.72\linewidth]{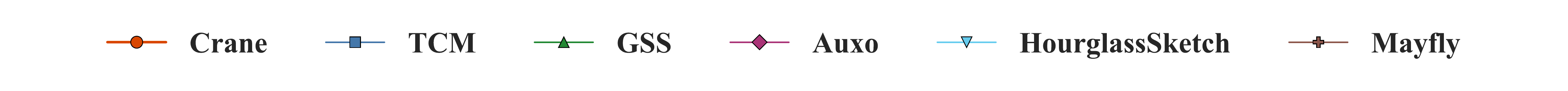}
		\vspace{-0.2cm}
	\end{minipage}
	\begin{minipage}{\linewidth}
		\captionsetup[subfigure]{skip=-0.02cm}
		\centering
		\subcaptionbox{\small Lkml ($\gamma{=}{-}0.5$)}{%
			\includegraphics[width=0.235\linewidth]{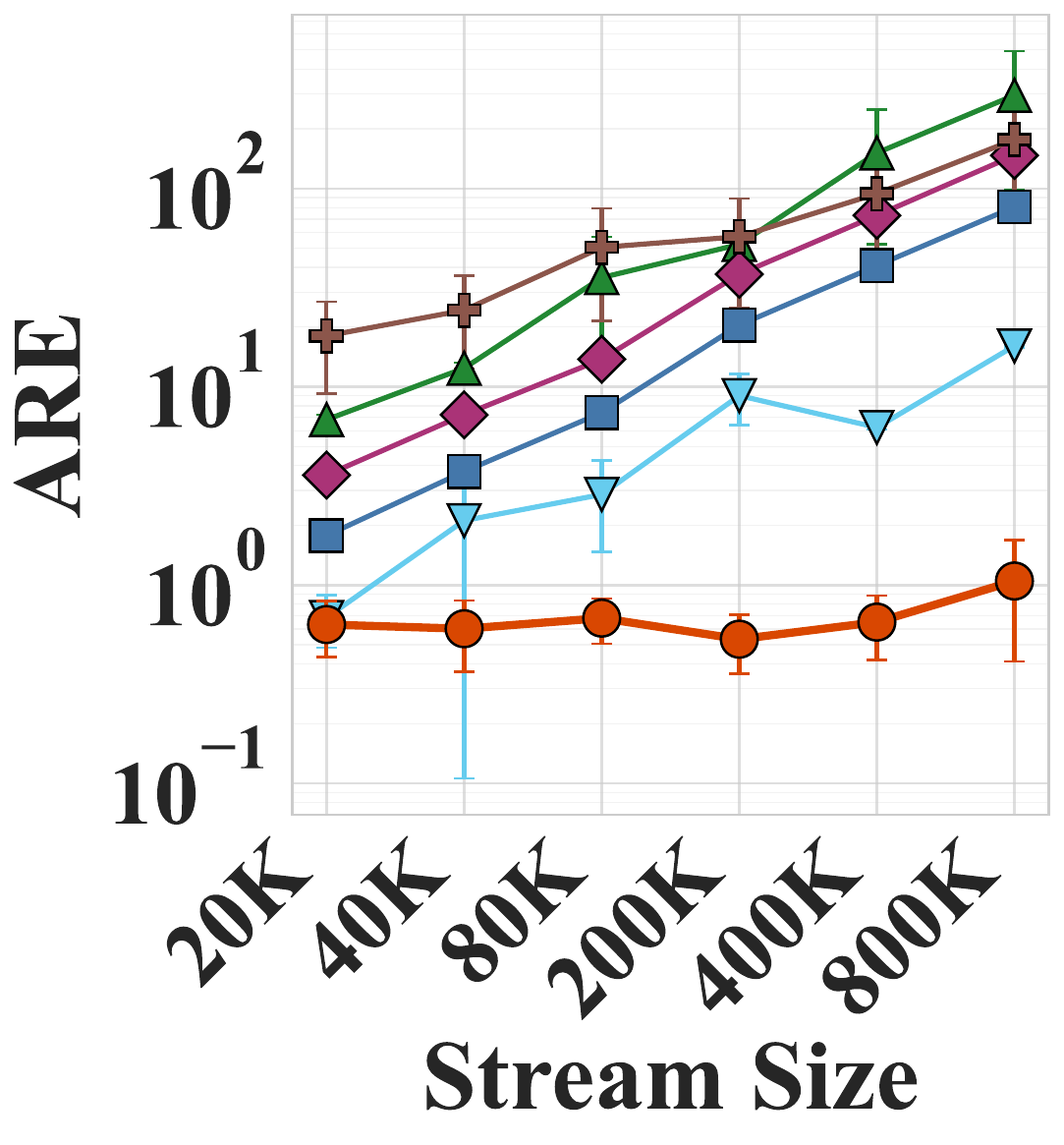}%
		}\hfill%
		\subcaptionbox{\small CAIDA ($\gamma{=}{-}0.5$)}{%
			\includegraphics[width=0.235\linewidth]{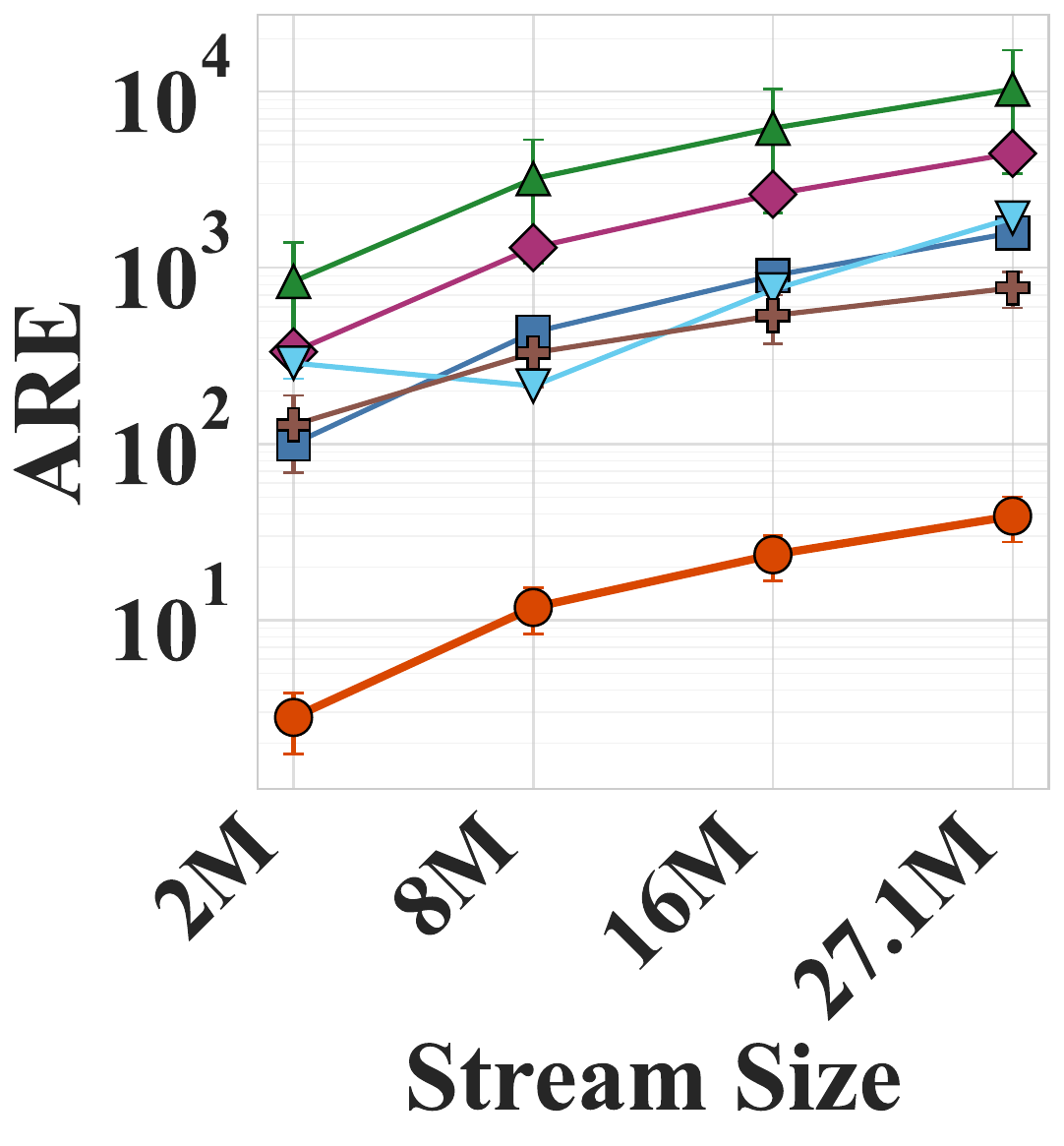}%
		}\hfill%
		\subcaptionbox{\small WikiTalk ($\gamma{=}{-}0.5$)}{%
			\includegraphics[width=0.235\linewidth]{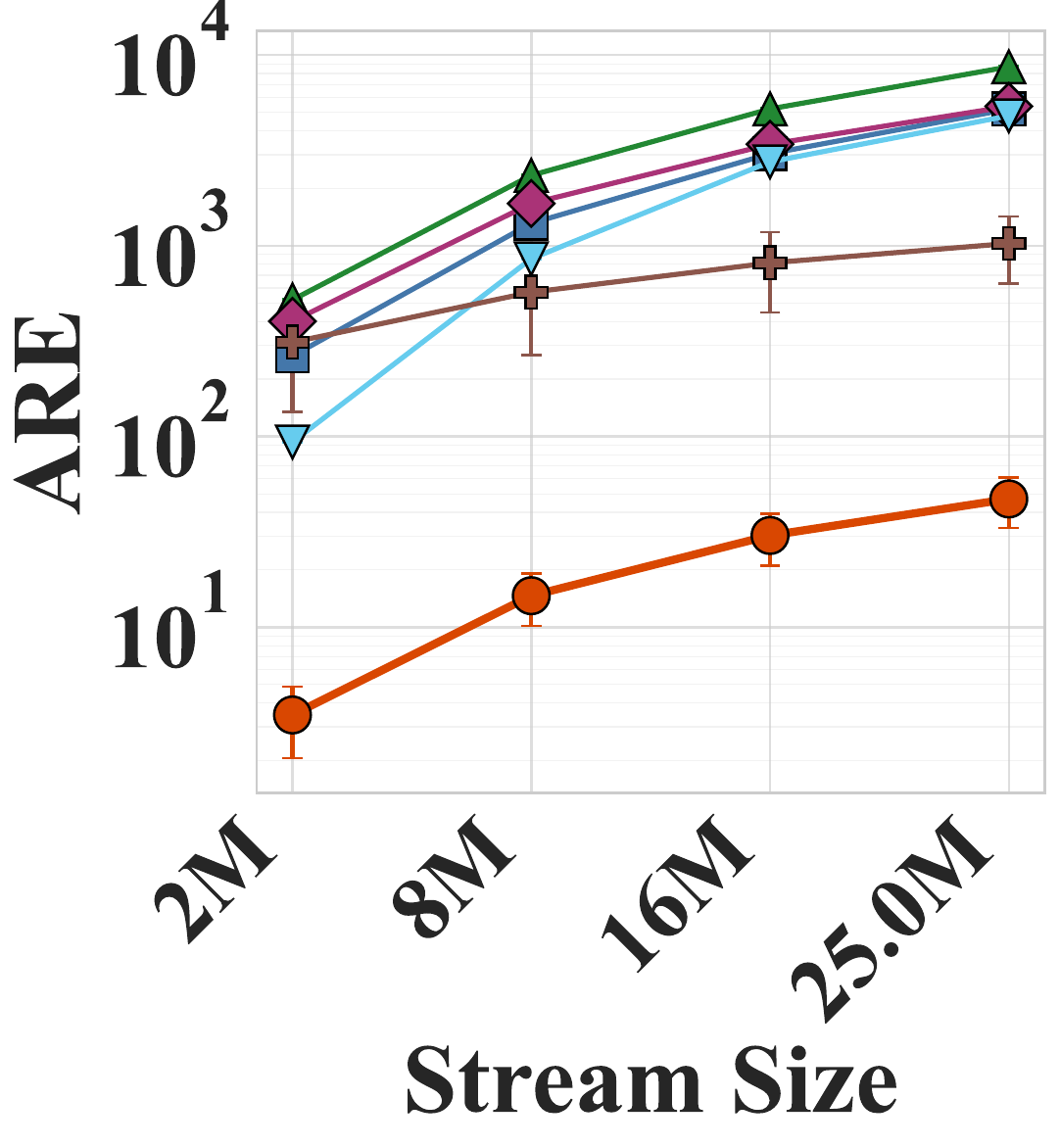}%
		}\hfill%
		\subcaptionbox{\small StackOverflow ($\gamma{=}{-}0.5$)}{%
			\includegraphics[width=0.235\linewidth]{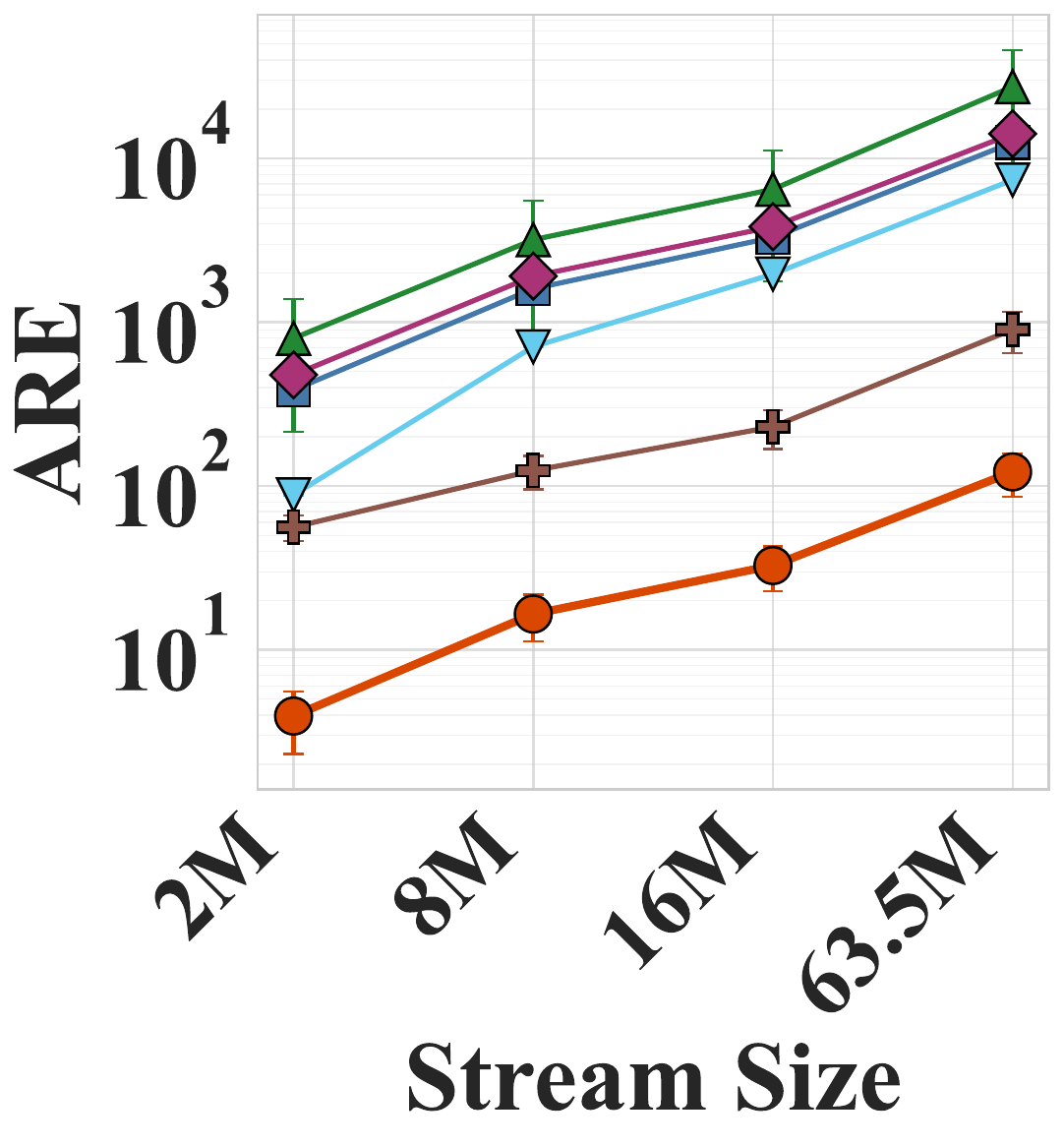}%
		}

		\vspace{0.15cm}

		\subcaptionbox{\small Lkml ($\gamma{=}0.5$)}{%
			\includegraphics[width=0.235\linewidth]{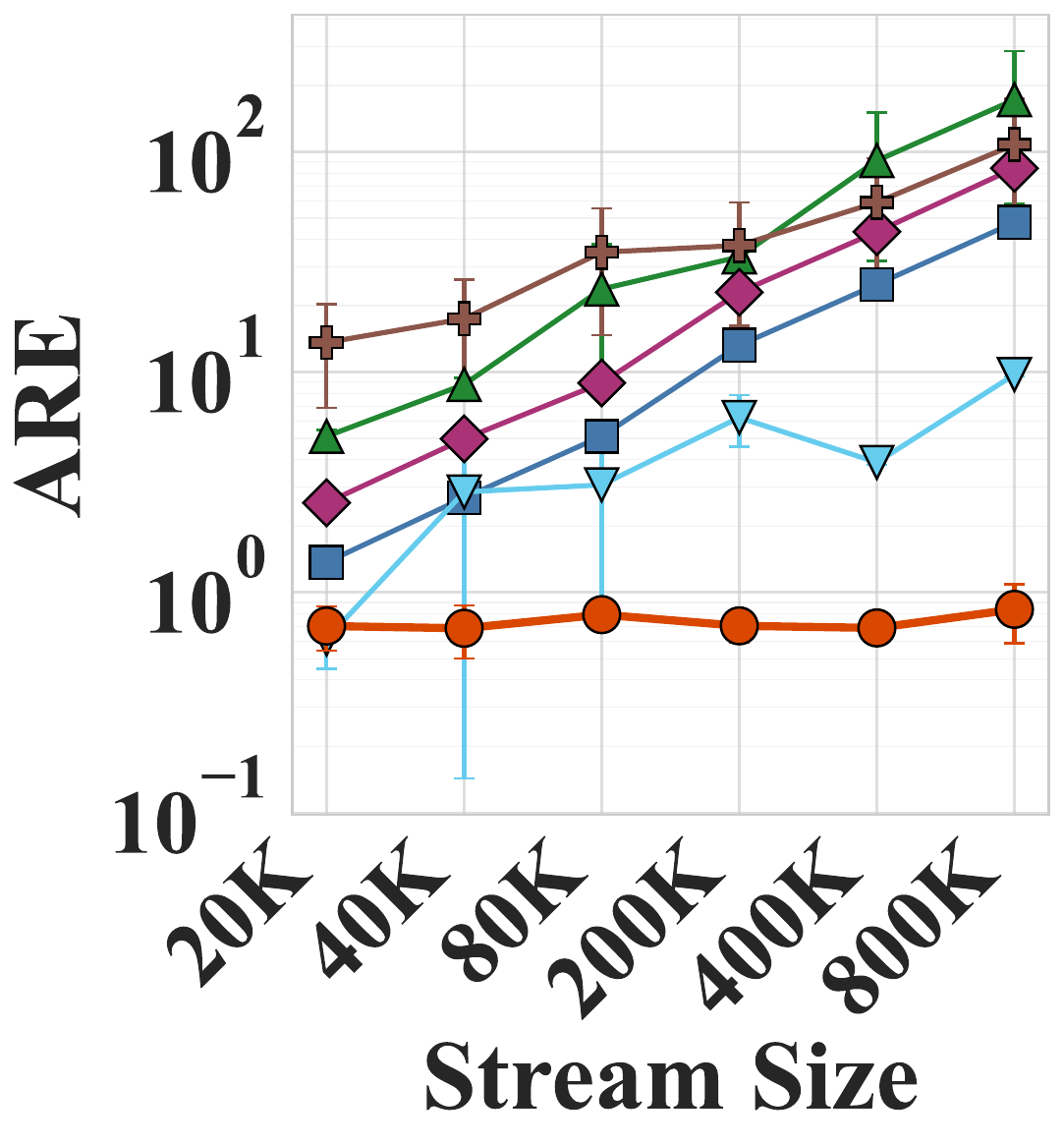}%
		}\hfill%
		\subcaptionbox{\small CAIDA ($\gamma{=}0.5$)}{%
			\includegraphics[width=0.235\linewidth]{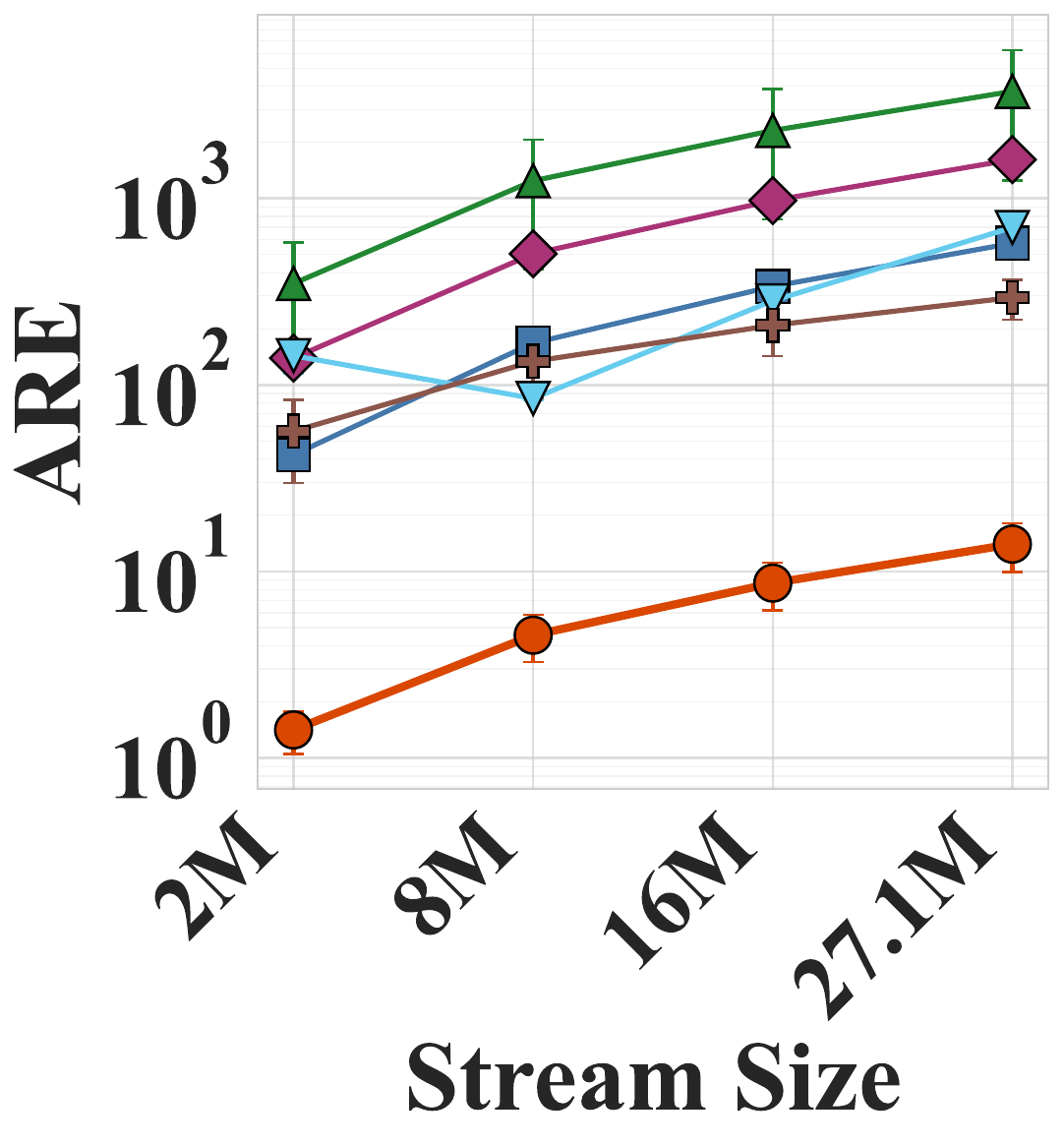}%
		}\hfill%
		\subcaptionbox{\small WikiTalk ($\gamma{=}0.5$)}{%
			\includegraphics[width=0.235\linewidth]{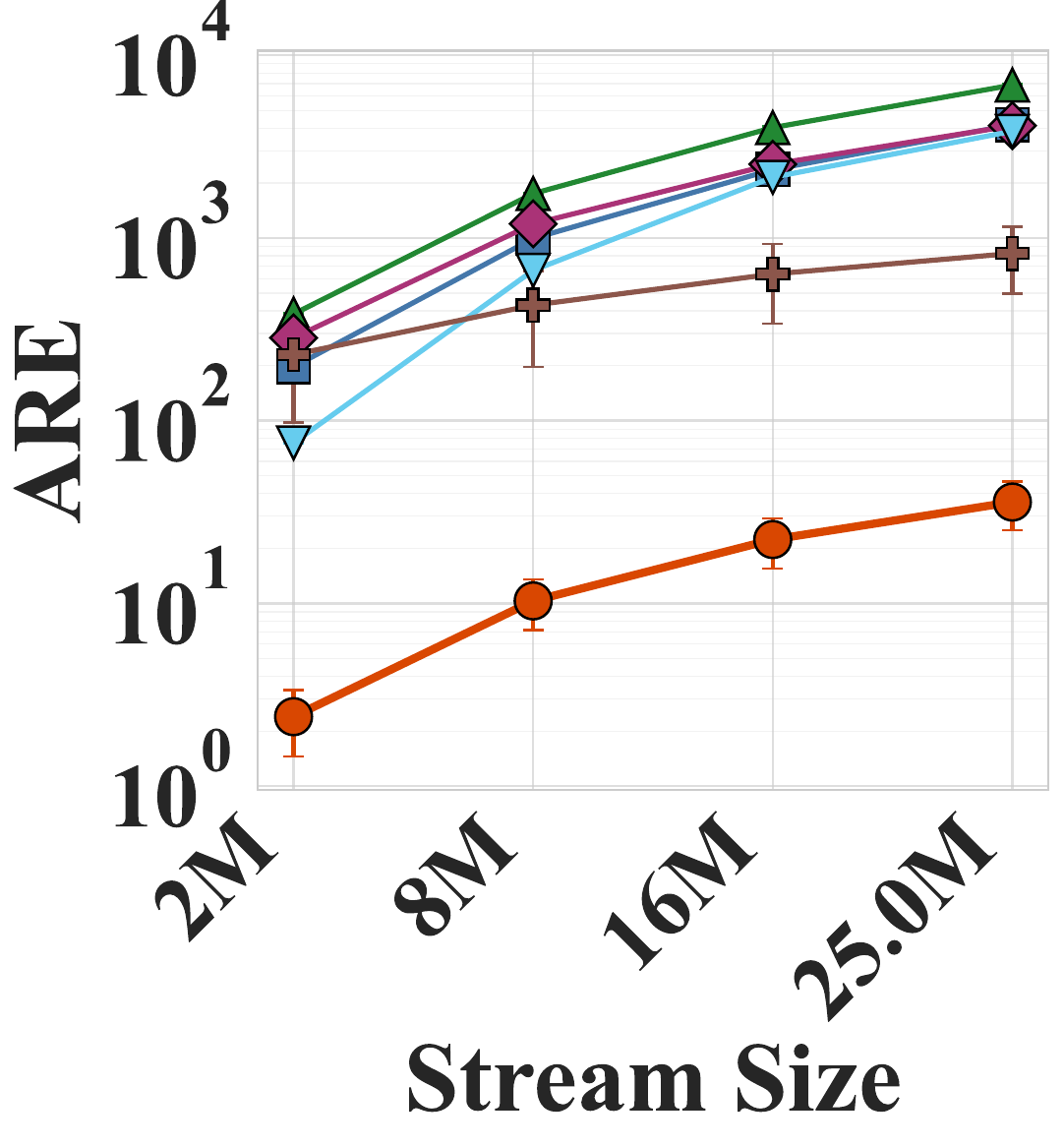}%
		}\hfill%
		\subcaptionbox{\small StackOverflow ($\gamma{=}0.5$)}{%
			\includegraphics[width=0.235\linewidth]{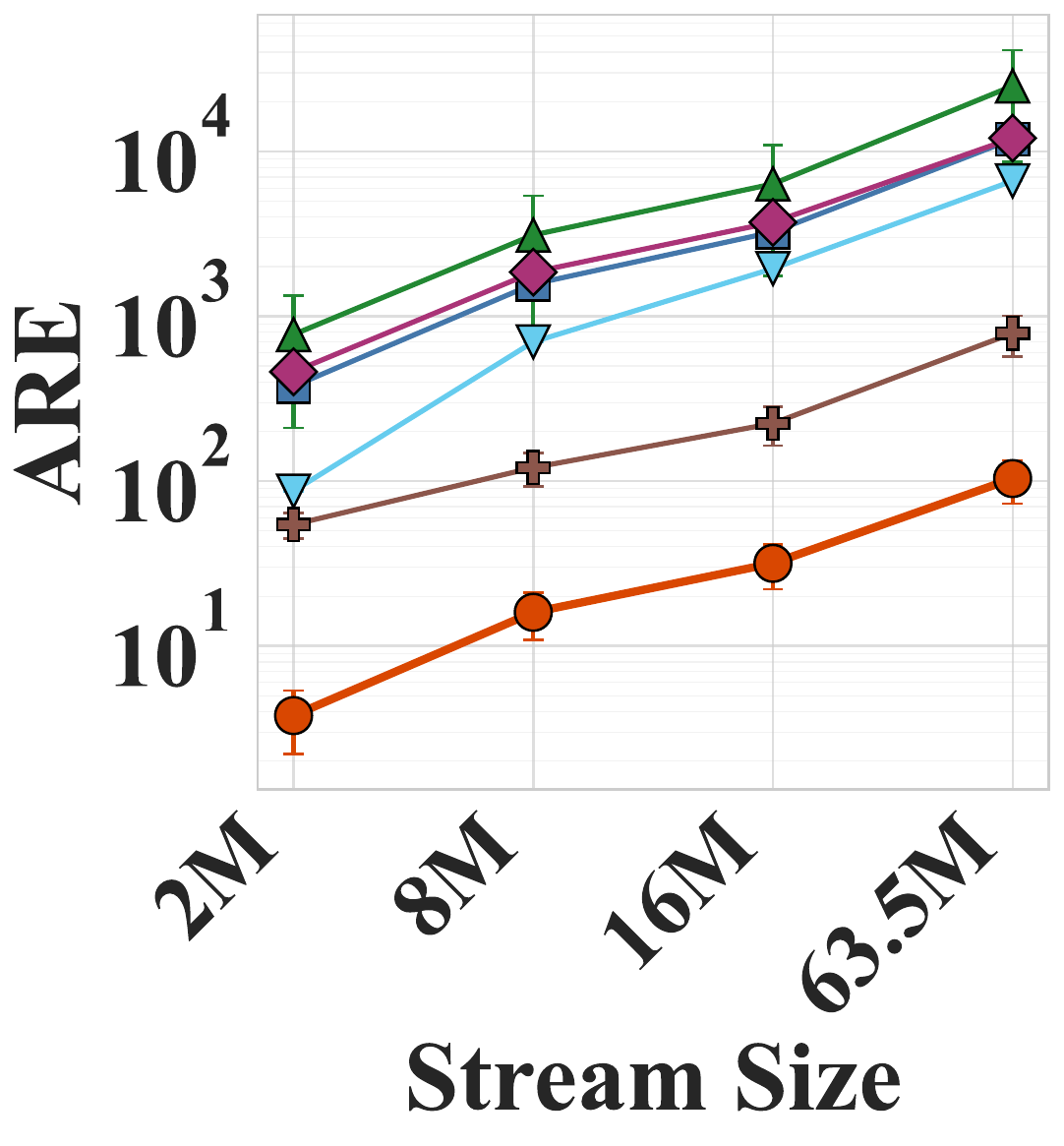}%
		}

		\vspace{0.15cm}

		\subcaptionbox{\small Lkml ($\gamma{=}1.0$)}{%
			\includegraphics[width=0.235\linewidth]{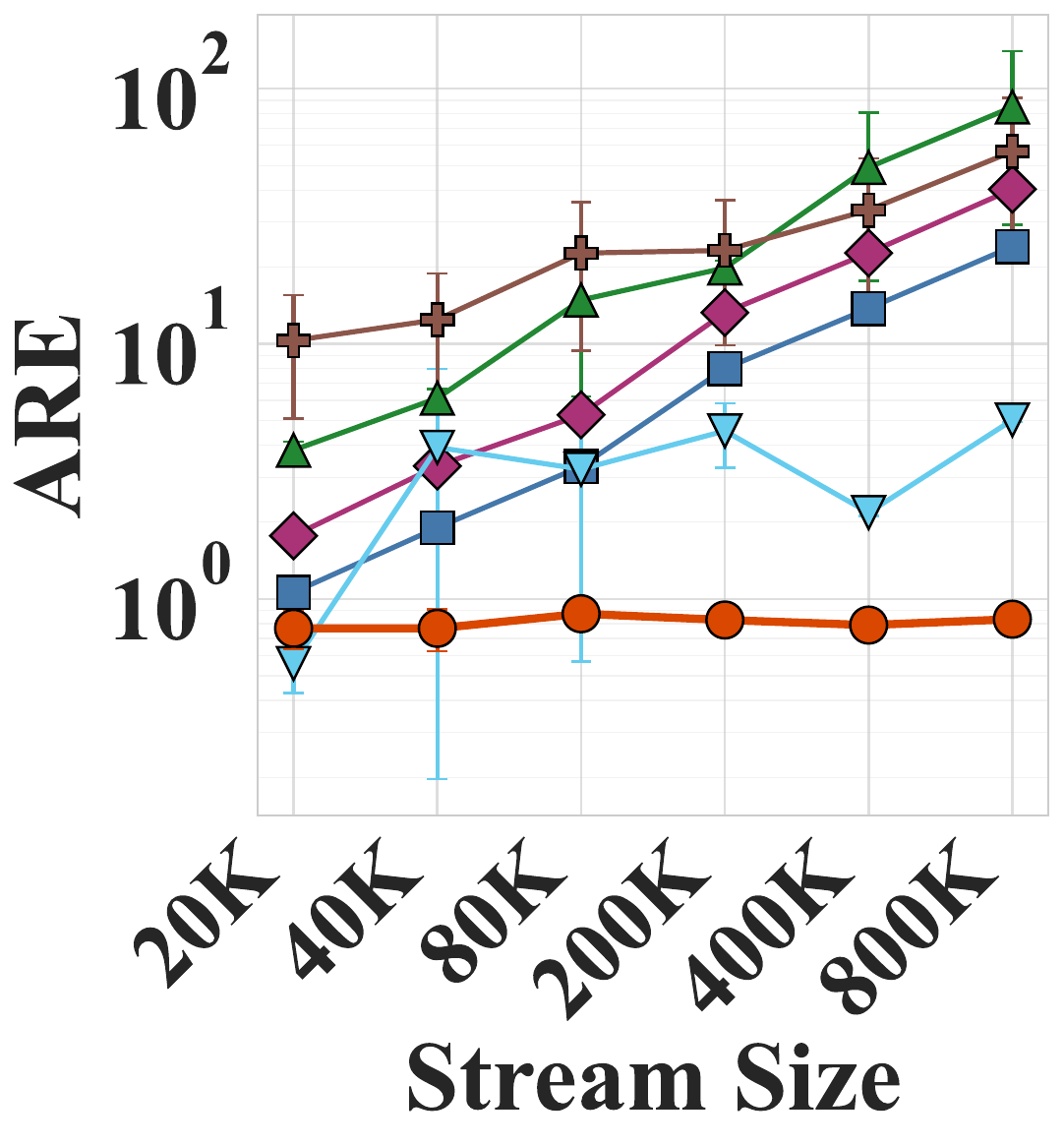}%
		}\hfill%
		\subcaptionbox{\small CAIDA ($\gamma{=}1.0$)}{%
			\includegraphics[width=0.235\linewidth]{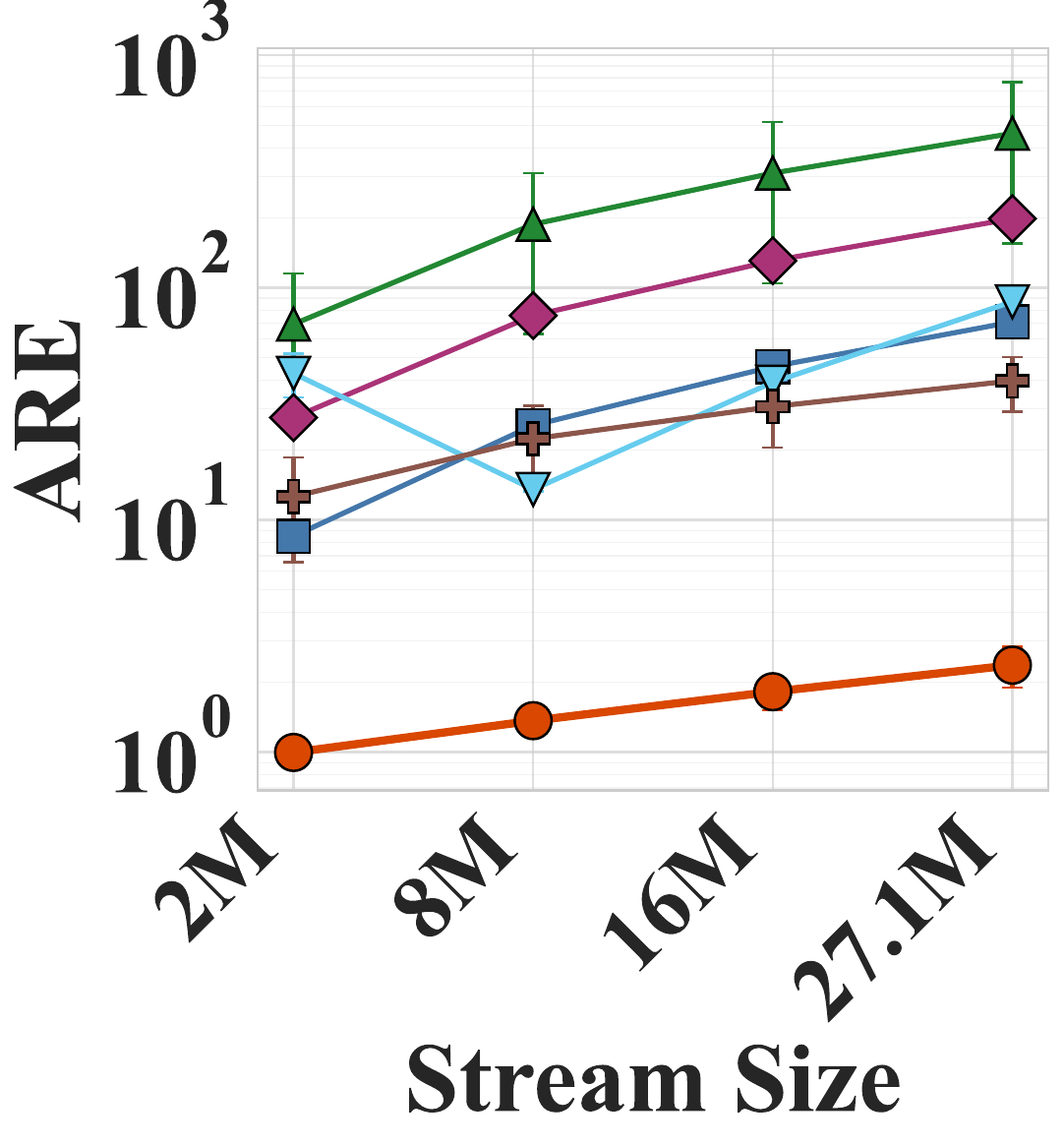}%
		}\hfill%
		\subcaptionbox{\small WikiTalk ($\gamma{=}1.0$)}{%
			\includegraphics[width=0.235\linewidth]{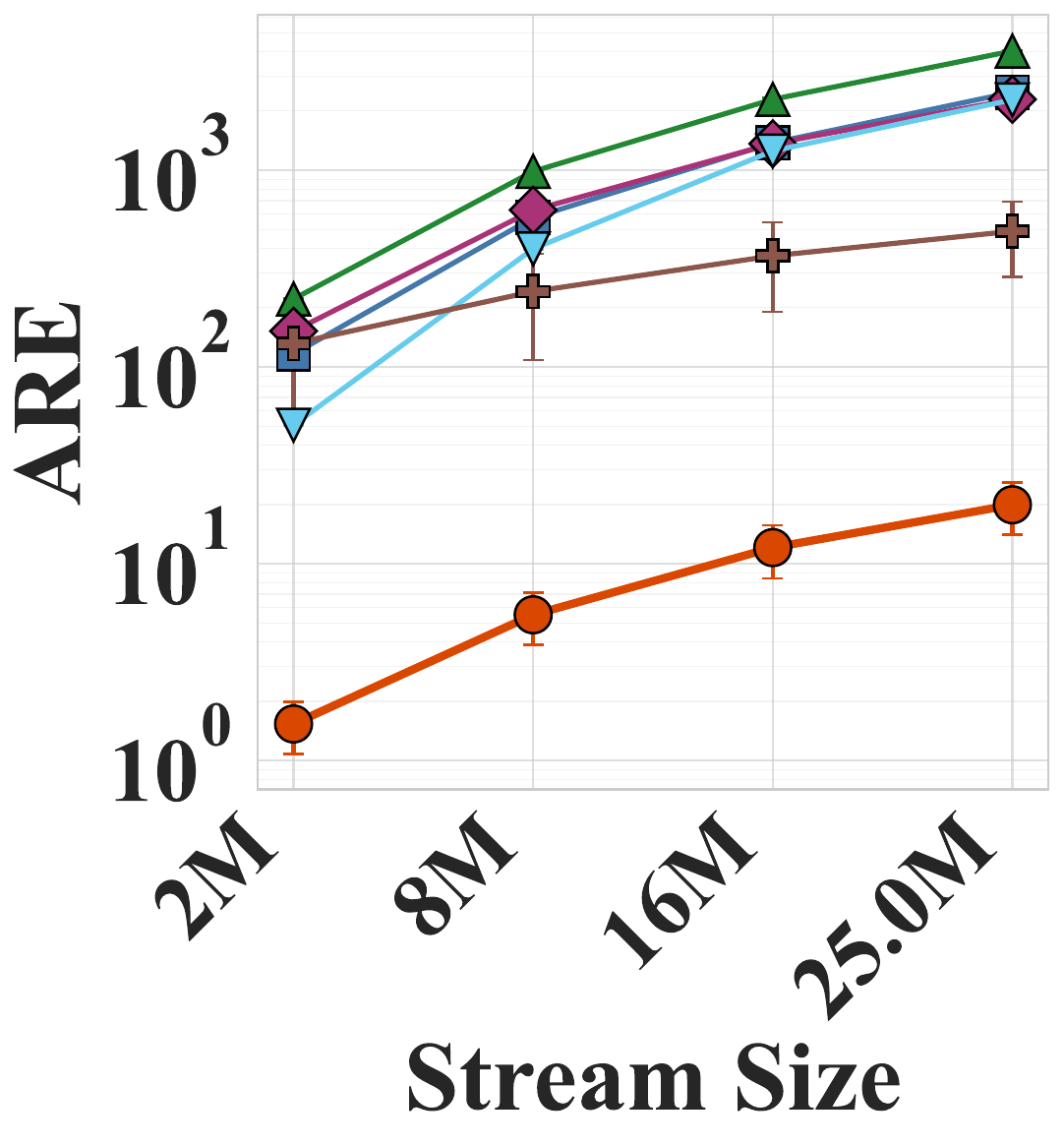}%
		}\hfill%
		\subcaptionbox{\small StackOverflow ($\gamma{=}1.0$)}{%
			\includegraphics[width=0.235\linewidth]{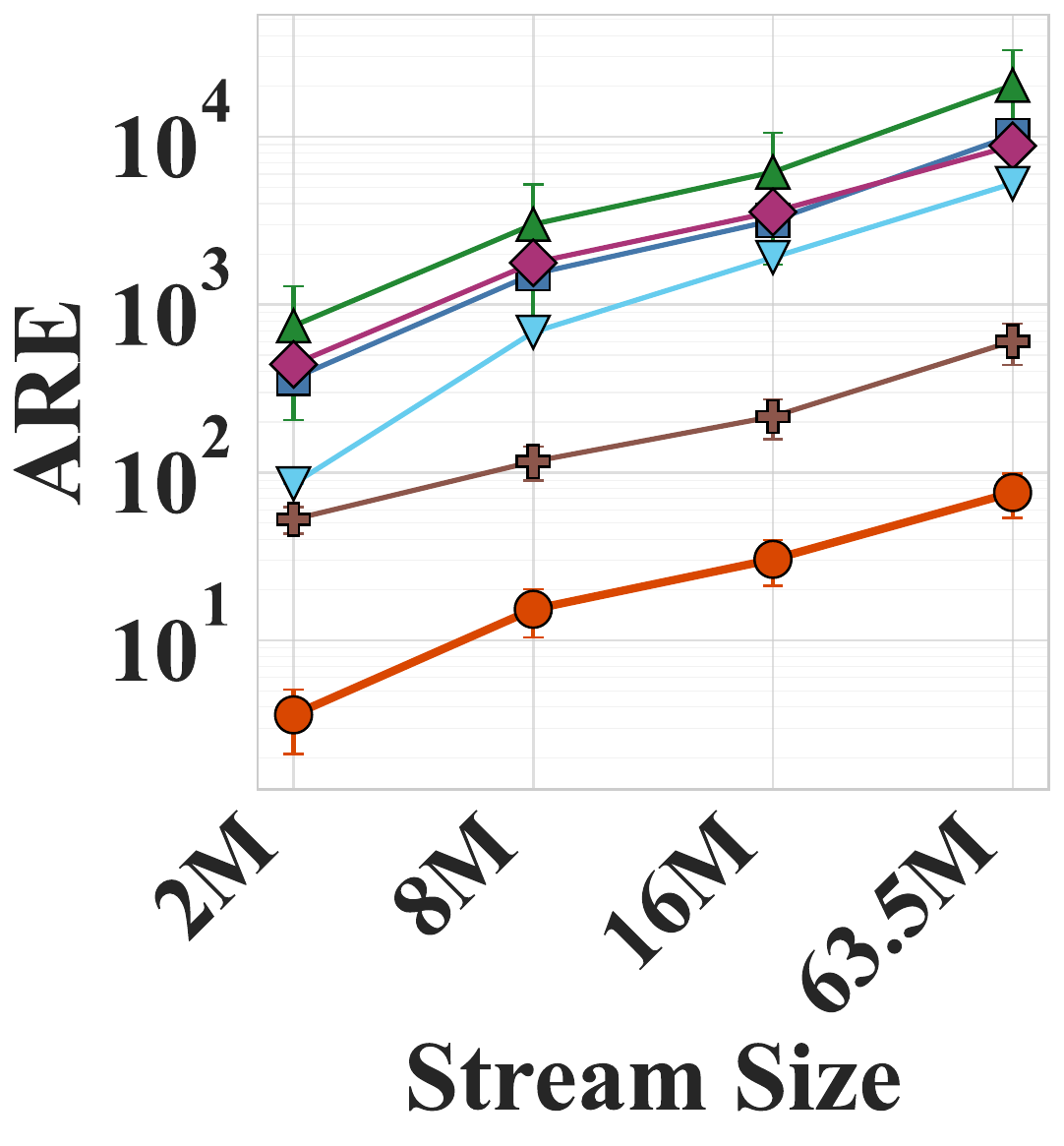}%
		}

		\vspace{-0.2cm}
		\caption{Query-weighted ARE under weight-biased workloads---tail-biased ($\gamma = -0.5$, top) and heavy-hitter-biased ($\gamma = 0.5$ and $1.0$, middle and bottom)---on the same stream segments and 64\,KB budget as Figure~\ref{fig:accuracy}.}
		\label{fig:weighted}
	\end{minipage}
	\vspace{-0.3cm}
\end{figure*}

\subsection{Query-Weighted Accuracy}
\label{app:weighted}

The evaluation in \S\ref{sec:accuracy} queries every distinct edge of a segment exactly once, so each edge contributes equally to ARE.
How often an edge is queried relative to its own cumulative weight is application-dependent and bidirectional: monitoring and traffic-engineering workloads are biased toward heavy hitters, whereas security workloads such as anomaly or superspreader detection predominantly probe the tail.
We parameterize query popularity as $q(e) \propto f(e)^{\gamma}$ and evaluate both correlation directions.
Setting $\gamma = 0$ recovers the equal-weight ARE of \S\ref{sec:setup} exactly---the standard protocol in graph-stream summarization~\cite{tang2016tcm,gou2022gss,jiang2023auxo,feng2023mayfly} and the neutral point of this family, presuming no correlation in either direction---while $\gamma > 0$ biases the workload toward heavy hitters and $\gamma < 0$ toward tail edges.
Figure~\ref{fig:weighted} reports $\gamma \in \{-0.5, 0.5, 1.0\}$, with all other settings and the trained checkpoints held identical to Figure~\ref{fig:accuracy}.
NotreDame is omitted because all of its edge weights are one, which makes every $\gamma$ produce the same query distribution.

Across the three settings Crane retains the lower ARE at nearly every stream length on every dataset.
The re-weighting itself costs Crane little accuracy: relative to the equal-weight setting, the tail-biased workload raises its ARE moderately, by about a third on CAIDA at worst, while the heavy-hitter-biased workloads mostly lower it, since the queried mass shifts onto heavy edges that carry proportionally less collision noise.
The advantage over the strongest baseline is similarly stable: the margins are not substantially eroded in either direction, and the tail-biased workload queries exactly the tail edges the hierarchy protects; the exceptions are Lkml's shortest segment, where Crane trails the strongest baseline by $10\%$ at $\gamma = 0.5$ and falls behind it at $\gamma = 1.0$, and the CAIDA margin at $\gamma = 1.0$, which narrows as the heaviest edges come to dominate the query mass.

Because the weighting is applied at query time only, these results also speak to the query-space generality of the trained model: the checkpoints of Figure~\ref{fig:accuracy} answer the re-weighted workloads with no retraining and an essentially unchanged ranking of methods.
For deployments with extreme query bias, two remedies are available: adjusting the training objective toward the workload, sampling training queries by the expected popularity $q(e)$ so that heavy hitters, or tail edges, are queried at higher rates during training; or reshaping the hierarchy, raising the carry threshold $\theta$ or the layer budget $N_{\max}$ so that the heaviest edges do not saturate the top layer, preserving the separation the hierarchy provides (see the $\theta$ guidance of \S\ref{sec:store-module} and the sensitivity study of \S\ref{sec:ablation}).

\subsection{Contribution of the Training Objective}
\label{app:loss}

\begin{table}[h]
\centering
\vspace{-0.1cm}
\caption{Loss ablation, averaged over the stream-length subsets of Figure~\ref{fig:accuracy}: Crane's MAE objective against Mayfly's balanced MSE objective, architecture unchanged; mean $\pm$ s.d.\ over random seeds.}
\vspace{-0.3cm}
\label{tab:loss}
\resizebox{\linewidth}{!}{
\begin{tabular}{llccccc}
\toprule
 & Objective & Lkml & NotreDame & CAIDA & WiKiTalk & StackOverflow \\
\midrule
\multirow{2}{*}{ARE} & MAE (default) & 0.72 $\pm$ 0.02 & 1.53 $\pm$ 0.72 & 14.40 $\pm$ 4.29 & 21.69 $\pm$ 6.63 & 42.20 $\pm$ 12.66 \\
 & Balanced MSE & 0.73 $\pm$ 0.04 & 2.17 $\pm$ 0.67 & 19.07 $\pm$ 2.11 & 28.95 $\pm$ 3.33 & 55.98 $\pm$ 6.76 \\
\midrule
\multirow{2}{*}{AAE} & MAE (default) & 2.82 $\pm$ 0.14 & 1.53 $\pm$ 0.72 & 43.62 $\pm$ 7.18 & 27.41 $\pm$ 8.19 & 46.95 $\pm$ 14.09 \\
 & Balanced MSE & 2.77 $\pm$ 0.09 & 2.17 $\pm$ 0.67 & 51.48 $\pm$ 3.64 & 36.38 $\pm$ 4.15 & 62.29 $\pm$ 7.52 \\
\bottomrule
\end{tabular}}
\vspace{-0.25cm}
\end{table}

To separate the contribution of the training objective from that of the architecture, we retrain Crane with Mayfly's balanced mean squared error, in which learned scalars weight an absolute and a relative term, and change nothing else.
As Table~\ref{tab:loss} shows, swapping the objective keeps the mean ARE on every dataset within the range the MAE default itself spans across its five training seeds, whereas collapsing the hierarchy to a single layer costs roughly $7\times$ (Table~\ref{tab:ablation_combined}).
Both variants remain one to two orders of magnitude below the strongest baseline of Figure~\ref{fig:accuracy}, so Crane's accuracy is attributable to the hierarchical memory rather than to the choice of loss; we retain MAE because it is simpler and introduces no additional balancing terms.

\subsection{Node Flow Results by Direction}
\label{app:workloads}

\begin{table}[h]
\centering
\vspace{-0.1cm}
\caption{The node flow workloads by direction: mean $\pm$ s.d.\ over random seeds; each direction is answered within its own 64\,KB sketch, and \S\ref{subsec:casestudy} reports the mean of the two directions.}
\vspace{-0.3cm}
\label{tab:nodeflow_directions}
\resizebox{\linewidth}{!}{
\begin{tabular}{llcccc}
\toprule
 & Method & D1 Out & D1 In & D2 Out & D2 In \\
\midrule
\multirow{3}{*}{ARE} & Crane & 2.21 $\pm$ 0.23 & 1.85 $\pm$ 0.06 & 1.08 $\pm$ 0.28 & 0.88 $\pm$ 0.16 \\
 & Mayfly & 8.89 $\pm$ 0.72 & 3.31 $\pm$ 0.07 & 6.15 $\pm$ 2.11 & 6.67 $\pm$ 2.78 \\
 & HourglassSketch & 1.99 $\pm$ 0.08 & 1.97 $\pm$ 0.05 & 1.86 $\pm$ 0.14 & 2.12 $\pm$ 0.17 \\
\midrule
\multirow{3}{*}{nDCG@100} & Crane & 0.752 $\pm$ 0.005 & 0.819 $\pm$ 0.001 & 0.948 $\pm$ 0.003 & 0.947 $\pm$ 0.022 \\
 & Mayfly & 0.002 $\pm$ 0.000 & 0.002 $\pm$ 0.000 & 0.003 $\pm$ 0.001 & 0.003 $\pm$ 0.004 \\
 & HourglassSketch & 0.397 $\pm$ 0.019 & 0.390 $\pm$ 0.051 & 0.879 $\pm$ 0.010 & 0.893 $\pm$ 0.017 \\
\midrule
\multirow{3}{*}{Recall@100} & Crane & 0.542 $\pm$ 0.004 & 0.686 $\pm$ 0.005 & 0.680 $\pm$ 0.010 & 0.736 $\pm$ 0.055 \\
 & Mayfly & 0.000 $\pm$ 0.000 & 0.000 $\pm$ 0.000 & 0.004 $\pm$ 0.005 & 0.002 $\pm$ 0.004 \\
 & HourglassSketch & 0.126 $\pm$ 0.013 & 0.106 $\pm$ 0.023 & 0.420 $\pm$ 0.012 & 0.546 $\pm$ 0.059 \\
\bottomrule
\end{tabular}}
\vspace{-0.25cm}
\end{table}

Table~\ref{tab:nodeflow_directions} reports the two node flow workloads by direction. The results are largely consistent across the two directions.

\section{Extensions}
\label{sec:extensions}

Three extensions follow naturally from Crane's modular design.
First, the Store and Query modules interact with the Embedding module only through the basis matrices, so the MLP encoders of \S\ref{sec:embedding-module} can be exchanged for encoder pairs satisfying the collision bound of Assumption~1; in particular, semantic-aware or user-defined encoders that consume node attributes rather than numerical IDs can be adopted without changing either module, allowing semantically related nodes to share representations.
Second, the synthetic training generator (\S\ref{sec:training}) acts as a workload prior along two axes: where a deployment's weight distribution is observable, the generator can be shaped to match it, with the cross-training results of \S\ref{sec:robustness} as guidance---err toward a pool heavier-tailed than the expected workload; and where its query distribution is biased, the training queries can be sampled by the expected query popularity, as in Appendix~\ref{app:weighted}.
Third, the carry mechanism targets the cash-register model of \S\ref{sec:definition}; extending it to signed (turnstile) updates, in which weight may also be removed, remains future work.